\documentclass[aps,tightenlines,floatfix,longbibliography,superscriptaddress,nofootinbib]{revtex4-1}
\usepackage{graphicx}
\usepackage{amssymb,amsthm}

\usepackage{amsmath}
\usepackage{amsfonts}
\usepackage{dsfont}
\usepackage{color}
\usepackage{braket}
\usepackage{float}
\usepackage{standalone}
\usepackage{mdframed}

\usepackage[utf8]{inputenc}
\usepackage[english]{babel}

\usepackage[bottom]{footmisc}
\usepackage{amsmath}
\usepackage{amsfonts}
\usepackage{amsthm}
\usepackage{mathrsfs}
\usepackage{algorithm}
\usepackage[noend]{algpseudocode}
\usepackage{caption}
\usepackage{subcaption}
\usepackage{color}
\usepackage{braket}
\usepackage{yfonts}
\usepackage[inline]{enumitem}
\usepackage{centernot}

\newtheorem{theorem}{Theorem}[section]

\newtheorem{lemma}[theorem]{Lemma}
\newtheorem{proposition}[theorem]{Proposition} 
\newtheorem{corollary}{Corollary}[theorem]
\newtheorem{definition}{Definition}[section]

\newenvironment{remark}{\par\noindent{\bf Remark.\ } \it}

\DeclareMathAlphabet{\mathpzc}{OT1}{pzc}{m}{it}

\newcommand{\mbf}{\mathbf}
\newcommand{\mbb}{\mathbb}
\newcommand{\bs}{\boldsymbol}
\newcommand{\mcr}{\mathcal}

\newcommand{\Norm}{\Big|\Big|}

\newcommand{\cond}{\,\lvert\,}
\newcommand{\lcond}{\,\,\big\lvert\,\,}

\newcommand{\norm}[1]{\left\lVert#1\right\rVert}

\DeclareMathOperator{\sign}{sgn}
\DeclareMathOperator{\proj}{proj}
\DeclareMathOperator{\argmin}{argmin}
\DeclareMathOperator{\ext}{ext}
\DeclareMathOperator{\interior}{int}

\begin{document}

\title{Efficient Solution of Boolean Satisfiability Problems with Digital MemComputing}

\author{Sean R.B. Bearden}
\email{email: sbearden@ucsd.edu}
\affiliation{Department of Physics, University of California, San Diego, La Jolla, CA 92093}

\author{Yan Ru Pei}
\email{email: yrpei@ucsd.edu}
\affiliation{Department of Physics, University of California, San Diego, La Jolla, CA 92093}

\author{Massimiliano Di Ventra}
\email{email: diventra@physics.ucsd.edu}
\affiliation{Department of Physics, University of California, San Diego, La Jolla, CA 92093}

\maketitle

{\bf 
	Boolean satisfiability is a propositional logic problem of interest in multiple fields, e.g., physics, mathematics, and computer science. Beyond a field of research, instances of the SAT problem, as it is known, require efficient solution methods in a variety of applications.
	It is the decision problem of determining whether a Boolean formula has a satisfying assignment, believed to require exponentially growing time for an algorithm to solve for the worst-case instances.
	Yet, the efficient solution of many classes of Boolean formulae eludes even the most successful algorithms, not only for the worst-case scenarios, but also for typical-case instances.
	Here, we introduce a memory-assisted physical system (a digital memcomputing machine) that, when its non-linear ordinary differential equations are
	 integrated numerically, shows evidence for polynomially-bounded scalability while solving ``hard'' planted-solution instances of SAT, known to require exponential time to solve in the typical case for both complete and incomplete algorithms. 
	Furthermore, we analytically 
	demonstrate that the physical system can efficiently solve the SAT
	problem in {\it continuous} time, without the need to introduce chaos or an exponentially growing energy.
	The efficiency of the simulations is related to the collective dynamical properties of the original physical system that persist in the numerical integration to robustly guide the solution search even in the presence of numerical errors.
	We anticipate our results to broaden research directions in physics-inspired computing paradigms 
	 ranging from theory to application, from simulation to hardware implementation.
}

\newpage

The Boolean satisfiability problem~\cite{petke2015bridging} (SAT) is an important decision problem solved by determining if a solution exists to a Boolean formula.
A SAT instance is satisfiable when there exists an assignment of Boolean variables (each either TRUE or FALSE) that results in the Boolean formula returning TRUE.
Apart from its academic interest, the solution of SAT instances is required in a wide range of practical applications, including, travel, logistics,  software/hardware design, etc.~\cite{practicalSAT}.\\

The SAT problem has been studied for decades, and has an important role in the history of computational complexity.
Computer scientists, while categorizing the efficiency of algorithms, defined the NP class for difficult decision problems~\cite{cook1971complexity,complexity_bible}.
Some are known as {\it intractable} problems, meaning they are ``hard'' in the sense that all known algorithms cannot be bounded in polynomial time when determining if a solution exists in the worst-case scenario.
The SAT problem was the first to be shown to belong to the class of NP-complete problems~\cite{cook1971complexity}, implying that any decision problem in NP is reducible to a SAT problem in polynomial time.  
There are no known polynomial time algorithms for solving an NP-complete problem, though there are exponential time algorithms that are efficient for special cases of problem structure~\cite{complexity_bible}.
There is a ``widespread belief''~\cite{complexity_bible} that creation of a polynomial time algorithm is impossible, but this belief does not limit the realization of a polynomial {\it continuous-time} physical system.\\

NP-completeness is not exclusive to SAT, with hundreds of other NP-complete problems ranging from those of academic interest (graph theory, algebra and number theory, mathematical programming) to industry application (network design, data storage and retrieval, program optimization)~\cite{complexity_bible}.
If a polynomial time algorithm can solve any NP-complete problem class, then all NP problems can be computed efficiently.
The 3-SAT problem is NP-complete and a special case of SAT~\cite{complexity_bible}.
Randomly-generated 3-SAT instances are known to be difficult to many solution methods because they lack an exploitable problem structure.
For instance, one lauded algorithm, survey inspired decimation (SID), performs well on large instances of uniform random 3-SAT in the ``hard regime''~\cite{mezard2002sid}, but performs poorly in what is known as the ``easy regime''~\cite{parisi2003remarks}. We focus on the 3-SAT problem in the following due to it being a subclass of SAT with a consistent formulaic representation (three literals per clause).\\

\noindent {\bf Physics-inspired approach to computing}

A research direction that has been far less explored concerns the solution of SAT using non-quantum {\it dynamical systems}~\cite{HavaFish,zoltan2011,LPNN,DMM2}. The idea behind this approach is that the solutions of the SAT instance are mapped into the equilibrium points of a dynamical system. If the initial conditions of the dynamics belong to the basin of attraction of the equilibrium points, then the dynamical system will have to ``fall'' into these points. 
The approach is {\it fundamentally} different from the standard algorithms because dynamical systems perform computation in {\it continuous time}. Numerical simulation of continuous-time physical systems, an algorithm, requires the discretization of time to integrate the ordinary differential equations (ODEs) 
representing the physical system. As such, the dynamical-systems approach is ideally suited for a hardware implementation.\\

The authors of Ref.~\onlinecite{zoltan2011} have shown that an appropriately designed dynamical system can find the solutions of hard 3-SAT instances in continuous polynomial time, however, at a cost of {\it exponential} energy fluctuations. 
The reason for this exponential energy cost can be traced to the transient {\it chaotic} dynamics of the dynamical systems proposed in Ref.~\onlinecite{zoltan2011}.
As the problem size grows, the chaotic behavior translates into an exponentially increasing number of integration steps required to find the equilibrium points of the corresponding ODEs. \\

\begin{figure} [t]
	\centering
	\includegraphics[width=0.38\textwidth]{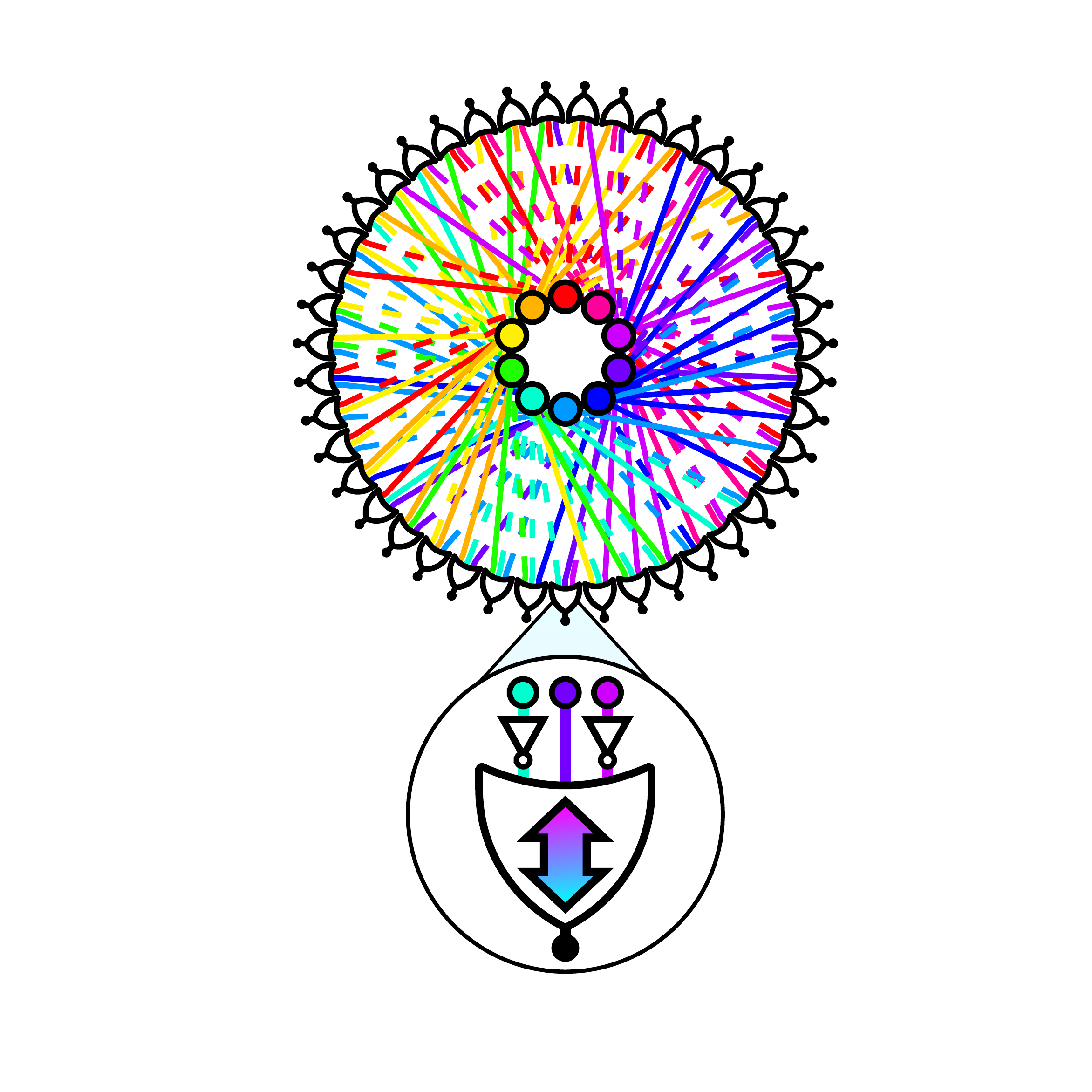}
	\caption{{\bf Schematic of a self-organizing logic circuit representing a 3-SAT instance.} The circuit is created from the constraints of a 3-SAT formula consisting of $N=10$ variables, and $M=43$ clauses. The formula is converted into $10$ voltage nodes (inner nodes) and $43$ self-organizing OR gates~\cite{bearden2018}. The black nodes traditionally associated with the output of the OR gates are fixed to TRUE to enforce the constraints. Dashed lines in the circuit represent NOT gates on the OR gate terminals.
		Ignoring the black nodes, the circuit can be interpreted as a factor graph with the gates becoming function nodes (see also Fig. \ref{fig:factor_graph}). The clause represented by the highlighted self-organizing OR gate is $(\bar{y}_i \vee y_j \vee \bar{y}_k)$, where NOT gates invert the polarity of the voltages. The double-headed arrow indicates this is a self-organzing logic gate with no distinction between an input and an output (terminal agnosticism).  The circular representation of the linear circuit is a reminder that the ordering of gates is irrelevant to the solution search.}
	\label{fig:circle_circuit}
\end{figure}

\noindent{\bf The digital memcomputing approach}

In recent years, a different physics-inspired computational paradigm has been introduced, known as {\it digital memcomputing}~\cite{DMM2,DMMperspective}. 
Digital memcomputing machines (DMMs) are non-linear dynamical systems specifically designed to solve constraint satisfaction problems, {\it e.g.},  3-SAT, with the assistance of memory~\cite{DMM2} (Fig. \ref{fig:circle_circuit}). The only equilibrium point(s) of the DMM is the solution(s) of the original problem. However, unlike previous work, DMMs are designed so that they have no other equilibrium points; see Sec. VI.D of the supplementary material (SM). Additionally, the dynamics will never enter a periodic orbit or a state of chaos~\cite{no-chaos} (see Sec. IX of SM). \\

The ability of continuous time dynamics to perform the solution search without resorting to chaotic dynamics results in efficient {\it simulations} (an algorithmic implementation) of DMMs using computationally-inexpensive integration schemes and modern computers. 
In addition, it was shown that DMMs find the solution of a given problem by employing topological objects, known as instantons, that connect critical points of increasing stability in the phase space~\cite{topo,topo1} (see Sec. XI of SM). 
Simulations found the DMMs then self-tune into a critical ({\it collective}) state which persists for the whole transient dynamics until a solution is found~\cite{bearden2019critical}. It is this critical branching behavior that allows DMMs to explore {\it collective} updates of variables during the solution search, without the need to check an exponentially-growing number of states. This is in contrast to local-search algorithms which are characterized by a ``small'' (not collective) number of variable updates at each step of the computation~\cite{NOAP}. \\

Here, we introduce a physical DMM to find solutions of the 3-SAT problem. 
So as to facilitate the reading of our paper, we have contained the mathematical description of our physical DMM to Box 1.  
We then perform numerical simulations of the ODEs (discretized time) of the DMM to solve random 3-SAT instances with planted solutions.
These instances are generated with a clause distribution control (CDC) procedure, known to require exponentially growing time to solve in the typical case for both 
complete and incomplete algorithms~\cite{barthel2002}. The CDC instances have found use as benchmarks in recent years of SAT competitions  (satcompetition.org)~\cite{cdc2016,satcomp2017benchmarks,satcomp2018benchmarks}. The simulations have been performed using a forward-Euler integration scheme~\cite{sauer2006numerical} with an adaptive time step, implemented in MATLAB R2019b with each solution attempt run on a single logic core of an AMD EPYC 7401 server (see also Sec. II of SM).\\

We compare our results with those obtained from two well-known algorithms: WalkSAT, a stochastic local-search procedure~\cite{selman1993walksat}, and survey-inspired decimation (SID), a message-passing procedure utilizing the cavity method from statistical physics~\cite{mezard2002sid}. (in Sec. II of the SM we also compare with the winner of a recent SAT competition and AnalogSAT~\cite{zoltangpu}). Comparison is achieved via scalability of some indicator vs. the problem size. 
As expected, both algorithms show an exponential scaling for the CDC instances (Fig. \ref{fig:varying_ratio}). 
Our simulations instead show a power-law scalability of integration steps ($\sim\!N^{a}$) for typical cases, where the typical case is inferred from the median
number of integration steps.\\

Finally, we show that the dynamics is capable of finding satisfying variable assignments for 3-SAT in polynomially-bounded (linear or sub-linear) {\it continuous} time without the need of an exponentially increasing energy cost demonstrated via certain dissipative and topological properties of the system (see Secs. X-XI of SM). \\

While the reported numerical and analytical results do not resolve the famous P vs. NP debate,
(which, incidentally, is formulated for Turing machines, that compute in {\it discrete}, not continuous time)
they show the tremendous advantage of physics-based approaches to computation over traditional algorithmic approaches.\\

\begin{mdframed}
\hfill\\
\noindent{ Box 1. \bf DMM for 3-SAT}

The 3-SAT formula is constructed by applying conjunction (AND), disjunction (OR), and negation (NOT) operations to Boolean variables (TRUE or FALSE), with parentheses used to indicate the order of operations~\cite{satreview}. A formula contains $N$ Boolean variables ($y_i$), $M$ clauses, and $3M$ literals. Each clause (constraint) consists of three literals connected by logical OR operations, i.e., $(l_i\vee l_j\vee l_k)$, where a literal, $l_i$, is simply one of the Boolean variables ($l_i=y_i$) or its negation ($l_i=\bar{y_i}$). A clause is satisfied if at least one literal is TRUE (OR operations), and the formula is satisfiable if all clauses (AND operations) are simultaneously satisfied. The complexity of the problem emerges from the interaction among constraints, and is observed in the well-studied easy-hard-easy transition in 3-SAT, where easy and hard regimes are identified by the ratio $\alpha_r=M/N$, with the complexity peak (hardest instances) occurring around $\alpha_r=4.27$~\cite{easy-hard-easy}.\\

To construct a DMM that finds a satisfying assignment for 3-SAT we follow the general procedure outlined in Ref.~\citenum{DMMperspective}. To begin, the Boolean variables, $y_i$, are transformed into continuous variables for use in the DMM. The continuous variables can be realized in practice as voltages on the terminals of a self-organizing OR gate~\cite{DMM2}. Such a gate can
influence its terminals to push voltages towards a configuration satisfying its OR logic {\it regardless} of whether the signal received by the gate originates from the traditional input or the traditional output (see Fig.~\ref{fig:circle_circuit}). The voltages are bounded, $v_i\in [-1,1]$, with Boolean values recovered by thresholding: TRUE if $v_i>0$, FALSE if $v_i<0$, and ambiguous if $v_i=0$. To perform the logical negation operation on the continuous variable, one trivially multiplies that quantity by $-1$. The self-organizing logic circuit that comprises the DMM is built by connecting all of the self-organizing OR gates (see Fig.~\ref{fig:circle_circuit}). See Sec. III.A of SM for an extended discussion of the thresholding procedure for the voltages. \\

Next, we interpret a Boolean clause as a dynamical constraint function, with its state of satisfaction determined by the voltages. The $m$-{\it th} Boolean clause, $(l_{i,m} \vee l_{j,m} \vee l_{k,m})$, becomes a constraint function, 
\begin{equation}
C_{m}(v_{i},v_{j},v_{k}) = \frac{1}{2}\min[(1- q_{i,m}v_i),(1 - q_{j,m}v_j),(1 - q_{k,m}v_k)],
\end{equation}
where $q_{i,m}=1$ if $l_{i,m}=y_i$, and $q_{i,m}=-1$ if $l_{i,m}=\bar{y}_i$. The function is bounded, $C_m \in [0,1]$, and a clause is necessarily satisfied when $C_m<1/2$. The instance is solved when  $C_m<1/2$ for all clauses. By thresholding the clause function we avoid the ambiguity associated with $v_i=0$. If some voltage is ambiguous ($v_j=0$) and all clauses are satisfied, then any Boolean assignment to $y_j$ will be valid in that configuration.
The use of a minimum function in $C_{m}$ preserves an important property of 3-SAT. A clause is a constraint, and, by itself, a clause can only constrain one variable (via its literal). (Note that the minimum operation introduces some form of discontinuity to the dynamical system, for which we develop the formalism to study in Secs. IV and V of SM.)
The values of two literals are irrelevant to the state of the clause
if the third literal results in a satisfied clause. \\

Finally, a DMM employs memory variables to assist with the computation~\cite{DMM2,DMMperspective}. The memory variables transform equilibrium points that do not correspond to solutions of the Boolean formula into unstable points in the voltage space (see Sec. VIII of SM), leaving the solutions of the 3-SAT problem as the only minima. 
We choose to introduce two memory variables per clause: short-term memory, $x_{s,m}$, and long-term memory, $x_{l,m}$. The terminology intuitively describes the behavior of their dynamics. For the short-term memory, $x_{s,m}$ lags $C_{m}$, acting as an indicator of the recent history of the clause. For the long-term memory, $x_{l,m}$ collects information so it can ``remember'' the most frustrated clauses, weighting their dynamics more than clauses that are ``historically'' easily satisfied. Both the number and type of memory variables, as well as the form of the resulting dynamical equations, are not unique provided neither chaotic dynamics nor periodic orbits are introduced~\cite{DMMperspective}. \\

We choose for the dynamics of voltages and memory variables the following,
\begin{eqnarray}
&&\dot{v}_n = \sum_{m} x_{l,m}x_{s,m}G_{n,m}(v_{n},v_{j},v_{k}) + (1+\zeta x_{l,m})(1-x_{s,m})R_{n,m}(v_{n},v_{j},v_{k}),
\label{eq:voltages}\\
&&\dot{x}_{s,m} = \beta (x_{s,m}+\epsilon)(C_m(v_{n},v_{j},v_{k})-\gamma), 
\label{eq:xshort}\\
&&\dot{x}_{l,m} = \alpha (C_m(v_{n},v_{j},v_{k})-\delta),
\label{eq:xlong}\\
&&G_{n, m}(v_{n},v_{j},v_{k}) = \frac{1}{2} q_{n,m} \min[(1 -q_{j,m}v_j),(1 - q_{k,m}v_k)],
\label{eq:G}\\
&&\begin{split}
R_{n,m}(v_{n},v_{j},v_{k}) &= 
&\begin{cases} 
\frac{1}{2}( q_{n,m}-v_n), & C_m(v_{n},v_{j},v_{k})=\frac{1}{2}(1- q_{n,m}v_n), \\
0, & \textrm{otherwise},
\end{cases}
\end{split}
\label{eq:rigid}
\end{eqnarray}\\
where $G_{n,m}$ and $R_{n,m}$ equal 0 when variable $n$ does not appear in clause $m$, and the summation is taken over all constraints in which the voltage appears. The memory variables are {\it bounded}, with $x_{s,m}\in [0,1]$ and $x_{l,m}\in [1,10^4M]$. The boundedness of voltage and memory variables implies that there are no diverging terms  
in the above equations (see Sec. VI.B of SM).\\

The parameters $\alpha$ and $\beta$ are the rates of growth for the long-term and short-term memory variables, respectively.
Each memory variable has a threshold parameter 
used for evaluating the state of $C_m$, and the two parameters are restricted to obey $\delta<\gamma <1/2$. (This also guarantees that there is a sufficiently large basin of attraction for the solutions. See Sec. VII of SM for a detailed explanation.).
Eq. (\ref{eq:xshort}) has a small, strictly-positive parameter, $0<\epsilon \ll 1$, to remove the spurious solution ($x_{s,m}=0$). However, $\epsilon$ additionally serves as a trapping rate in the sense that smaller values of $\epsilon$ make it more difficult for the system to flip voltages when some $C_m$ begins to grow larger than $\gamma$. \\

In Eq.~(\ref{eq:voltages}), the first term in the summation is a ``gradient-like'' term, the second term is a ``rigidity'' term~\cite{bearden2019critical}. The gradient-like term attempts to influence the voltage in a clause based on the value of the other two voltages in the associated clause. Consider the two extremes: if the minimum results is $G_{i,m}=1$, then $v_i$ needs to be influenced to satisfy the clause. Conversely, if the minimum gives $G_{i,m}=0$, then $ v_i$ does not need to influence the clause state (see Sec. II.A of SM).\\ 

The purpose of the three rigidity terms for a constraint is to attempt to hold one voltage at a value satisfying the associated $m$-{\it th} clause, while doing nothing to influence the evolution of the other two voltages in the constraint. 
Again, this aligns with the 3-SAT interpretation that a clause can only constrain one variable. The short-term memory variable acts as a switch between gradient-like dynamics and rigid dynamics. During the solution search, $G_m$ will seek to influence three voltages until clause $m$ has been satisfied. Then, as $x_{s,m}$ decays to zero, $R_m$ takes over. 
The long-term memory variables weight the gradient-like dynamics, giving greater influence to clauses that have been more frustrated during the solution search. The rigidity is also weighted by $x_{l,m}$, but reduced by $\zeta$.\\

\end{mdframed}

 \begin{figure*} [t]
	\centering
	\includegraphics[width=1 \textwidth]{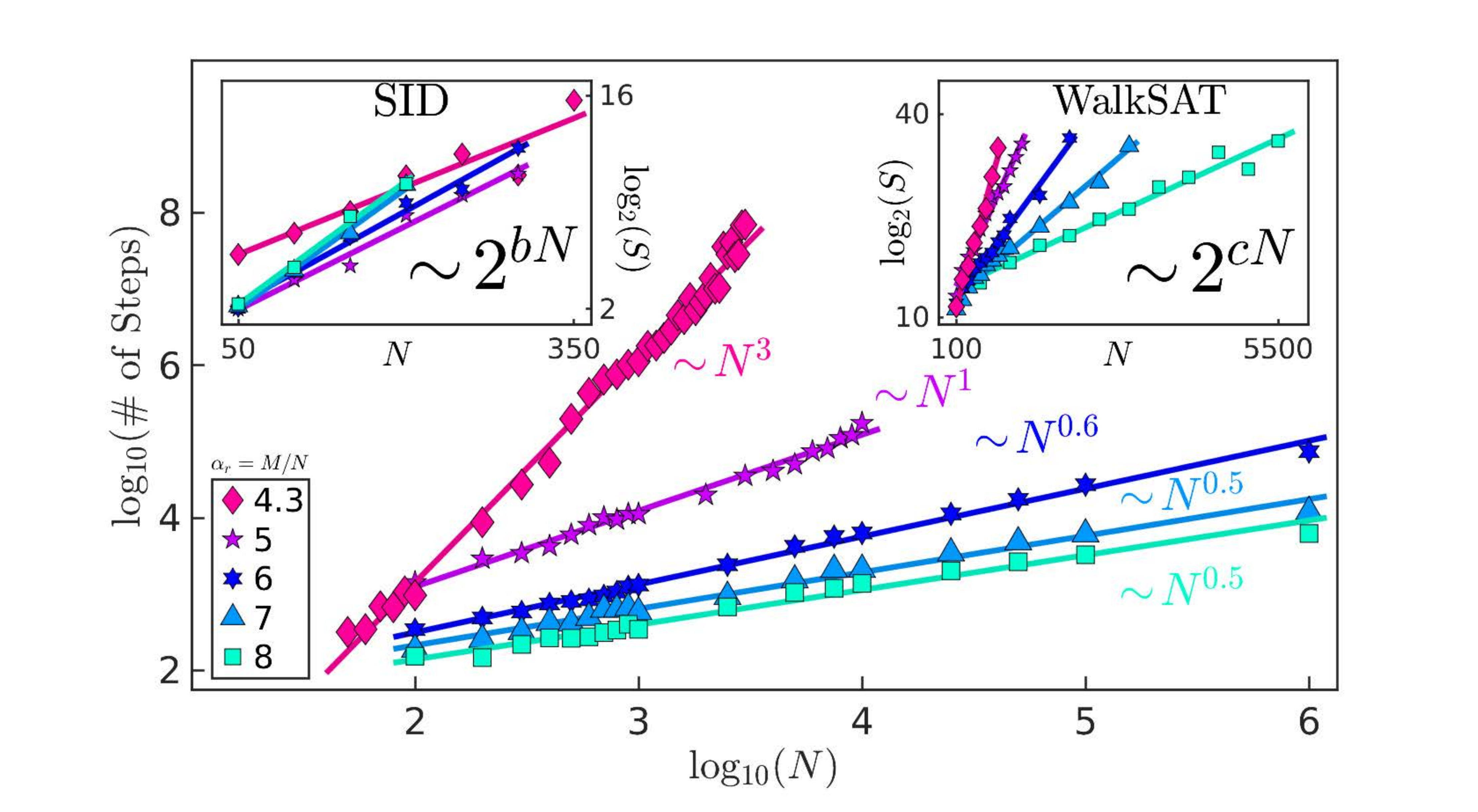}
	\caption{{\bf Typical case scalability of 3-SAT instances at fixed clause-to-variable ratio.} In the main panel, we use our DMM algorithm to attempt to solve 100 planted-solution instances of 3-SAT per pair of $\alpha_r$ (clause-to-variable ratio) and $N$ (number of variables). When we achieve more than 50 instances solved, we find power-law scalability of the median number of integration steps (typical case) as the number of variables, $N$, grows. (In the SM, we show many data points are comprised of 90 or more instances solved within the allotted time.)
 The exponent values ($\sim\!N^a$) are
		$a_{4.3}=3.0\pm0.1$, 
		$a_{5}=1.00\pm0.05$,
	$a_{6}=0.63\pm0.03$,
$a_{7}=0.48\pm0.03$, and
$a_{8}=0.46\pm0.04$.
		The insets show exponential scalability for a stochastic local-search algorithm (WalkSAT) and a survey-inspired decimation procedure (SID) on the same instances. (S is for number of steps.) Notice the scalability for SID has a trend opposite that seen in the DMM and WalkSAT. This is expected when one considers the increase in factor graph loops as $\alpha_r$ grows. For the SID scaling of $\alpha_r=4.3$, the $N=350$ did not achieve a median number of solutions, and is thus a lower bound. Parameters of the scaling for SID: $b_{4.3}=(3\pm1)\times 10^{-2}$,
	$b_{5}=(3.7\pm0.7)\times 10^{-2}$,
	$b_{6}=(4.1\pm0.6)\times 10^{-2}$,
	$b_{7}=(5\pm1)\times 10^{-2}$, and
	$b_{8}=(5\pm1)\times 10^{-2}$; for WalkSAT:
	$c_{4.3}=(3.2\pm0.3)\times 10^{-2}$,
	$c_{5}=(1.9\pm0.2)\times 10^{-2}$,
	$c_{6}=(1.2\pm0.1)\times 10^{-2}$,
	$c_{7}=(7.5\pm0.6)\times 10^{-3}$, and
	$c_{8}=(4.1\pm0.5)\times 10^{-3}$.}

	\label{fig:varying_ratio}
\end{figure*}

\begin{figure*}
	\centering
	\includegraphics[width=1 \textwidth]{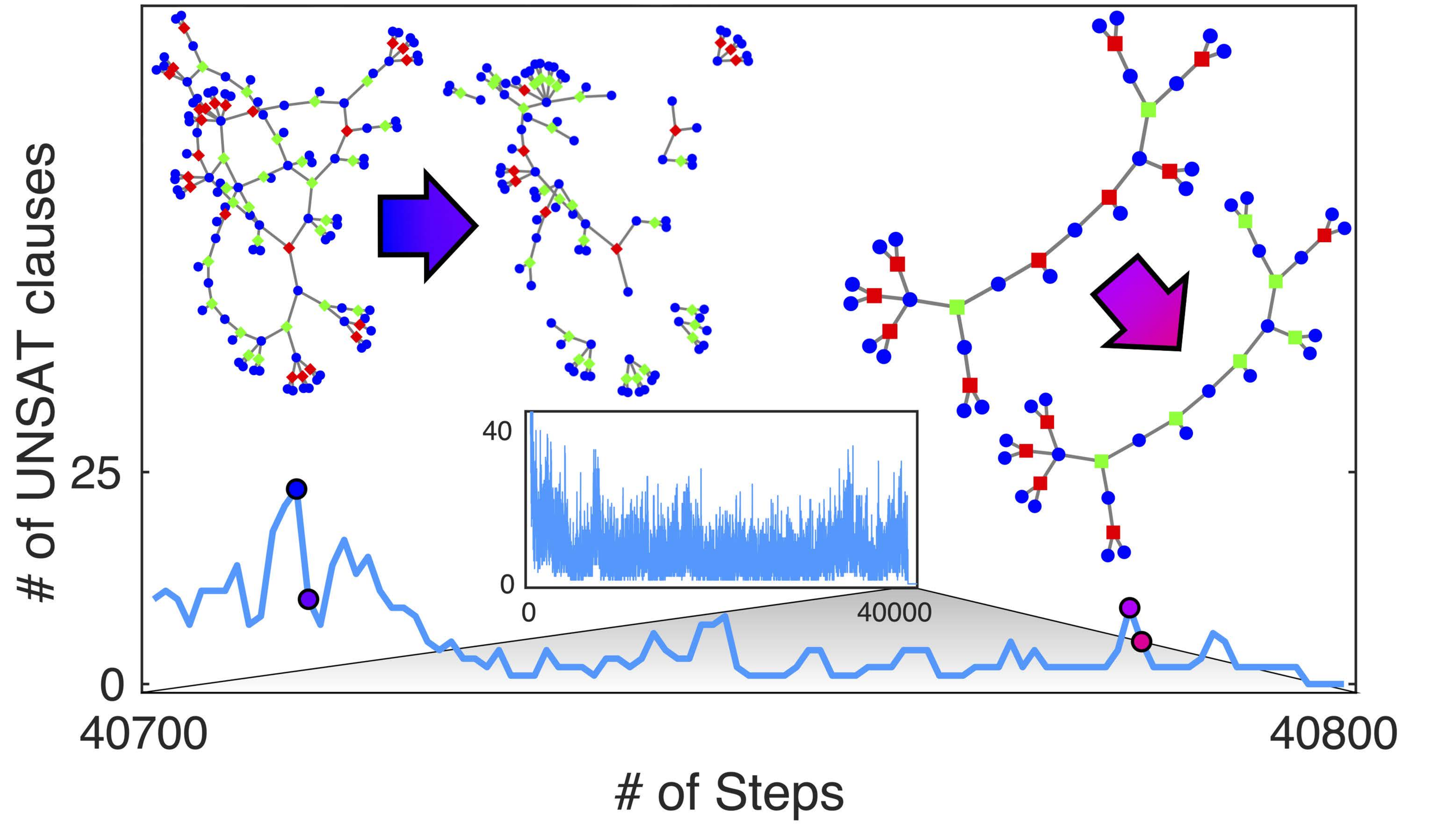}
	\caption{{\bf Time evolution of a typical DMM simulation showing collective updates to the solution search.} The figure highlights one solution attempt of a CDC instance of size $N=500$ at $\alpha_r=4.3$. The inset shows the number of unsatisfied clauses during the entire solution search. The main panel zooms in on the search as the solution is approached. We choose two single integration step transitions and explore the local factor graph. The circles are the variable nodes (blue), and squares are function nodes (red if unsatisfied, green if recently unsatisfied). The transition at left is characterized by 13 clauses becoming satisfied, the transition at right results in 4 clauses becoming satisfied. Neither transition results in satisfied clauses becoming unsatisfied.}
	\label{fig:factor_graph}
\end{figure*}

\noindent{\bf Numerical results and discussion}

It is important to realize that any simulation of a dynamical system is an algorithm because the continuous-time dynamics of the system must be discretized. Identifying our simulation as an algorithm invites a method to compare our results with those of popular algorithms, specifically, WalkSAT~\cite{selman1993walksat} and survey inspired decimation (SID)~\cite{mezard2002sid}. Before we compare results, we then need a general definition of a step.\\


We define an algorithmic step to be all the computation that occurs between checks of satisfiability. The WalkSAT algorithm flips one variable at a time then checks the satisfiability of the formula. Therefore, a WalkSAT step is a single variable flip. SID uses WalkSAT as part of its solution search, so the interpretation of steps is the same when SID uses WalkSAT. Prior to entering into WalkSAT, SID performs a message-passing procedure known as survey propagation~\cite{mezard2002sid}. 
In the SID implementation we used~\cite{grover2018streamlining} there is no check for satisfiability during the decimation procedure, so we generously identify the entire survey propagation with decimation as a single step. Our DMM algorithm checks the satisfiability of the formula after each time step of the integration. 
Of course, the amount of computation within a step may vary greatly based on the algorithm, but this does not affect comparison of the scalability. In fact, if an algorithm is exponential in the number of steps, then the amount of computation within a step cannot improve its scalability. For our DMM, each step has a constant amount of computation per time step of integration. With  this definition of an algorithmic step, we have a method to meaningfully compare the different algorithms. \\

We can now test these approaches on CDC instances with planted solutions. In Sec. III.C of the SM, we give an account 
of how these instances are generated, and why they are difficult to solve. Here, we just note that difficult CDC instances are created when $\alpha_r>4.25$ and $0.077<p_0<0.25$, where $p_0$ is the probability that the planted solution results in a clause with zero false literals~\cite{barthel2002}.
We have performed no preprocessing on the 3-SAT instances to reduce their size, not even the removal of pure literals (those appearing wholly negated or unnegated)~\cite{gu1999algorithms}. \\

We numerically integrated Eqs.~(\ref{eq:voltages}),~(\ref{eq:xshort}), and~(\ref{eq:xlong}) with the forward-Euler method using an adaptive time step, $\Delta t\in[2^{-7},10^3]$. For parameters, we have used $\alpha = 5$, $\beta = 20$, $\gamma=1/4$, $\delta = 1/20$, and $\epsilon = 10^{-3}$. For high ratio, $\alpha_r\geq6$, we find $\zeta=10^{-1}$ to provide better scaling results. For ratios that approach the complexity peak, we used $\zeta=10^{-2}$ for $\alpha_r=5$, and $\zeta=10^{-3}$ for $\alpha_r=4.3$. 
In Fig. \ref{fig:varying_ratio}, we report the results for CDC instances generated with $p_0=0.08$.
In our simulations, we expectedly find the difficulty of CDC instances increases with increasing $p_0$ (see Sec. II in SM). \\

In Fig. \ref{fig:varying_ratio}, for the problem sizes tested,
 we find a power-law scaling for the median number of integration steps for the simulations of DMMs. We also find that integration time variable ($t$), CPU time, and long-term memory ($x_l$) are bounded by a polynomial scaling, and the average step size shows power-law decay (see Sec. II.C of SM). The optimized WalkSAT algorithm~\cite{walksat56} we have used  instead exhibits an exponential scaling at relatively small problem sizes, confirming the previous results of Ref.~\citenum{barthel2002}. An exponential scaling is also observed for the SID algorithm~\cite{grover2018streamlining}.\\

The CDC instances are structured to confuse stochastic local-search algorithms, so the exponential scaling of WalkSAT is expected (right inset Fig.~\ref{fig:varying_ratio}). To understand the exponential performance of SID (left inset Fig.~\ref{fig:varying_ratio}), we need to understand the success of SID on random 3-SAT.
When generating uniform random 3-SAT at the complexity peak with a general method (no planted solutions), the typical case can be exploited by SID due to the existence of treelike structures in the factor graph~\cite{braunstein2005survey}.
(For those unfamiliar with factor graphs, if the factor graph was a tree, then one would be able to visually, thus easily, find the solution from the graph~\cite{Mezard}.) However, as demonstrated in Fig.~\ref{fig:varying_ratio}, SID performs poorly when given a 3-SAT instance with a factor graph that is not locally treelike. It is also known that SID performs poorly at high ratios ($\alpha_r \gtrsim 4.25$)~\cite{parisi2003remarks}, as loops in the factor graph become more common, explaining the opposite scaling trend seen in Fig.~\ref{fig:varying_ratio}.\\

To further confirm that the usefulness of our DMM algorithm on CDC instances is independent of our generation of formulae, we have solved generalized CDC instances~\cite{cdc2016} used in the 2017~\cite{satcomp2017benchmarks} and 2018~\cite{satcomp2018benchmarks} SAT competitions (satcompetition.org).
Our modified competition DMM solves {\it all} tested competition CDC instances on its first attempt with random initial conditions, and does so within the 5000-second timeout established by the competition (see Sec. II.E of SM). We find the overhead of numerical simulations of ODEs does not forbid our DMM from being competitive due to the use of the forward-Euler integration scheme.\\

\noindent{\bf Long-range order and analytical properties of DMMs for 3-SAT}

We finally show that collective behavior ({\it long-range order})~\cite{topo,topo1} in DMMs is responsible for the observed efficiency in the solution search. In order to do this, it is helpful to visualize subgraphs of the factor graph generated from a 3-SAT instance. In Fig.~\ref{fig:factor_graph}, we visualize the change in state of local factor graphs during a single time step of integration as our DMM approaches a solution.
It is apparent that the system explores many paths in the factor graph, collecting information as it does. 
However, unlike SID, when the DMM explores a path leading to contradiction it can {\it correct itself}.
The factor graphs shown in Fig.~\ref{fig:factor_graph} only include clauses (function nodes) that are unsatisfied (red) or recently unsatisfied (green), and all variable nodes connected to these clauses. A clause, $m$, is identified as recently unsatisfied if the short-term memory is $x_{s,m}>0$ but the clause is currently satisfied. The factor graph transitions show that {\it collective events} occur that satisfy multiple clauses. This is in agreement with many results on DMMs for different types of problems~\cite{topo,spinglass}. Additionally, the factor graph transition on the left of Fig.~\ref{fig:factor_graph} breaks up the graph into smaller, disconnected factor graphs, making the search exponentially more efficient. \\


As anticipated, to strengthen these numerical results, we have also analytically demonstrated that the dynamics described by Eqs.~(\ref{eq:voltages}),~(\ref{eq:xshort}), and~(\ref{eq:xlong}) terminate {\it only} when the system has found the solution to the 3-SAT problem (namely the phase space has only saddle points and 
the minima corresponding to the solution of the given problem; Secs. VI and VII of SM). In addition, neither periodic orbits nor chaos can coexist 
if solutions of the 3-SAT are present (Sec. IX of SM). Finally, using supersymmetric topological field theory, we have demonstrated that the continuous-time dynamics (physical implementation) reach the solution of a 3-SAT instance, for a fixed $\alpha_r$, in linear or sub-linear {\it continuous} time, irrespective of the difficulty of the instance (Sec. XI of SM). \\

However, note that such a scalability {\it does not} necessarily translate to the same scalability of the {\it numerical} integration of Eqs.~(\ref{eq:voltages}),~(\ref{eq:xshort}), and~(\ref{eq:xlong}), where the discretization of time is necessary.
Nevertheless, due to the absence of chaos, we empirically find that the scalability of our numerical simulations is still polynomially bounded for typical-case CDC instances.\\

\noindent{\bf Conclusions}

We have presented an efficient dynamical-system approach to solving Boolean satisfiability problems. Along with arguments for polynomial-time scalability in {\it continuous} time, we have found that the {\it numerical} integration of the corresponding ODEs show power-law scalability for typical cases of 3-SAT instances which required exponential time to solve with successful algorithms. 
The efficiency derives from {\it collective} updates to the variables during the solution search ({\it long-range order}). \\

In contrast to previous work~\cite{zoltan2011}, our dynamical systems do not suffer from exponential fluctuations in the energy function due to chaotic behavior. The dynamical systems we propose find the solution of a given problem without ever entering a chaotic regime, 
by virtue of the variables being bounded. 
The implication is that a hardware implementation 
of DMMs would only require a polynomially-growing energy cost. Our work then also serves as a counterexample to the claim of Ref.~\citenum{zoltan2011} that chaotic behavior is necessary for the solution search of hard optimization problems. In fact, we find chaos to be an undesirable feature for a scalable approach (See Sec. II.F of SM). \\

Although these analytical and numerical results do not settle the famous P vs. NP question,
they show that appropriately designed physical systems are very useful tools for new avenues of research in constraint satisfaction problems. \\

\noindent {\bf Data availability}\\
All instances used to generate all figures in this paper are available upon request from the authors.\\

\noindent {\bf Acknowledgments}\\
Work supported by DARPA under grant No.
HR00111990069. M.D., S.R.B.B., and Y.R.P. also acknowledge partial support from the Center for Memory and Recording Research at the University of California, San Diego. S.R.B.B. acknowledges partial support from the NSF Graduate Research Fellowship under Grant No. DGE-1650112, and from the Alfred P. Sloan Foundation's Minority Ph.D. Program.\\

\noindent {\bf Author contributions}\\
M.D. has supervised the project. S.R.B.B. has performed all simulations reported and designed the digital memcomputing machine employed in this work. Y.R.P. proved the theorems in the SM. All authors have discussed the results and contributed to the writing of the paper.\\

\noindent{\bf Competing Interests}\\
M.D. is the co-founder of MemComputing, Inc. (https://memcpu.com/) that is attempting to commercialize the memcomputing technology. All other authors declare no competing interests.

%

\newpage

\section*{Supplementary Material: Efficient Solution of Boolean Satisfiability Problems with Digital MemComputing}

\newpage
\section{Summary of Major Results}

For the benefit of the reader we summarize the major results presented in this Supplementary Material (SM). 

\begin{itemize}
	\item In Section \ref{Addsimulation} we describe the numerical method and implementation we used to solve Eqs. (2)-(4) in the main text. We also show several other numerical results on additional 3-SAT instances to support the ones reported in the main text. In particular, we show that the time variable of integration, CPU time, and slow memory variables all scale as a {\bf power law} in the size of the problem. We also show that the 
	average time step of the integration needs only to decrease as a {\bf power law} with increasing problem size. 
	\item In Sections \ref{lip} and \ref{cara}, we show that a {\bf unique} dynamical trajectory can be constructed for the discontinuous flow field governing our dynamics. For practical purposes, the analytic trajectory is constructed such that it is approximated by the numerical trajectory obtained with the forward Euler integration method used in our numerical analysis reported in the main text.
	\item In Section \ref{sec_mem}, we show that our dynamics are {\bf bounded} by a positive invariant compact set, and the dynamics terminate {\bf only} when the system has found the solution to the 3-SAT problem. This guarantees a correspondence between the fixed points of the dynamics and the solutions of the 
	3-SAT problem, and absence of local minima.
	\item In Section \ref{basin}, we show that the {\bf basin of attraction} of the solution for our flow field contains a large hypercube in the voltage space. In other words, once the trajectory has entered this region, the dynamics are guaranteed to converge to a solution.
	\item In Section \ref{no_period}, we show the {\bf absence of periodic orbits} in the voltage dynamics. This result, augmented by the fact that the memcomputing flow is {\bf not topologically transitive}, implies {\bf absence of chaos} (\`a la Devaney). 
	\item In Section \ref{diss}, we show that our system is {\bf dissipative}, in the sense that the volume of any initial set in the phase space contracts under the flow field.
	\item In Section \ref{TFT} we show using topological field theory that the {\bf continuous-time} dynamics reach the fixed points in a time that scales with problem size, $n$, as $O(n^{\alpha})$ with $\alpha \leq 1$. This result does not necessarily apply to the numerical solution of the dynamical equations due to integration overhead and numerical noise. 
	
\end{itemize}

\section{Numerical implementation and additional simulation results}\label{Addsimulation}

\subsection{Numerics}\label{Numerics}
For ease of discussion, the equations of motion for the digital memcomputing machine (DMM) are reproduced here. The $m$-{\it th} Boolean clause, $(l_{i,m} \vee l_{j,m} \vee l_{k,m})$, becomes a clause function, 
\begin{equation}
C_{m}(v_{i},v_{j},v_{k}) = \frac{1}{2}\min[(1- q_{i,m}v_i),(1 - q_{j,m}v_j),(1 - q_{k,m}v_k)],
\label{eq:C}
\end{equation}
where $q_{n,m}=1$ if $l_{n,m}=y_n$, and $q_{n,m}=-1$ if $l_{n,m}=\bar{y}_n$. The DMM's equations then read:

\begin{eqnarray}
&&\dot{v}_n = \sum_{m} x_{l,m}x_{s,m}G_{n,m} + (1+\zeta x_{l,m})(1-x_{s,m})R_{n,m},
\label{eq:voltages}\\
&&\dot{x}_{s,m} = \beta (x_{s,m}+\epsilon)(C_m-\gamma), 
\label{eq:xshort}\\
&&\dot{x}_{l,m} = \alpha (C_m-\delta),
\label{eq:xlong}\\
&&G_{n, m}(v_{n},v_{j},v_{k}) = \frac{1}{2} q_{n,m} \min[(1 -q_{j,m}v_j),(1 - q_{k,m}v_k)],
\label{eq:G}\\
&&\begin{split}
R_{n,m}(v_{n},v_{j},v_{k}) &= 
&\begin{cases} 
\frac{1}{2}( q_{n,m}-v_n), & C_m(v_{n},v_{j},v_{k})=\frac{1}{2}(1- q_{n,m}v_n), \\
0, & \textrm{otherwise},
\end{cases}
\end{split}
\label{eq:rigid}
\end{eqnarray}\\
where $G_{n,m}$ and $R_{n,m}$ equal 0 when variable $n$ does not appear in clause $m$.

In Eq.~(\ref{eq:voltages}), each of the $N$ voltages (variables) are guided by $M$ constraints (clauses). Each constraint influences three voltages simultaneously, while switching between two dynamical terms containing a gradient-like function, $G_{n,m}$, and a ``rigidity'' function, $R_{n,m}$. \\

In addition to the voltages, memcomputing utilizes memory variables to assist with the computation. The short-term memory, $x_{s,m}$, controls the switching between  $G_{n,m}$ and $R_{n,m}$. The long-term memory, $x_{l,m}$ collects information so it can ``remember'' the most frustrated constraints (unsatisfied clauses), weighting their dynamics more than clauses that are ``historically'' easily satisfied.\\

To understand the gradient-like function better, consider the two extremes: if $G_{n,m}=1$, then $v_n$ needs to be influenced to satisfy the clause.  (Recall, there are three voltages associated with the $m$-\emph{th} constraint, but, independent of information from other constraints, no determination can be made on which voltage needs to be influenced.) Conversely, if $G_{n,m}=0$, then $ v_n$ does not currently need to influence the $m$-{\emph{th}} constraint state. The purpose of the rigidity term, $R_{n,m}$, is to attempt to hold one voltage at a value satisfying the associated $m$-\emph{th} constraint, but to do nothing to influence the evolution of the other two voltages in the constraint.\\

The long-term memory variable weights the gradient-like dynamics, giving greater influence to constraints that have been more frustrated during the solution search. The rigidity term is also weighted by $x_{l,m}$, but reduced by $\zeta$. The parameter $\zeta$ can be thought of as a ``learning rate''. More difficult instances, as characterized by their clause-to-variable ratio,
require more time for $x_{l,m}$ to evolve (slower learning rate) so the phase space can be more efficiently explored. \\ 

Note that the memory dynamics generate a {\it dynamical energy landscape} under which the voltages evolve. This guarantees that the trajectory has the ability to escape any local minima of the original, {\it static} energy landscape of the Boolean satisfiability problem. Visually, whenever the voltages fall into a local minimum of the 
original problem, the memory variables ``deform'' the energy landscape in such a way that the local minimum is transformed into a saddle point, and the trajectory is allowed to continue exploring the energy landscape until it finds the global minimum, which is left invariant by the memory variables (see proposition \ref{fix_sol}). An extended discussion of such dynamical properties is given in Section \ref{v_flow}.\\

It is advised to avoid $\gamma=1/2$ or $\gamma=\delta$.
When $C_{m}=1/2$ and $\gamma=1/2$  we find $\dot{x}_{s,m}=0$ and the system has difficulty leaving the ambiguous state. 
To avoid this complication, assign $\gamma<1/2$.
Eq.~(\ref{eq:xlong}) gives $x_{l,m}$ the ability to decay and it aids dynamics to have $0<\delta<\gamma$. Assigning $\delta$ less than $\gamma$ allows the system an indirect means to influence $x_{s,m}$ when $\dot{x}_{s,m}=0$ and $C_m\neq 0$, by allowing $x_{l,m}$ to continue to grow ($\dot{x}_{s,m}=0$ and $\dot{x}_{l,m}>0\implies \gamma =C_m > \delta$). The parameter $0<\epsilon \ll 1$ is chosen as a small positive number to guarantee that the dynamics of the short-term memory does not terminate when it reaches $x_{s,m} = 0$.  \\

Equations~(\ref{eq:voltages})-(\ref{eq:xlong}) have been numerically integrated with the forward Euler method using an adaptive time step, $\Delta t \in [2^{-7}, 10^3]$, until all clauses have been satisfied, as determined from thresholding Eq. \ref{eq:C} for all $m$. The code has been written in interpreted MATLAB R2019b. Each attempt at solving a clause distribution control (CDC) instance was performed on a single core (no parallelization employed) of an AMD EPYC 7401 server. \\

Note that the above integration scheme is the most basic and, hence, the most unstable we could implement. We thus expect more 
refined integration schemes may provide both better stability and scaling. 

\subsection{Trends of several indicators}

\begin{figure}[t]
	\begin{center}
		\includegraphics[scale=0.425]{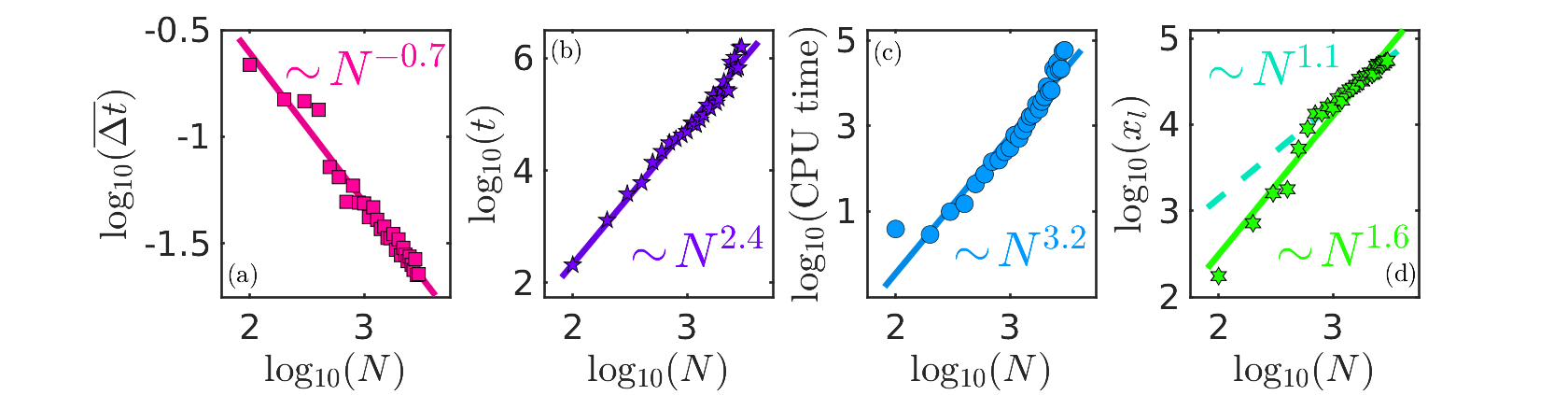}
		\caption{\small \label{stats} Typical-case analysis of other numerical indicators for the $p_0=0.08$, $\alpha_r=4.3$, CDC instances, with $N\in[100,200,...3000]$. Each data point is the median value of 100 instances. (a) The average time step, $\overline{\Delta t}$ (arb. units), showcases a power-law decay. (b) The median time-to-solution for the integration variable, $t$ (arb.units), scales with power-law growth. (c) The CPU time (seconds) scales with power-law growth (d) The median of the maximum values of $x_l$ (arb. units) when the solutions were found. There appears to be a transition in the data, so for a more informative fit (dashed line) we used data from $N\geq10^3$, which have nearly linear growth. }
	\end{center}
\end{figure}

In Fig. \ref{stats}, we show the typical-case behavior of other indicators in the DMM's dynamics as a function of problem size, $N\in[100,200,...3000]$, for difficult CDC instances, corresponding to $p_0=0.08$, $\alpha_r = 4.3$. Each data point is the median value of 100 instances, where 51 or more instances have been solved, with $N\leq2600$ having 90 or more instances solved before a timeout of $10^8$ steps. We observe a power-law growth ($\sim \!\!N^a$ with $a=2.4$ in Fig. \ref{stats}(b)) in the time variable, $t$ (arb. units), and in the CPU time ($a=3.2$ in Fig. \ref{stats}(c)), measured in seconds by MATLAB. \\

We also monitored the growth of $x_{l}$ to make sure there were no exponential ``energy'' costs. For each instance, we collect the maximum value of $x_{l}$, then find the median of those values. Figure \ref{stats}(d) confirms that 
the typical growth of the maximum value of $x_{l}$ follows a power law. 
Visually, we can see the fit ($a=1.6$) on the data from $N\in[100,200,...3000]$ is poor. However, when we fit data for $N\in[1000,1100,...,3000]$ the fit is almost linear ($a=1.1$), in agreement with the approximately linear growth in Eq.~(\ref{eq:xlong}) above (by taking $(C_m-\gamma)$ to be a positive constant).\\

Finally, we observe a power-law decay in the mean size of the time step of
our adaptive integration scheme as a function of problem size (Fig.~\ref{stats}(a)). 
In other words, as the problem size increases, the average time step is decreasing with a lower polynomial bound, rather than exponentially decaying. This observation invites modifications for speeding up solutions without introducing exponential growth into the system.

\subsection{Trends for different values of $p_0$}

In the generation of Barthel instances~\cite{barthel2002}, the parameter $p_0$ increases the backbone size as $p_0\rightarrow0.25$ (see also Sec.~\ref{barthelsection}). A large backbone implies, though not necessarily, a more difficult instance to solve because the solution space is smaller (less solutions). In Fig. \ref{p0}, we indeed see the exponent of the power-law scaling of the typical-case (median) CDC instances increases with increasing $p_0$.\\

The increase of backbone size also seems to cause issues with the forward-Euler integration scheme. We observe that our DMM algorithm encounters integration issues when attempting to extend these trends farther. This indicates that reducing the lower bound of the time step and/or a better integration scheme would be beneficial. Thus, we terminate simulations when the median number of steps is beyond $10^8$. \\

To effectively sample the distribution for typical-case analysis requires a larger sample size per data point. In Fig.~\ref{p0}, each data point represents the median number of steps for a sample of 500 instances.\\

For $p_0=0.08$ and $p_0=0.1$, $N\in[100,200,...1000]$. For $p_0=0.15$ and $p_0= 0.2$, the forward Euler integration scheme becomes unreliable before $N=1000$ could be reached, and such failures occur after $10^8$ steps. For $p_0=0.15$, $N\in[100,200,...600]$, and for $p_0=0.2$, $N\in[100,150,...350]$. This, again, indicates that decreasing the lower bound in the time step and/or a better integration scheme is needed for large $N$ instances. The power-law exponents calculated are $b_{0.08}=3.0\pm0.3$, $b_{0.10}=3.6\pm0.3$, $b_{0.15}=5.5\pm0.7$, and $b_{0.20}=6.6\pm1.1$. (Note that $b_{0.08}=3.0$ differs from the value reported in the main text because we are fitting data for $N\in[100,200,...,1000]$, with each data point being the median of 500 instances, rather than 100 as in the main text.) We compare these results to the 2017 Random Track competition winner, YalSAT~\cite{yalsat}, which clearly showcases exponential scaling.

\begin{figure}
	\begin{center}
		\includegraphics[scale=0.375]{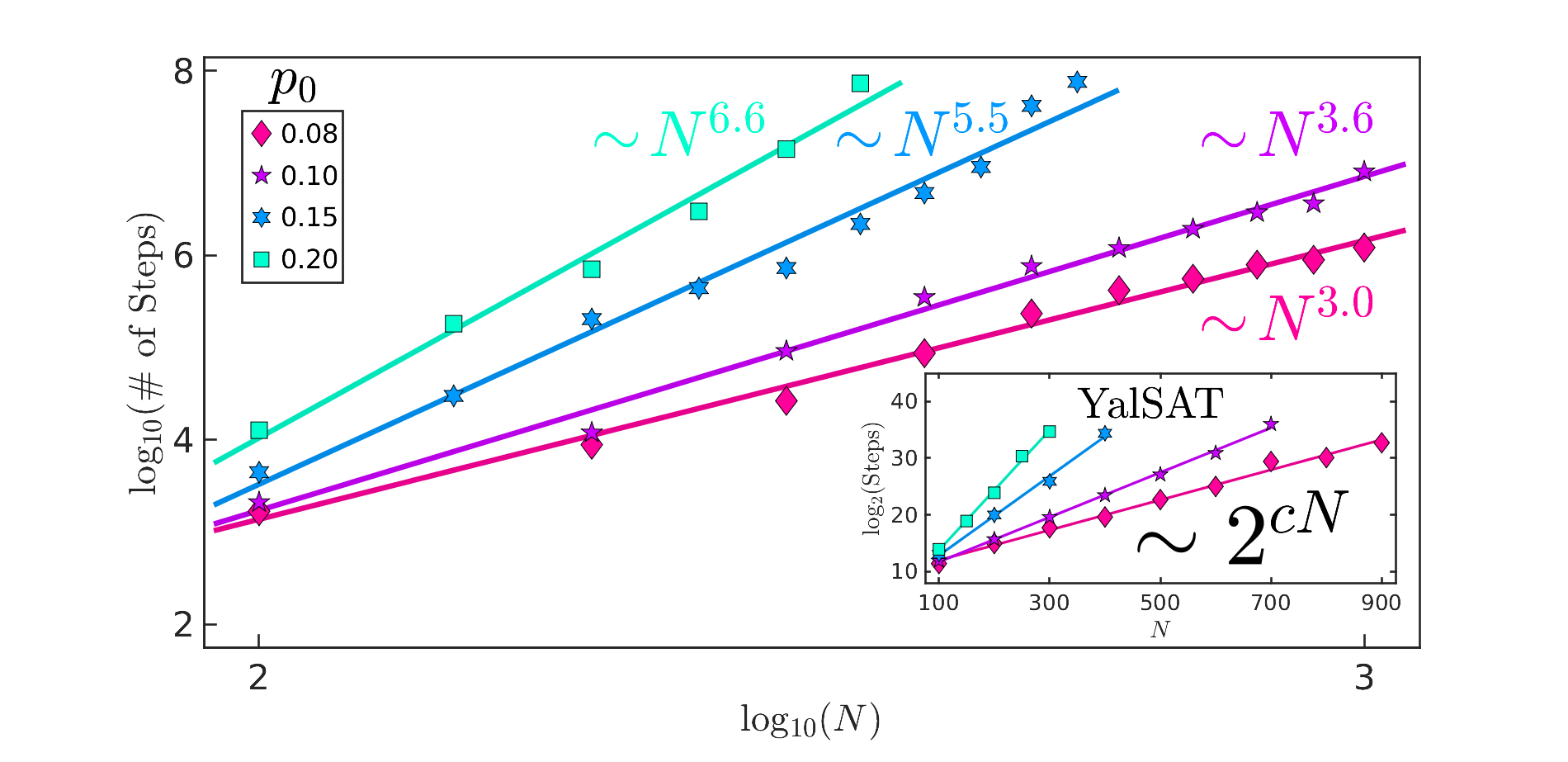}
		\caption{\small \label{p0} Evidence for power-law scaling ($aN^b$) for various values of $p_0$, with $b_{0.08}=3.0\pm0.3$, $b_{0.10}=3.6\pm0.3$, $b_{0.15}=5.5\pm0.7$, and $b_{0.20}=6.6\pm1.1$. (inset) We use the 2017 Random Track competition winner, YalSAT~\cite{yalsat}, to test scalability of a state-of-the-art algorithm. The fitted values of the exponential rates, in arbitrary units, are: $c_{0.08}=0.03\pm0.002$, $c_{0.10}=0.04\pm0.002$, $c_{0.15}=0.07\pm0.02$, $c_{0.20}=0.11\pm0.01$.}
	\end{center}
\end{figure}

\subsection{10-{\it th} to 90-{\it th} percentile range}
Here, we show results beyond our typical-case analysis without changing any parameters or integration scheme. We find a power-law trend as a function of problem size, $N$, at both the 10-{\it th} percentile and 90-{\it th} percentile for $p_0=0.08$. \\

In Fig.~\ref{90}, each data point for $\alpha_r=4.3$ is a median value of 100 instances, where $N\in[100,200,...,2600]$; $\alpha_r=5$, with $N\in[100,200,...,1000,]\cup[2000,3000,...,10^4]$; $\alpha_r=6$ with $N\in[100,200,...,1000]\cup[2500,5000,7500,10^4,2.5\times10^4,5\times10^4,10^5, 10^6]$. \\

Notice how the slopes of $\alpha_r=5,6$ appear to be converging. This may indicate that finite-size effects contribute to the variance of solution steps.
In Fig.~\ref{90}, the $\alpha_r=6$ data points at $N=10^6$ fall below their respective power-law trend lines. This behavior was also observed in Fig. 2 of the main text for $\alpha_r=6, 7, 8$. 
\begin{figure}
	\begin{center}
		\includegraphics[scale=0.4]{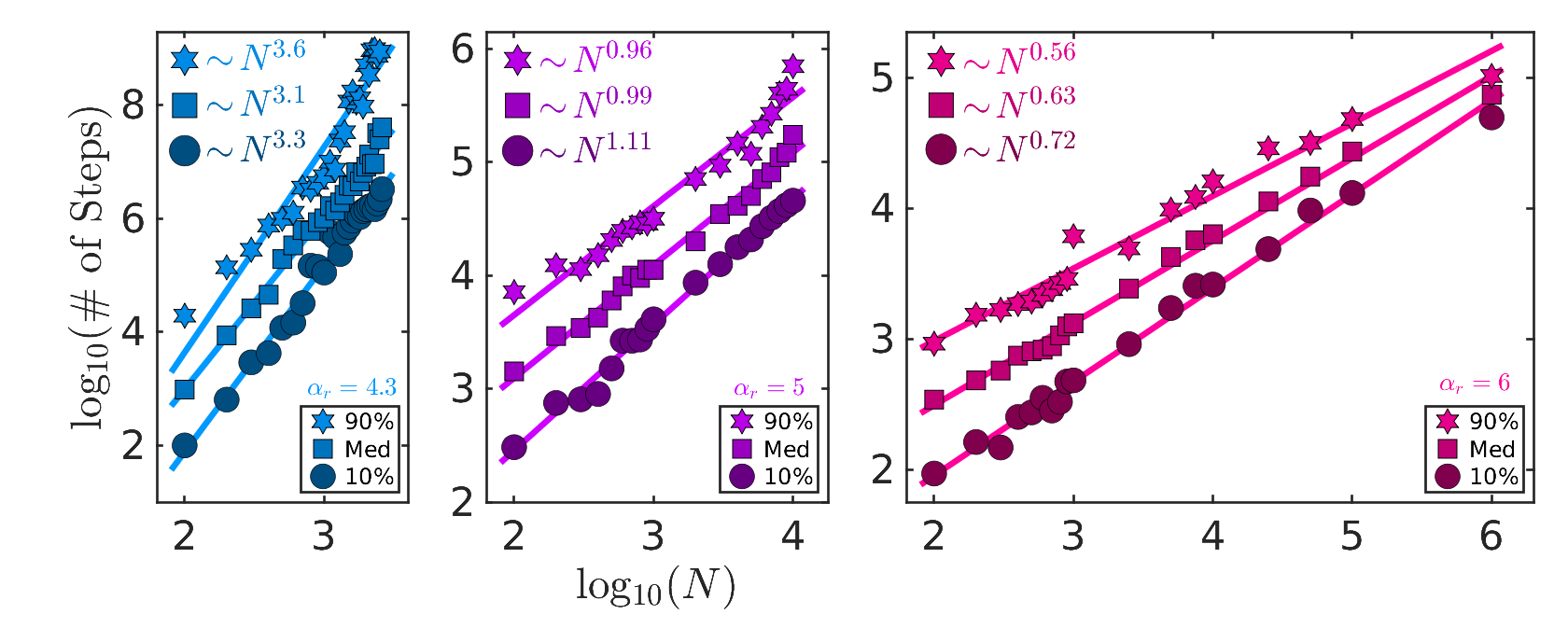}
		\caption{\small \label{90}Extending our typical-case analysis to include the 10-{\it th} and 90-{\it th} percentiles for $\alpha_r=4.3,5,6$ and 
			$p_0=0.08$ fitted to power-law trends.}
	\end{center}
\end{figure}

\subsection{Competition instances}

We sought an independent verification of our DMMs by applying them to instances taken from previous SAT competitions \cite{satcomp2017benchmarks,satcomp2018benchmarks}. Our solver was not designed for competition, so we added a heuristic to enhance its performance. Some competition instances are labeled ``barthel'' ($\alpha_r = 4.3$), ``komb'' ($\alpha_r = 5.205$), and ``qhid'' ($\alpha_r = 5.5$). 
As shown in Fig. \ref{comp}, our DMM is capable of solving all 285 competition problems from the 2017 and 2018 ``Random Tracks'' bearing one of these three labels. Furthermore, we can solve all of these instances within the competition's allotted CPU time (5000 second timeout). While we cannot directly compare CPU times of different machines, the reader can easily verify that our AMD EPYC 7401 server does not have any significant advantage over the machines used in the 2017 and 2018 competitions.\\

We chose to focus on the ``small'' competition instances because so many competition solvers failed to solve them.
For instance, in the 2017 Random Track there were 120 ``small'' instances that should be ``easy'' to solve in 5000 seconds. However, the 2017 Random Track winner (YalSAT) solved 124 out of 300 competition instances~\cite{satcomp2017benchmarks}.
Similarly, in the 2018 Random Track there were 165 ``small'' instances that should be ``easy'' to solve in 5000 seconds. The 2018 Random Track winner (Sparrow2Riss-2018) solved 188 out of 255 competition instances~\cite{satcomp2018benchmarks}. With the addition of more heuristics to our system, our DMM algorithm could possibly surpass previous competition performances.\\

We modified our algorithm to perform in the context of competition, by making one major modification: each constraint has its own $\alpha_m$ associated with $C_m$, and it is modified in regular intervals during the solution search. Initially, for all clauses, $\alpha_m=5$, and all other parameters remain unchanged from the main text. The search for the solution is initialized as before, but after $10^4$ arbitrary time units the simulation is paused to modify the values of $\alpha_m$. 
The procedure starts by finding the median of the $x_{l,m}$ values for all $m$. If $x_{l,m}$ is greater than the median, then the corresponding $\alpha_m$ is increased by a multiplicative factor of 1.1, otherwise, the corresponding $\alpha_m$ is decreased by a multiplicative factor of 0.9. To prevent decay to zero, $\alpha_m=1$ is the minimum value. If $x_{l,m}$ grows to its maximum cutoff, the process restarts by setting  $x_{l,m}=1$ and $\alpha_m=1$. The integration is resumed without modification to any other variables or parameters, and will repeat after another $10^4$ arbitrary time units.

\begin{figure}
	\begin{center}
		\includegraphics[scale=0.4]{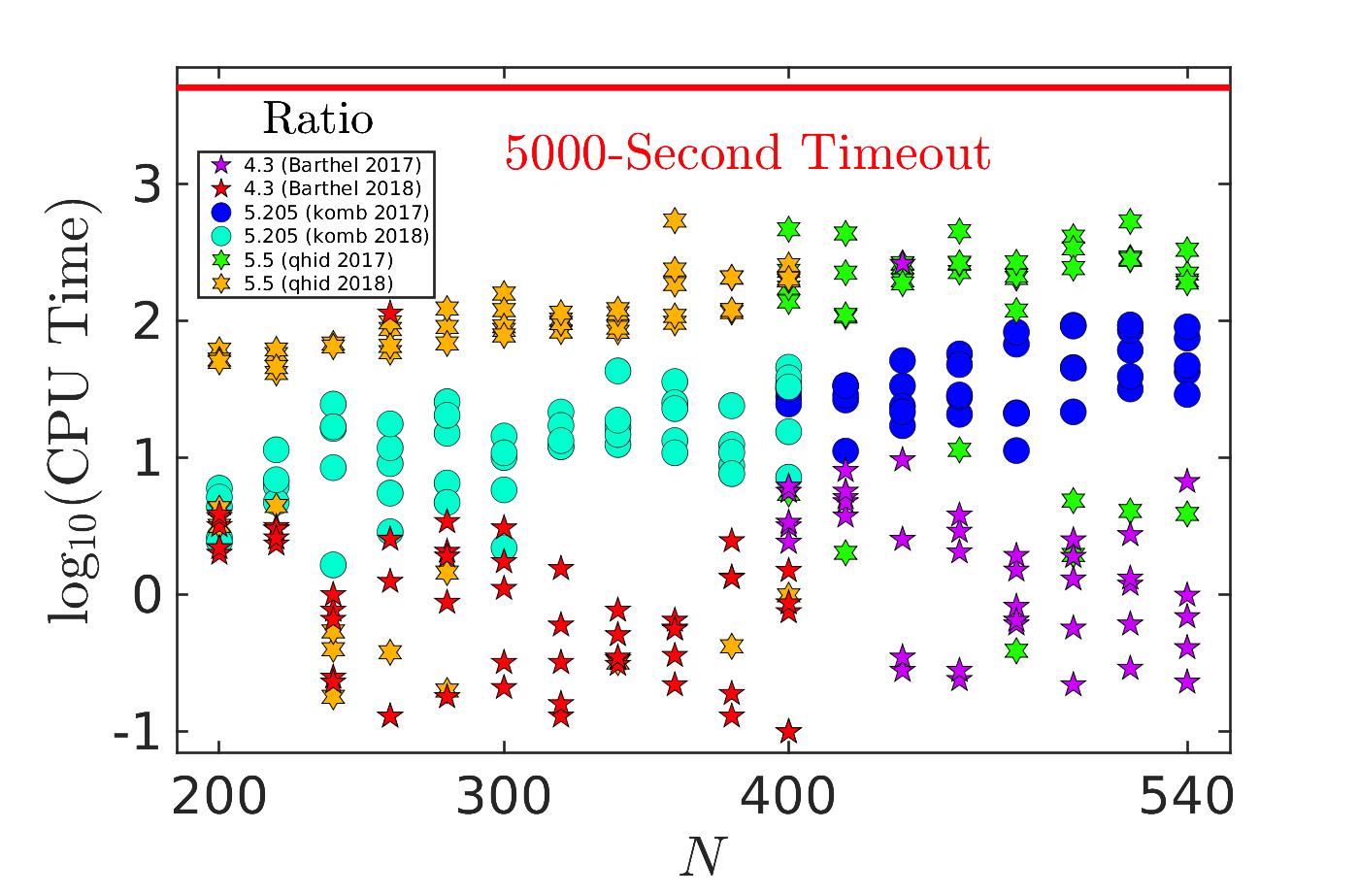}
		\caption{\small \label{comp} Results of a DMM algorithm solving competition instances (from the 2017 and 2018 SAT competitions \cite{satcomp2017benchmarks,satcomp2018benchmarks}), on a single core (no parallelization employed) of an AMD EPYC 7401 server, with only one set of random initial conditions. Note that some data points overlap.}
	\end{center}
\end{figure}

\subsection{Random 3-SAT}

While we have chosen to use planted-solution 3-SAT instances for the stated reasons (solution existence known), other authors~\cite{braunstein2005survey, zoltangpu} choose to work with 3-SAT instances that lack clause distribution control and have no guarantee of the existence of a solution. Recall that the algorithms discussed herein are all incomplete SAT solvers, meaning they cannot prove a solution does not exist (UNSAT). Therefore, scalability tests on general random 3-SAT instances have a degree of uncertainty regarding whether it is possible to find solutions to all instances tested. The SID algorithm removes much of the uncertainty by manipulating a property of the SAT/UNSAT transition: for $\alpha_r < 4.267$, the probability that a randomly generated instance has a solution approaches $1$ as $N$ grows~\cite{braunstein2005survey}. \\

When $N$ is small, it is unlikely all generated instances will be satisfiable, so the numerical simulation of AnalogSAT~\cite{zoltangpu} takes another approach to generate satisfiable instances. Starting with random 3-SAT instances, the authors use another algorithm, MiniSAT~\cite{minisat}, to filter the instances. That is, AnalogSAT is only tested on instances that MiniSAT can solve. However, this has the drawback of excluding  3-SAT instances that the filtering algorithm is incapable of solving. \\ 

In Fig. \ref{fig:zoltan}, our DMM solves all of the 3-SAT instances from Ref. \cite{zoltangpu} that have 100 instances per value of $N$. (Large $N$ instances only have 1 instance per value of $N$.) We use the same DMM and parameters as presented in the main text, where $\zeta=10^{-1}$ for $\alpha_r=3.4, 3.8$, and $\zeta=10^{-2}$ for $\alpha_r=4.25$

The authors of Ref. \cite{zoltangpu} prefer wall time as the indicator used to show polynomial scaling, claiming it is a realistic measure of hardware. Therefore, we show scaling of both steps (Fig. \ref{fig:zoltan}(a)-(c)) and wall time (Fig. \ref{fig:zoltan}(d)-(f)).

For very small values of $N$, our DMM encounters overhead that dominates the wall time scalability (solution wall time $\sim10^{-1}$ seconds). The initialization of the MATLAB code dominates the scalability for wall time so we exclude small values of $N$ from the curve fitting procedure. (Comparing the scaling of steps and wall time in Fig. \ref{fig:zoltan}, it can be seen there is no initialization effect in the scaling of steps.) With these considerations taken into account, we show several improvements.

In Fig. \ref{fig:zoltan}(d), for $\alpha_r=3.4$, we see the DMM's scalability of the maximum solution times, $\sim\! N^{1.25}$, is approximately the same as that reported for AnalogSAT's scaling of the mean, $\sim\! N^{1.26}$ \cite{zoltangpu}.
In Fig. \ref{fig:zoltan}(e), for $\alpha_r=3.8$, we see the DMM's scalability of the maximum solution times, $\sim\! N^{1.11}$, is better than that reported for AnalogSAT's scaling of the mean, $\sim\! N^{1.63}$ \cite{zoltangpu}. Additionally, our range of $N$ goes beyond $N=10^4$.
In Fig. \ref{fig:zoltan}(f), for $\alpha_r=4.25$, we see the DMM's scalability of the maximum solution times, $\sim\! N^{3.55}$, is better than that reported for AnalogSAT's scaling of the mean, $\sim\! N^{4.35}$ \cite{zoltangpu}. For the largest value of $N$ tested, $N=463$, the maximum wall time is 177 seconds, where AnalogSAT's mean wall time for the same value of $N$ is $\sim\! 10^3$ seconds.\\

For another test, we generated random 3-SAT instances at $\alpha_r=4.25$, where no solution has been planted (0-hidden). Due to being to the left of the SAT/UNSAT transition, there is a high probability that a randomly generated 3-SAT instance will be satisfiable. Therefore, we should be able to solve more than 50\% of instances generated, and can use the median as another measure of scalability. We use the same DMM and parameters as presented in the main text, with $\zeta=10^{-2}$. In Fig. \ref{fig:gen3sat}, we find power-law scalability for these instances as well.

\begin{figure}
	\begin{center}
		\includegraphics[scale=0.4]{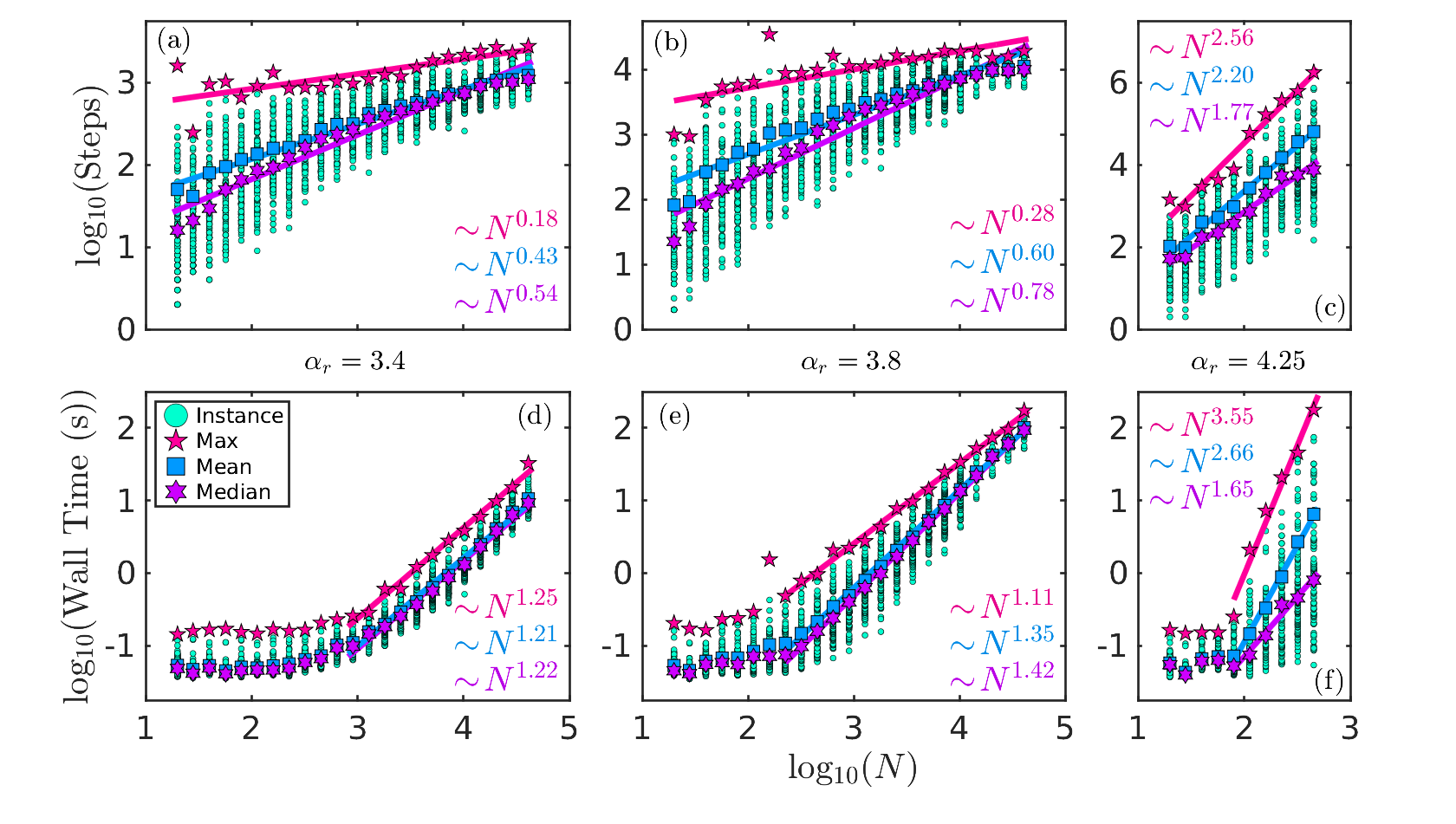}
		\caption{\small \label{fig:zoltan} Scalability on instances from Ref. \cite{zoltangpu} for $\alpha_r=3.4$ (a) and (d), $\alpha_r=3.8$ (b) and (e), $\alpha_r=4.25$ (c) and (f). We show scalability in integration steps (a)-(c) and wall time (d)-(f). For each $N$ there are 100 solved instances shown in each panel.}
	\end{center}
\end{figure}

\begin{figure}
	\begin{center}
		\includegraphics[scale=0.4]{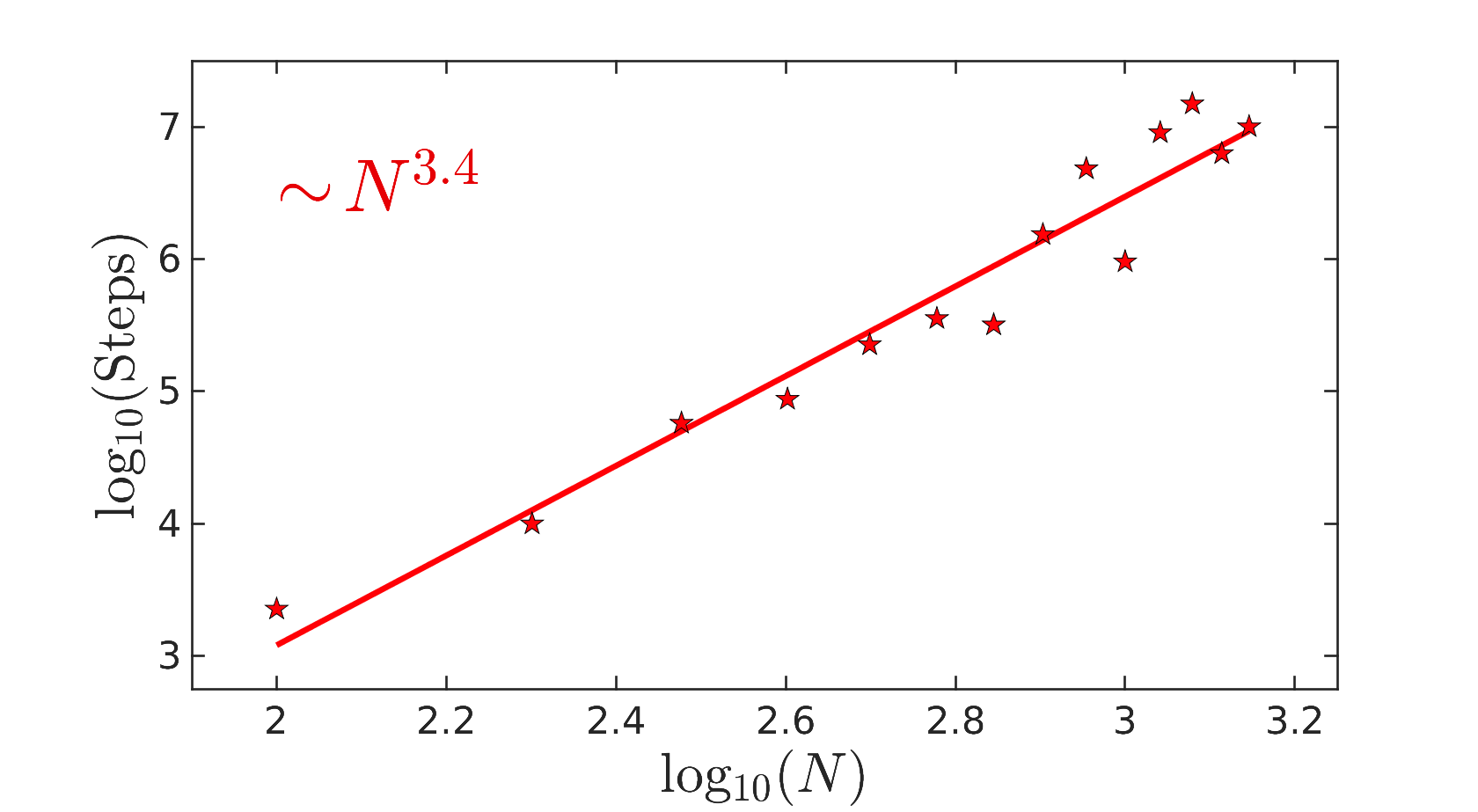}
		\caption{\small \label{fig:gen3sat} Scalability of typical general 3-SAT instances at $\alpha_r=4.25$, generated without knowledge of solution existence. For each $N$, we attempt to solve 100 general 3-SAT instances, and calculate the median when 51 instances have solved.}
	\end{center}
\end{figure}

\section{Continuous 3-SAT}
\label{prelim}

In this Section, we establish the formalism for studying the continuous version of the 3-SAT problem we have solved in the main text. This continuous version  generates an energy landscape that we explore with the memcomputing dynamics (see Section \ref{sec_mem}). In addition, we provide a brief discussion on the class of planted 3-SAT instances~\cite{barthel2002} that we used in this paper as benchmarks. To facilitate the theoretical analysis we will also slightly 
change the notation so that we can write Eqs.~(\ref{eq:C})-(\ref{eq:rigid}) in a more compact way.

\subsection{From Discrete to Continuous Variables}

Consider a 3-SAT Boolean formula with $n$ variables and $m$ clauses, where $\alpha_r$ is commonly referred to as the {\bf clause density}, as it is the ratio between the number of clauses and number of Boolean variables\footnote{Note that we are using a slightly different notational conventions from the main text. In the following, $n$ and $m$ are cardinal numbers denoting the numbers of variables and clauses respectively, and $i$ and $j$ are used as their respective indices.}. We let $+1$ correspond to the true assignment of a Boolean variable, and $-1$ to the false assignment. We then map the $n$ Boolean variables into $n$ continuous variables, $\mbf{v}\in [-1,+1]^n$, which we term {\it voltages}. For each clause, we can define various energy functions indicating the state of satisfaction of each clause given a voltage assignment. The expression of these functions are most compactly expressed by making use of the definition of {\it polarity}.

\begin{definition}[Polarity and Constraint]
	\label{c_energy}
	Consider a 3-SAT Boolean formula with $n$ Boolean variables and $m$ clauses. We denote the $i$-th Boolean variable as $x_i$, and its {\bf polarity} in the $j$-th clause as
	\begin{equation*}
	q_{ij} 
	=
	\begin{cases}
	+1 \quad &\text{if $x_i$ appears positively in clause $j$,} \\
	-1 \quad &\text{if $x_i$ appears negatively in clause $j$,} \\
	0  \quad &\text{if $x_i$ does not appear in clause $j$}.
	\end{cases}
	\end{equation*}
	The polarity matrix, $Q$, is the matrix with the element on the $i$-th row and $j$-th column being $q_{ij}$. Note that a 3-SAT Boolean formula is completely specified by $Q$. \\
	
	Given a voltage assignment $\mbf{v} \in [-1, +1]^n$, we rewrite Eq. (\ref{eq:C}), the {\bf constraint} of the $j$-th clause as
	\begin{equation}
	\label{cj}
	C_j(\mbf{v}) = \frac{1}{2}\min_{ \{ i \cond q_{ij}\neq 0 \} }\big( 1-q_{ij}v_i \big).
	\end{equation}
	The {\bf global constraint} is the sum of the constraints of all clauses
	\begin{equation}
	\label{eng}
	\mcr{C}(\mbf{v}) = \sum_j C_j(\mbf{v}).
	\end{equation}
	For any $\mbf{v}_0$ such that $\mcr{C}(\mbf{v}_0) = 0$ is satisfied, we call $\mbf{v}_0$ a {\bf solution vector}. 
\end{definition}

\begin{remark}
	$\mbf{v}_0$ is called a solution vector because if we take the corresponding Boolean vector $\mbf{x}_0$ by thresholding $\mbf{v}_0$ (converting $v_{0,i}>0$ to true and $v_{0,i}\leq 0$ to false), then $\mbf{x}_0$ must be a solution to the original 3-SAT problem. This is because the global energy being zero, $\mcr{C}(\mbf{v}_0)=0$, implies that every clause energy must also be zero, $C_j(\mbf{v}_0)=0$, which further implies that every clause is satisfied under the assignment $\mbf{x}_0$. Note that the converse is also true; if $C(\mbf{v}') = 0$, then $\mbf{v'}$ must be a solution vector. \\
\end{remark}

To ease the burden of notation, it is useful to define the following index notation
\begin{equation}
\label{sigma}
\sigma_j = \argmin_{ \{ i \cond q_{ij}\neq 0 \} }( 1 - q_{ij}v_i ),
\end{equation}
which can be simply interpreted as the index of the voltage whose assignment is closest to satisfaction among all voltages in clause $j$. Note that by this definition, we have $q_{\sigma_j,j} = \pm 1$, denoting the polarity of the Boolean variable whose assignment determines the value of $C_j(\mbf{v})$. This notation allows us to simplify the expression of the clause constraint as given in Eq. (\ref{cj})
\begin{equation*}
C_j(\mbf{v}) = \frac{1}{2} (1 - q_{\sigma_j,j}v_{\sigma_j}).
\end{equation*}\\

Note that if the goal is for an effective numerical implementation of the memory dynamics solely as a means to find a solution, rather than relaxing into an equilibrium point, one can exploit the fact that if an assignment of $\mbf{v}$ such that $C_j(\mbf{v}) < \frac{1}{2}$ for every clause, then the original 3-SAT problem is solved by thresholding $\mbf{v}$ to generate $\mbf{x_0}$.

\begin{proposition}
	\label{sign}
	Given an assignment of the voltages $\mbf{v} \in [-1,+1]^n$ such that $\mbf{C}(\mbf{v}) < \frac{1}{2}$, $\sign(\mbf{v})$ is a solution vector\footnote{While it is possible for $v_i=0$, resulting in $\sign(0)=0$, this rare event does not affect the remaining nonzero voltages from satisfying all clauses. In such a scenario, $x_i$ can be set to TRUE or FALSE without affecting the satisfiability of the solution vector.}.
\end{proposition}

\begin{proof}
	Recall that $C_j(\mbf{v}) = \frac{1}{2}(1-q_{\sigma_jj}v_i)$. Since $\forall j$ we have $C_j(\mbf{v}) < \frac{1}{2}$, then $q_{\sigma_jj}v_{\sigma_j} > 0$. If we let $\mbf{v}_0 = \sign(\mbf{v})$, then $q_{\sigma_jj}v_{0,\sigma_j} = q_{\sigma_jj}\sign(v_{\sigma_j}) = \sign(q_{\sigma_jj}v_{\sigma_j}) = +1$, as $q_{\sigma_jj}=\pm 1$. Therefore, we have $C_j(\mbf{v}_0) = \frac{1}{2}(1 - q_{\sigma_jj}v_{0,\sigma_j}) = \frac{1}{2}(1 - 1) = 0$, so $\mcr{C} = \sum_j C_j(\mbf{v_0}) = 0$. Therefore, $\sign(\mbf{v})$ is a solution vector. 
\end{proof}

\begin{remark}
	This means that once we have discovered an assignment of voltages such that the constraints of all clauses are less than $\frac{1}{2}$, we can simply threshold the voltages to obtain the corresponding Boolean variables for a solution of the original 3-SAT problem. \\
\end{remark}

The global constraint defined in Eq.~(\ref{eng}) is not everywhere differentiable with respect to the voltages due to the use of a minimum operation, and this causes some inconvenience in analyzing certain properties of the 3-SAT problem structure from the perspective of statistical mechanics (see Eq. (\ref{ham})). We then construct an energy function that is continuous (and also smooth) in anticipation of such analysis.

\begin{definition}[Energy]
	\label{def_ham}
	Given a 3-SAT Boolean formula defined by an $n\times m$ polarity matrix $Q$, we define the {\bf energy} of the $j$-th clause for any voltage assignment $\mbf{v} \in [-1,+1]^n$ as
	\begin{equation}
	\label{c_ham}
	E_j(\mbf{v}) = \frac{1}{8} \prod_{ \{ i \cond q_{ij} \neq 0 \} } (1 - q_{ij}v_i).
	\end{equation}
	The {\bf global energy} is the sum of the energies for all clauses
	\begin{equation}
	\label{g_ham}
	\mcr{E}(\mbf{v}) = \sum_j E_j(\mbf{v}).
	\end{equation}
\end{definition}

\begin{remark}
	We can show in a similar fashion (see the remark of definition \ref{c_energy}) that if the 3-SAT problem is satisfiable, then the global energy of a solution vector $\mbf{v}_0$ will also be zero, or $\mcr{E}(\mbf{v}_0) = 0$, which is also its global minimum. The converse is also true. Therefore, the problem of minimizing the global constraint, $\mcr{C}$, and minimizing the global energy, $\mcr{E}$, are in fact equivalent problems. \\
\end{remark}

The flow field of the memory dynamics for the voltages (see Section \ref{sec_mem}) contains two terms, one being similar to the gradient of $\mcr{E}(\mbf{v})$ (see Eq.~(\ref{G})) which we name the {\bf gradient-like term} and the other one closely related the clause function $C_j(\mbf{v})$ (see Eq.~(\ref{R})) which we name the {\bf rigidity term}. At certain hyperplanes, the gradient-like term is not differentiable and the rigidity term is discontinuous (see section \ref{plane}). We develop the mathematical formalism for studying such irregular flow fields in Section \ref{cara}. 

\subsection{Gauging the 3-SAT Problem}
\label{gauge}

If the original 3-SAT Boolean formula has a known solution, analysis can be simplified by converting the 3-SAT formula into an equivalent 3-SAT formula in such a way that the known solution of the original formula is now a solution to the gauged formula with an all-true assignment of the Boolean variables\footnote{For all planted-solution CDC instances generated for numerical simulations, the all-true solution is first assumed and then randomly changed by a local gauge transformation~\cite{barthel2002} to remove any solver bias towards the all-true solution.}. After the conversion, there will be a restriction on the possible clause types that can appear in the formula (no clause appears with all variables negated). This will allow for a natural description of the clause distribution control (CDC) class of planted instances, and greatly simplify the analysis of memory dynamics.

\begin{definition}[Gauge Fixing]
	\label{def_gauge}
	Consider a satisfiable 3-SAT Boolean formula given by an $n\times m$ polarity matrix $Q$. Given any solution $\mbf{x}_0$ to the 3-SAT problem, we {\bf gauge} fix the polarity matrix $Q$ with respect to $\mbf{x}_0$, $G_{\mbf{x}_0}: \{-1,+1\}^{nm} \mapsto \{-1,+1\}^{nm}$, such that each element of $Q$ transforms as follows
	\begin{equation*}
	q_{ij} \mapsto x_{0,i}q_{ij}.
	\end{equation*}
	We refer to $Q' = G_{\mbf{x}_0}(Q)$ as the {\bf gauged} polarity matrix.
\end{definition}

\begin{remark}
	It can be easily shown that the formulas given by $Q$ and $Q'$ have the same structure. In particular, given some mapping of the polarity matrix $G_{\mbf{x}_0}$, we can simultaneously map each Boolean state $\mbf{x}$ to a new one as follow
	\begin{equation*}
	\mbf{x} \mapsto \mbf{x} \ast \mbf{x}_0,
	\end{equation*}
	where $\ast$ denotes component-wise multiplication. It is then obvious that the satisfaction state of each literal $q_{ij}x_i$ is invariant under this mapping. A similar procedure applies for mapping the voltages as well
	\begin{equation*}
	\mbf{v} \mapsto \mbf{v} \ast \mbf{v}_0.
	\end{equation*}
	
	Note that the performance of most SAT solvers (including the canonical Walk-SAT algorithm~\cite{selman1993walksat} and our memcomputing one as presented in Eqs.~(\ref{eq:C})-(\ref{eq:rigid})) are invariant under gauge conjugation \cite{mezard}. Informally, this means that nothing is gained or lost in terms of the efficiency of optimization by gauging the problem first before running the algorithm, as the behavior of a SAT solver at each time step will not change under a gauge mapping (see Section \ref{gauge_inv}). The choice to gauge fix a solution to $\mbf{+1}$ is purely for analytic convenience. \\
\end{remark}

An important property of a gauge fixed 3-SAT formula is that no clause can contain three negated Boolean variables.

\begin{lemma}\label{one}
	Given a gauged polarity matrix $Q$ of a k-SAT problem \cite{complexity_bible}, we have the following
	\begin{equation*}\label{one1}
	\forall j, \exists i, q_{ij} = +1.
	\end{equation*}
	In other words, all clauses must contain at least one literal that is an unnegated variable.
\end{lemma}

\begin{proof}
	We prove this by contradiction. We first assume the negation of the Lemma, then
	\begin{equation*}
	\exists j, \forall i, q_{ij} = -1.
	\end{equation*}
	Then without loss of generality (WLOG), we can assume that the $j$-th clause is the following
	\begin{equation*}
	( \overline{x}_1 \lor \overline{x}_2 \lor ... \lor \overline{x}_k).
	\end{equation*}
	Since $Q$ is a gauged polarity matrix, a solution must be $\mbf{x}_0 = \mbf{+1}$. However, this assignment evaluates to false by the above clause, so it cannot be a solution. We therefore have a contradiction.
\end{proof}

\begin{remark}
	It should be noted that the inclusion of clauses with all negations does not preclude the possibility of the formula being satisfiable, as solutions other than $\mbf{+1}$ may still exist. \\
\end{remark}

Lastly, we point out that the clause constraint defined in Eq.~(\ref{cj}) has the important property of being invariant under a gauge mapping.

\begin{lemma}[Gauge Invariance of Constraints]
	\label{cj_inv}
	Given a satisfiable 3-SAT instance with some solution vector $\mbf{v}_0$, $C_j(\mbf{v})$ is invariant under the following transformation for $\forall j$
	\begin{equation*}
	q_{ij} \mapsto q_{ij}v_{0,i}
	\qquad
	\mbf{v} \mapsto \mbf{v} \ast \mbf{v}_0. 
	\end{equation*}
\end{lemma} 

\begin{proof}
	Recall from Eq.~(\ref{cj}) that
	\begin{equation*}
	C_j(\mbf{v}) = \frac{1}{2}\min_{q_{ij}\neq 0}\big( 1-q_{ij}v_i \big).
	\end{equation*}
	If we let $q'_{ij} = q_{ij}v_{0,i}$ and $v'_i = v_i v_{0,i}$, then we have
	\begin{equation*}
	C'_j = \frac{1}{2}\min_{q'_{ij}\neq 0}\big( 1-q'_{ij}v'_i \big)
	= \frac{1}{2}\min_{q_{ij}\neq 0}\big( 1-q_{ij}v_i (v_{0,i})^2 \big)
	= C_j,
	\end{equation*}
	where we note that $v_{0,i} = \pm 1$, so $(v_{0,i})^2 = 1$.
\end{proof}

\begin{remark}
	It directly follows that the global constraint must be gauge invariant as well. It can be shown in a similar fashion that the energy of each clause is also gauge invariant.
\end{remark}

\subsection{Planted Instances}

Here, we consider a class of random 3-SAT instances generated with a planted solution to guarantee an instance to be satisfiable, however, planted in such a way so as to be hard for local-search SAT solvers to find~\cite{barthel2002}. In particular, we consider instances whose polarity matrix $Q$ satisfies Lemma \ref{one} up to a gauge mapping. In other words, when we construct $Q$, we cannot allow the appearance of columns whose nonzero elements are all $-1$. We formally describe a particular method of constructing such matrices in the following section.

\subsubsection{Randomly Planted Formula}
\label{gen_plant}

We first consider the general method of generating satisfiable formulas where every clause is formed independently by randomly including Boolean variables, with the clause type randomly sampled from some given distribution \cite{NOAP}.

\begin{definition}[Planted Instance]
	We consider a random matrix $Q$ generated by parameters $\{\alpha_r, p_0, p_1, p_2\}$ that satisfies the following normalization condition
	\begin{equation}
	\label{norm}
	p_0 + 3p_1 + 3p_2 = 1.
	\end{equation}
	For each column $j$, we randomly select three distinct rows $\{ i_{j,1}, i_{j,2}, i_{j,3} \}$ uniformly. We then randomly assign the elements $(q_{i_{j,1}j}, q_{i_{j,2}j}, q_{i_{j,3}j})$ with an element from the following set
	\begin{equation*}
	\{ (q_1,q_2,q_3)\in \mbb{R}^3 \,\, \big\lvert \,\, |q_1| = |q_2| = |q_3| = +1 \}\, \big/ \, (-1,-1,-1),
	\end{equation*}
	with each assignment associated with the sampling probability given as follows
	\begin{equation*}
	\begin{split}
	&p_0: \, q_1+q_2+q_3 = 3, \\
	&p_1: \, q_1+q_2+q_3 = 1, \\
	&p_2: \, q_1+q_2+q_3 = -1.
	\end{split}
	\end{equation*}
	We then assign all other elements in column $j$ to zero.
	\label{con_plant}
\end{definition}

\begin{remark}
	To explain this construction in simple terms, we can consider a 3-SAT Boolean formula where each clause is independently generated through the inclusion of 3 randomly chosen Boolean variables out of the $n$ total variables without replacement. The negations of the Boolean variables in the clause are randomly assigned such that there is a probability $p_0$ that all variables appear without negation; there is a probability $3p_1$ that only one variable is negated (the prefactor of 3 is to account for the fact that there are 3 possible variables to negate); and there is a probability $3p_2$ that two variables are negated (the prefactor of 3 arises similarly). 
\end{remark}

\subsubsection{Clause Distribution Control Instances}\label{barthelsection}

We now consider a class of hard instances~\cite{barthel2002} that is generated based on the method described in Definition \ref{con_plant}. In particular,  the generation method is restricted in the presence of a new constraints on the parameters $\{\alpha_r, p_0, p_1, p_2\}$, in addition to the normalization condition given in Eq.~(\ref{norm}). This gives us only $4-2 = 2$ degrees of freedom in the selection of the parameters, $\alpha_r$ and $p_0$.

\begin{definition}[Clause Distribution Control Instances]
	A Clause Distribution Control~\footnote{While the method can be generalized, for example as in Ref. \cite{qhid}, we report the method outlined in Ref. \cite{barthel2002}.} (CDC) instance generated with the parameters $\alpha_r$ and $p_0$ is an instance whose polarity matrix $Q$ is randomly generated by the following constraints
	\begin{equation}
	\label{bar_trip}
	\alpha_r > 4.25, \qquad 0.077 < p_0 < 0.25, \qquad p_1 = \frac{1-4p_0}{6}, \qquad p_2 = \frac{1+2p_0}{6},
	\end{equation}
	based on the method given in Definition \ref{con_plant}. 
\end{definition}

\begin{remark}
	It has been claimed that this class of instances is difficult for local-search procedures~\cite{barthel2002}, though, it has been shown that the difficulty does not persist for  some upper limit on $\alpha_r$ that depends on the problem size, $n$~\cite{bulatov2015phase}.The results from the Walk-SAT algorithm confirm the instances generated for numerical simulation are difficult in that the showcase exponential scalability.\\
	
	The reason for enforcing the condition $p_0 < \frac{1}{4}$ is twofold. First, $p_0$ is restricted so that parameter $p_1$ is non-negative, as it represents a probability.
	Second, the instances created with $p_0=1/4$ are known to be solvable in polynomial time using a global algorithm~\cite{barthel2002}.
	It can be easily verified that the probabilities given in Eq.~(\ref{bar_trip}) satisfy the normalization condition (Eq.~(\ref{norm})) in addition to the following condition
	\begin{equation}
	\label{bar_cond}
	p_0 + p_1 - p_2 = 0
	\end{equation}
	If the above constraint is satisfied, then it can be shown that a greedy local-search SAT solver initialized with a random assignment of variables will not be biased towards the planted solution~\cite{barthel2002}. In the language of statistical mechanics, we say that the instance is equivalent to an instance of a disordered diluted spin glass with couplings up to three spins~\cite{zecchina1997}. The Hamiltonian of this diluted spin glass can be written as
	\begin{equation}
	\label{ham}
	H = -\sum_i H_i S_i - \sum_{ij} T_{ij}S_iS_j - \sum_{ijk}S_iS_jS_k,
	\end{equation}
	which is equivalent to the global energy as defined in Eq.~(\ref{g_ham}). If Eq. (\ref{bar_cond}) is enforced, then the average of the local field over the disorder $\overline{H}_i$ is zero for all spins, so there is typically no direct bias towards the planted state $\mbf{S} = \mbf{+1}$. An extended discussion of the CDC instances can be found in literature on the statistical mechanics of Boolean satisfiability problems~\cite{NOAP}.
\end{remark}

\subsubsection{Solution Backbone and Cluster}
\label{backbone}

As briefly addressed in the remark of Lemma \ref{one}, planting the $\mbf{+1}$ solution in an instance does not forbid the existence of additional solutions.
In fact, multiple solutions may exist, however, their locations in phase space, with respect to one another, and the similarity of solutions are generally what determine the difficultly of an instance.
In most cases, some solutions will overlap non-trivially, meaning that their assignments will coincide for a certain number of variables. For instances admitting overlapping solutions, there are two concepts (occurring non-exclusively) important for analytic studies. \\

For the first concept, given a solution to an instance, we can define a solution {\bf cluster} as the subset of all solutions that can be assigned from the given solution via a sequence of single spin flips (Boolean variable negation)~\cite{zoltan2011}. Note, after each flip the assignment must remain a solution to be considered part of the cluster. While the clustering of solutions into one big cluster may intuitively seem like a more difficult instance, knowing only one solution cluster exists is not enough information to categorize an instance as more difficult than others.
The second concept will give additional information about the difficulty. Given the set of all solutions, we define the {\bf backbone} to be the number of variables that appear with only one parity in all solutions~\cite{NOAP}. 
In other words, for the SAT solver to find a solution, it is necessary for the backbone to be assigned correctly\footnote{In the case of the CDC instances that we use, the fashion in which the backbone appears as the clause density is increased is dictated directly by the parameter $p_0$. More particularly, this CDC parameter induces a phase transition from a continuous appearance of a backbone to a discontinuous appearance of a backbone \cite{barthel2002}.}. In general, the emergence of a backbone in a 3-SAT instance results in variables that must be assigned to a particular value to find any solution (an inherent difficulty), however, there can still exist a local field that can guide a greedy local-search SAT solver to the solution.\\

To understand why the CDC instances (planted solution) are difficult, it aids understanding to describe the solution cluster distribution in uniform random 3-SAT (no guaranteed solution). Using the replica symmetry approximation~\cite{zecchina1996,zecchina1997,monasson1999determining}, a variational approach accounting for replica-symmetry breaking~\cite{biroli2000variational}, and the cavity method~\cite{mezard2002sid,mezard2002random} from statistical mechanics,
it was shown that the 3-SAT problem undergoes phase transitions as clause density is increased~\footnote{See Ch. 7 of Ref. \cite{NOAP} for a self-contained account of the following results.}. For $\alpha_r < \alpha_d \simeq 3.92$, there is one large solution cluster, and solutions are relatively easy to find.
At $\alpha_d$ the large solution cluster breaks into an exponential amount of solution clusters, with an exponential amount of solutions within each. These clusters are far from each other in phase space, and their frequency diminishes as
$\alpha_r \rightarrow \alpha_c \simeq 4.267$, until only one solution cluster remains. That is, the solutions become less frequent as $\alpha_c$ (the complexity peak) is approached, until no solutions exist (the SAT/UNSAT transition)~\cite{NOAP}.\\

At $p_0=0.25$, for $\alpha_r<4.27$, there is no difference between the CDC class and  uniform random 3-SAT, with the solution entropy and clustering transition, $\alpha_d$, being the same~\cite{NOAP}. However, the SAT/UNSAT transition at $\alpha_c\simeq4.27$ is obviously absent, being that the solution is always planted.
Now, the instance class undergoes a first-order ferromagnetic transition at $\alpha_c$, resulting in only one solution cluster remaining. The first-order transition is more pronounced for $0.077 < p_0 < 0.25$, and there is a discontinuous appearance of a backbone. (For $p_0 < 0.077$, no backbone appears.) At $\alpha_c\simeq4.27$, the paramagnetic phase (many solution clusters) transitions to a ferromagnetic phase (one cluster containing the planted solution) with the discontinuous appearance of a backbone\footnote{The reader may notice the transition is reported as $\alpha_c\simeq4.27$, but Def.  \ref{bar_trip} has $\alpha_r>4.25$. To avoid any discrepancy, the smallest ratio used in numerical simulations is $\alpha_r=4.3$.}.\\

The approximate backbone size for CDC instances range from $0.72n$ at $p_0\simeq0.077$ to $0.94n$ at $p_0=0.25$~\cite{barthel2002}. Therefore, with all factors considered above, $p_0$ serves as a measure of difficulty for the CDC instances.\\

In this material, we base our focus on the study of the dynamical properties of our DMM by defining solution planes on hyperfaces of the voltage hypercube \cite{zoltan2011}.
When the solution vector is on a hyperface that corresponds to a solution plane (see \ref{sec_cor}), the voltage dynamics are near a branch of a solution cluster, effectively solving the CDC instance. To further associate the concepts, when a solution is found on a vertex of the hypercube, the solution cluster can be traversed by traveling along the hyperedges of the hypercube that connect to other solution vertices.

\section{Lipschitz Continuity}
\label{lip}

Before we present the equations governing the dynamics of our memcomputing solver in Section \ref{sec_mem}, it is necessary to first introduce a few formal mathematical arguments that will help establish the existence and uniqueness of the solution trajectory under an ordinary differential equation (ODE). For instance, the requirement for the existence and uniqueness of a local solution to a first order autonomous ODE is the Lipschitz continuity of the flow field~\cite{ode}. We begin by  formally defining Lipschitz continuity.

\begin{definition}[Lipscthiz Continuity]
	Let $X$ and $Y$ be two metric spaces. A function $f: X \mapsto Y$ is Lipschitz continuous if there is a real constant $K \geq 0$ such that
	\begin{equation*}
	\forall x_1,x_2 \in X, \quad d_Y(f(x_1),f(x_2)) \leq K d_X(x_1,x_2),
	\end{equation*}
	where $d_X$ and $d_Y$ denote the metrics on $X$ and $Y$ respectively.
\end{definition}

\begin{remark}
	This definition can be easily specialized to a vector field $V: \mbb{R}^n \mapsto \mbb{R}^n$.
\end{remark}

\begin{theorem}[Picard–Lindelöf theorem]
	\label{pic}
	Given a Lipschitz continuous vector field $V: \mbb{R}^n \mapsto \mbb{R}^n$, the classical solution $\mbf{x}(\mbf{x_0}, t)$ to the first order autonomous ODE, $\dot{\mbf{x}}(t) = V(\mbf{x})$, exists and is unique for $\forall t \in \mbb{R}$. 
\end{theorem}

Our dynamics are governed by a high dimensional vector flow field, $F: \mbb{R}^n \mapsto \mbb{R}^n$. To study the Lipschitz continuity of the vector field $F$, we simply study the Lipschitz continuity of the field components in the quotient spaces instead, by the following lemma.

\begin{lemma}
	\label{component}
	Given a metric space $X$, and a product metric space $Y = Y_1 \times Y_2 \times ... \times Y_n$ equipped with a $p$-product metric, where $p \in (0,+\infty)$, let $f_i: X \mapsto Y_i$ be a mapping and $f: X \mapsto Y$ be defined as $f(x) = \big(f_1(x),f_2(x),...,f_n(x)\big)$. Then $f$ is Lipschitz continuous if and only if $f_i$ is Lipschitz continuous for $\forall i \in [[1,n]]$. 
\end{lemma}

\begin{proof}
	We first assume that $f_i$ is Lipschitz continuous $\forall i$, with its Lipschitz constant being $K_i$. Then $\forall x_1,x_2 \in X$, we have
	\begin{equation*}
	\begin{split}
	d_Y\big(f(x_1),f(x_2)\big) 
	=& \Big( \sum_{i=1}^n d_{Y_i}\big( f_i(x_1),f_i(x_2) \big)^p \Big)^{1/p} \\
	\leq & \Big( \sum_{i=1}^n K_i^p \, d_X(x_1,x_2)^p \Big)^{1/p} \\
	\leq & \Big( \max_i(K_i)^p\sum_{i=1}^n \, d_X(x_1,x_2)^p \Big)^{1/p} \\
	= & \big[ \max_i(K_i) n^{1/p} \big] d_X(x_1,x_2).
	\end{split}
	\end{equation*}
	In other words, the Lipschitz constant for $f$ is simply $\max_i(K_i)n^{1/p}$ so $f$ is Lipschitz continuous. \\
	
	Now, we assume that $f_{i'}$ is not Lipschitz continuous for some $i'$. Then $\exists x_1,x_2 \in X$ such that
	\begin{equation*}
	\forall K \geq 0, \quad d_{Y_i'}\big(f_{i'}(x_1), f_{i'}(x_2)\big) > K d_X(x_1,x_2).
	\end{equation*}
	We then have
	\begin{equation*}
	\begin{split}
	d_Y\big(f(x_1),f(x_2)\big) 
	=& \Big( \sum_{i=1}^n d_{Y_i}\big( f_i(x_1),f_i(x_2) \big)^p \Big)^{1/p} \\
	\geq & \Big( d_{Y_{i'}} \big( f_{i'}(x_1), f_{i'}(x_2) \big)^p \Big)^{1/p} \\
	> & K d_X(x_1,x_2),
	\end{split}
	\end{equation*}
	meaning that $f$ is also not Lipschitz continuous.
\end{proof}

For our work, we are also interested in the Lipschitz continuity of a vector field that is projected onto another vector field. In particular, in definition \ref{patch}, we show how a vector field can be projected onto a regular surface. In the following lemma, we give the condition for this ``projected" vector field to be Lipschitz continuous. From here on, we shall use the notation $\braket{a,b}$ to denote the inner product of vectors $a$ and $b$.

\begin{lemma}[Continuity of Projection]
	\label{proj}
	Let $\proj_{\mbf{v}}: \mbb{R}^n \mapsto \mbb{R}^n$ be the projection mapping defined as
	\begin{equation*}
	\proj_{\mbf{v}}(\mbf{u}) = \braket{\mbf{u}, \mbf{v}} \frac{\mbf{v}}{||\mbf{v}||^2} = \braket{\mbf{u},\mbf{\hat{v}}}\mbf{\hat{v}}.
	\end{equation*}
	Let $X$ be a metric space. Let $f_1: X \mapsto \mbb{R}^n$ be some Lipschitz continuous function bounded from below by $\exists m > 0$ in norm, and let $f_2: X \mapsto \mbb{R}^n$ be some Lipschitz continuous function bounded from above by $\exists M > 0$. Then $f(x) = \proj_{f_1(x)}\big( f_2(x) \big)$ is Lipschitz continuous.
\end{lemma}

\begin{proof}
	$\forall x_1,x_2 \in X$, we have
	\begin{equation*}
	|| f_1(x_2) - f_1(x_1) || \leq K_1 d(x_1,x_2);
	\qquad
	|| f_2(x_2) - f_2(x_1) || \leq K_2 d(x_1,x_2),
	\end{equation*}
	for some constants $K_1,K_2 > 0$. For the sake of simplicity, we denote $f_1 = f_1(x_1)$, $f_2 = f_2(x_1)$, $f_1' = f_1(x_2)$, and $f_2' = f_2(x_2)$. Then we can write 
	\begin{equation}\label{normf}
	\begin{split}
	|| f(x_2) - f(x_1) || 
	= & \, || \proj_{f_1'}(f_2') - \proj_{f_1}(f_2) || \\
	= & \, \Norm \braket{f_2',\hat{f_1'}}\hat{f_1'} - \braket{f_2,\hat{f_1}}\hat{f_1} \Norm \\
	= & \, \Norm \braket{f_2'-f_2, \hat{f_1}}\hat{f_1} + \big( \braket{f_2',\hat{f_1'}}\hat{f_1'} - \braket{f_2',\hat{f_1}}\hat{f_1} \big) \Norm \\
	\leq & \, \Norm \braket{f_2'-f_2, \hat{f_1}} \Norm + \Norm \big( \braket{f_2',\hat{f_1'}}\hat{f_1'} - \braket{f_2',\hat{f_1}}\hat{f_1} \big) \Norm.
	\end{split}
	\end{equation}
	Note that the first term is bounded as follows
	\begin{equation*}
	\Norm \braket{f_2'-f_2, \hat{f_1}} \Norm 
	\leq \, \Norm f_2'-f_2 \Norm\,\Norm \hat{f_1} \Norm
	\leq \, K_2 d(x_1,x_2).
	\end{equation*}
	To bound the second term, it is convenient to denote $\phi = \arccos\big( \braket{\hat{f_1},\hat{f_1'}} \big)$, then it can be easily shown that
	\begin{equation*}
	\phi \leq
	\begin{cases}
	\arcsin\big( \frac{K_1 d(x_1,x_2)}{m} \big) \quad &\text{if} \quad K_1 d(x_1,x_2) \leq m, \\
	\pi & \text{otherwise}.
	\end{cases}
	\end{equation*}
	This means that $\phi \leq \frac{2 K_1 d(x_1,x_2)}{m}$. We then see that the second term in the last line of Eq.~(\ref{normf}) is bounded as follows
	\begin{equation*}
	\Norm \big( \braket{f_2',\hat{f_1'}}\hat{f_1'} - \braket{f_2',\hat{f_1}}\hat{f_1} \big) \Norm \leq M\phi \leq \frac{2K_1 M}{m}d(x_1,x_2).
	\end{equation*}
	Therefore, we have
	\begin{equation*}
	|| f(x_2) - f(x_1) || \leq \Big( \frac{2K_1M}{m} + K_2 \Big) d(x_1,x_2),
	\end{equation*}
	so $f$ is Lipschitz continuous. 
\end{proof}

\begin{remark}
	We use this lemma to study the Lipschitz continuity of a vector flow field projected onto some regular boundary, which allows for the existence of a solution at the boundary that follows the projected flow field almost everywhere. We formalize this discussion in Section \ref{sol_bound}. \\
\end{remark}

In Section \ref{reduce}, techniques of linear algebra are used extensively to relate the dynamics of the voltage flow field to the trajectory of the auxiliary variable, so to conclude this Section, we provide the following useful lemma in anticipation. 

\begin{lemma}[Continuity of Linear Maps]
	\label{mat_con}
	Given a metric space $X$, let $v: X \mapsto \mbb{R}^m$ be a Lipschitz continuous map with bounded image, and let $M: X \mapsto \mbb{R}^{n \times m}$ be another Lipschitz continuous map with bounded image. Then $M(\mbf{x})\cdot v(\mbf{x})$ is Lipschitz continuous, where $v$ is treated as a column vector, $M$ is treated as an $n\times m$ matrix.
\end{lemma}

\begin{proof}
	From lemma \ref{component}, we see that every component of $v$ and every element of $M$ must be Lipscthiz continuous and bounded. Then every component of $M \cdot v$ is Lipschitz continuous and bounded as well, as the addition and multiplication of bounded Lipschitz continuous functions are also bounded Lipschitz continuous. Therefore, using lemma \ref{component} again in reverse, we see that $M\cdot v$ must be Lipschitz continuous.
\end{proof}

\section{Existence and Uniqueness of Caratheodory Solution}
\label{cara}

As discussed in the main text (see also Eq.~(\ref{mem}) in Section~\ref{sec_mem}), the flow field we have chosen to govern the dynamics of our memcomputing machines are discontinuous. This is due to the presence of the $\min$ function and the explicit enforcement of the bounds on the dynamics. Therefore, the existence and uniqueness of a classical solution to the ODEs is not guaranteed. We then require the construction of a Caratheodory solution, and show that such construction is well-defined and unique. A Caratheodory solution is formally defined as follows:

\begin{definition}[Caratheodory Solution]
	\label{def_cara}
	Let $V: \mathbb{R}^n\mapsto\mathbb{R}^n$, then a solution to the ODE $\dot{x} = V(x)$ is a Catheodory solution if it satisfies
	\begin{equation*}
	x(t) = x(t_0) + \int_{t_0}^t f\big(x(s)\big)\,ds, \quad \forall t>t_0,
	\end{equation*}
	where $\int$ denotes the Lebesgue integral.
\end{definition}

\begin{remark}
	An equivalent definition states that the Caratheodory solution follows the vector field everywhere along the solution trajectory except for a subset of measure zero~\cite{discon}. \\
\end{remark}

We construct the Caratheodory solution in a way such that the analytic trajectory is closely mimicked by the dynamics governed by numerical simulations. In particular, the memory dynamics are governed by a discontinuous flow field, where occasionally the discretized trajectories will oscillate at certain hyper-planes of discontinuities until they ``escape" the planes when the fields become sufficiently regular to allow so. The analytic construction of the Caratheodory solution is given such that the oscillatory dynamics at these hyperplanes are accounted for in a similar fashion. An extended discussion of how the analytic trajectory is simulated effectively by forward Euler is given in Section \ref{plane}.

\subsection{Patching Vector Fields}

Before we discuss the construction of Caratheodory solutions, we first formally define the class of discontinuous vector fields of interest referred to as the {\it patchy} vector fields. As the name suggests, the vector field is the result of patching together two different vector fields in a way such that a Caratheodory solution is admitted. For ease of analysis, we first assume some regularity condition on the boundary at which the fields are patched together.

\begin{definition}[Regular Domain]
	\label{regular}
	Let $\Omega \subset \mbb{R}^n$ a domain in Euclidean space. The domain is said to be {\bf regular} if it is bounded, with its boundary $\partial\Omega$ being $C^{\infty}$ diffeomorphic to an $n-1$ sphere. 
\end{definition}

\begin{remark}
	A regular domain is equipped with an orientable boundary, where the unit normal vector $\mbf{n}$ can be defined at every point to be pointing towards the exterior of the domain. From here on, we shall use $\interior(\Omega)$ to denote the interior of the domain, which is simply itself if it is open in $\mbb{R}^n$. And we use $\Omega^c$ to denote its complement in $\mbb{R}^n$, and $\ext(\Omega) = \Omega^c/\partial\Omega$ to denote the exterior. \\
\end{remark}

For any vector field with domain $\partial\Omega$, there is a unique ``projection" of the field onto the boundary, such that the projection is in the tangent bundle generated by $\partial\Omega$.

\begin{definition}
	For $\mbf{v},\mbf{w} \in \mbb{R}^n$, we denote the parallel and orthogonal components of $\mbf{v}$ with respect to $\mbf{w}$ as follows
	\begin{equation*}
	\mbf{v}_{\mbf{w},\parallel} = \braket{\mbf{v},\mbf{\hat{w}}} \hat{\mbf{w}}
	\qquad
	\mbf{v}_{\mbf{w}, \perp} = \mbf{v} - \mbf{v}_{\mbf{w}, \parallel}
	\end{equation*}
	where $\hat{\mbf{w}} = \frac{\mbf{w}}{|\mbf{w}|}$.
\end{definition}

\begin{definition}[Decomposition at Boundary]
	\label{deco}
	Let $\Omega \subset \mbb{R}^n$ be a regular domain, and let $\mbf{n}(\mbf{x})$ be the unit normal vector of $\Omega$ at $\mbf{x} \in \partial\Omega$. Let $V: \mbb{R}^n \mapsto \mbb{R}^n$ be some vector field, then we denote the decomposition of the vector field at the boundary, $V_{\partial\Omega,\parallel}: \partial\Omega \mapsto \mbb{R}^n$ and $V_{\partial\Omega,\perp}: \partial\Omega \mapsto \mbb{R}^n$, as follows
	\begin{equation*}
	V_{\partial\Omega,\parallel}(\mbf{x}) = V(\mbf{x})_{\mbf{n}(\mbf{x}),\perp}
	\qquad
	V_{\partial\Omega,\perp}(\mbf{x}) = V(\mbf{x})_{\mbf{n}(\mbf{x}),\parallel}
	\end{equation*}
	$\forall \mbf{x} \in \partial\Omega$. 
\end{definition}

\begin{lemma}
	\label{proj_con}
	If $V: \partial\Omega \mapsto \mbb{R}^n$ is bounded above and Lipscthiz continuous, then $V_{\partial\Omega,\parallel}$ and $V_{\partial\Omega,\perp}$ are Lipschitz continuous as well.
\end{lemma}

\begin{proof}
	Note that since $V$ is bounded, $\exists M$, $||V(\mbf{x})|| \leq M$, $\forall \mbf{x}\in \partial\Omega$. Furthermore, it is clear that $\mbf{n}(\mbf{x})$ is bounded from below as $||\mbf{n}(\mbf{x})|| = 1$ by definition of a unit vector. It can also be easily shown that $\mbf{n}(\mbf{x})$ is Lipschitz continuous due to the regularity of $\Omega$. Therefore, by using lemma \ref{proj}, we see that $V(\mbf{x})_{\partial\Omega,\perp}$ is Lipschitz continuous, which implies that $V(\mbf{x})_{\partial\Omega,\parallel} = V(\mbf{x}) - V(\mbf{x})_{\partial\Omega,\perp}$ is also Lipschitz continuous.
\end{proof}

Since a projected vector field is Lipscthiz continuous, it admits a classical solution on the boundary (see Lemma \ref{pic}). However, at some point the trajectory has to escape the boundary once the field outside the boundary admits it. This escape condition depends on the direction of the field relative to the curvature of the boundary (see Proposition \ref{exit}). It is difficult to give a general definition of curvature for high dimensional hyper-surfaces. However, the definition of a {\it directional} curvature is relatively straightforward. 

\begin{definition}[Directional Curvature]
	\label{curve}
	Let $\Omega \subset \mbb{R}^n$ be a regular domain. Given a point in the boundary $\mbf{x_0}\in \partial\Omega$ and a vector field $V: \mbb{R}^n \mapsto \mbb{R}^n$. Let $\bs{\gamma}(t) \in \partial\Omega$ be a trajectory such that $\exists t_{\epsilon}$,
	\begin{equation*}
	\bs{\gamma}(0) = \mbf{x_0}, \qquad \dot{\bs{\gamma}}(t) = V_{\partial\Omega,||}(\bs{\gamma}(t)), \quad \forall t\in [0,t_{\epsilon}).
	\end{equation*}
	We then define the $m$-th order {\bf directional curvature} at point $\mbf{x_0}$ with respect to $V$ as
	\begin{equation*}
	\kappa_{V}^{(m)}(\mbf{x_0}) = \Big( \frac{d^m}{dt^m} \mbf{n}(\mbf{\gamma}(t)) \Big) \,\cdot\, \hat{V}_{\partial\Omega,||}(\mbf{x_0}),
	\end{equation*}
	for $m \geq 0$. For notational compactness, we define
	\begin{equation*}
	m_0(\mbf{x_0}) = \inf\{m \cond \kappa_V^{(m)}(\mbf{x_0}) \neq 0\}, \qquad \kappa'_V(\mbf{x_0}) = \kappa_{V}^{m_0(\mbf{x_0})},
	\end{equation*}
	as the lowest order curvature that does not vanish.
\end{definition}

\begin{remark}
	Visually, the sign of $\kappa$ is an indicator of whether the boundary curves outward or inward at point $\mbf{x_0}$ along the projected direction of $\mbf{v}$, and this informs whether the solution should exit to the interior $\Omega$ or the exterior $\Omega^c/\partial\Omega$ (see Theorem \ref{construct}). It is clear that $\kappa^{(m)}_{\mbf{v}}(\mbf{p})$ is well defined and Lipschitz continuous to all orders due to the regularity of $\Omega$. \\
\end{remark}

This definition of the curvature informs the patching operation of two vector fields at the boundary.

\begin{definition}
	\label{bin}
	Let $\Omega \subset \mbb{R}^n$ be a regular domain, and $V: \mbb{R}^n \mapsto \mbb{R}^n$ be some vector field. For $\mbf{x} \in \partial\Omega$, we define the function $\psi_{\partial\Omega,V}: \partial\Omega \mapsto \{0,1\}$ as follows
	\begin{equation*}
	\psi_{\partial\Omega,V}(\mbf{x}) = 
	\begin{cases}
	1 \quad \text{if} \quad \kappa'_V(\mbf{x}) \leq 0, \\
	0 \quad \textrm{otherwise}.
	\end{cases}
	\end{equation*}
	Similarly, we define the function $\phi_{\partial\Omega,V}: \partial\Omega \mapsto \{0,1\}$ as follows
	\begin{equation*}
	\phi_{\partial\Omega,V}(\mbf{x}) = 
	\begin{cases}
	1 \quad \text{if} \quad \kappa'_V(\mbf{x}) \geq 0, \\
	0 \quad \textrm{otherwise}.
	\end{cases}
	\end{equation*}
\end{definition}

\begin{remark}
	Note that the definition of $\psi$ and $\phi$ is symmetric with respect to the exchange of the interior and exterior of the domain $\Omega$.
\end{remark}

\begin{definition}[Patching]
	\label{patch}
	Let $\Omega \subset \mbb{R}^n$ be a smooth open domain, and $V,W: \mbb{R}^n \mapsto \mbb{R}^n$ be two distinct vector fields. We define the {\bf patching} of the two vector fields with respect to domain $\Omega$ as
	\begin{equation*}
	\begin{split}
	& \mcr{P}_{\Omega}(V,W)(\mbf{x}) \\
	= &
	\begin{cases}
	V(\mbf{x}) \quad &\text{if} \quad \mbf{x}\in\Omega, \\
	W(\mbf{x}) \quad &\text{if} \quad \mbf{x}\in\ext(\Omega), \\
	V(\mbf{x})_{\partial\Omega,\parallel} + W(\mbf{x})_{\partial\Omega,\parallel} 
	+
	\psi_{\partial\Omega,V}(\mbf{x})V(\mbf{x})_{\partial\Omega,\perp} + \phi_{\partial\Omega,W}(\mbf{x})W(\mbf{x})_{\partial\Omega,\perp} \quad 
	&\text{if} \quad \mbf{x}\in\partial\Omega.
	\end{cases}
	\end{split}
	\end{equation*}
\end{definition}

\begin{remark}
	Note that the vector field $\mcr{P}_{\Omega}(V,W)$ is piecewise Lipschitz continuous, with its discontinuity being at the boundary $\partial \Omega$. We can refer to $V$ as the {\it interior} vector field and $W$ as the {\it exterior} vector field. Visually, we can view the patched field at the boundary $\partial\Omega$ as some form of ``projection" of the interior field $V$ and exterior field $W$.
\end{remark}

\subsection{Solution in the Boundary}
\label{sol_bound}

It is clear that the patched field $\mcr{P}_{\Omega}(V,W)(\mbf{p})$ is Lipschitz continuous in $\Omega$ and $\ext(\Omega)$ separately. This implies that a classical solution to the ODE $\dot{\mbf{x}} = \mcr{P}_{\Omega}(V,W)(\mbf{x})$ with initial value $\mbf{x}_0 \in \Omega$ exists up to the boundary $\partial \Omega$ (and similarly for $\mbf{x}_0 \in \ext(\Omega)$). Naturally, we also have to discuss the existence of a classical solution with $\mbf{x}_0 \in \partial\Omega$. To do so, we first make a preliminary definition that specifies two important subsets of $\partial\Omega$, relative to which we attach the start- and end-points of the solution segments. 

\begin{definition}
	\label{D}
	Given a regular domain $\Omega\subset \mbb{R}^n$ and two vector fields $V,W: \mbb{R}^n\mapsto \mbb{R}^n$, we denote $D_1 = \{ \mbf{x}\in \partial\Omega \,|\, \psi_{\partial\Omega, V}(\mbf{p}) = 0 \}$ and $D_2 = \{ \mbf{x}\in \partial\Omega \,|\, \phi_{\partial\Omega,W}(\mbf{p}) = 0 \}$. 
\end{definition}

\begin{remark}
	Visually, $D_1$ describes a region of the boundary where the interior field points outward, and $D_2$ describes a region of the boundary where the exterior field points inward. This gives rise to an irregular region $D_1 \cap D_2$ where the two fields ``collide" at the boundary, which generates a Lipschitz continuous field that admits a classical solution in the boundary. 
\end{remark}

\begin{lemma}[Continuity in Boundary]
	Given a regular domain $\Omega\subset \mbb{R}^n$ and two bounded Lipschitz continuous vector fields $V,W: \mbb{R}^n\mapsto \mbb{R}^n$, $D = D_1 \cap D_2$ is open with respect to $\partial\Omega$. Furthermore, the vector field $\mcr{P}_{\Omega}(V,W)$ is Lipschitz continuous in $D$.
\end{lemma}

\begin{proof}
	From the definitions of $\kappa$ and $\psi$ (see definitions \ref{curve} and \ref{bin}), we can express $\partial\Omega/D_1$ as the following intersection of countably many sets
	\begin{equation*}
	\begin{split}
	\partial\Omega/D_1 = 
	& \{ \mbf{x}\in\partial\Omega \,|\, \kappa^{(0)}_{V}(\mbf{x}) \leq 0 \} \,\cap \\
	& \bigcap_{m=1}^{\infty} \{ \mbf{x}\in\partial\Omega \,|\, \kappa^{(m-1)}_{V}(\mbf{p}) = 0 \,\land\, 
	\kappa^{(m)}_{V}(\mbf{p}) \leq 0 \}.
	\end{split}
	\end{equation*}
	We first assume that $\partial\Omega/D_1$ is non-empty, otherwise $D_1 = \partial\Omega = D$ is clearly open. Note from corollary \ref{proj_con} that $V_{\partial\Omega, \perp}$ is Lipschitz continuous in $\partial\Omega$, so the $\kappa^{(0)}_{V}(\mbf{x}) = \braket{V(\mbf{x}), \mbf{n}(\mbf{x})}$ is a continuous mapping from $\partial\Omega$ to $\mbb{R}$. Furthermore, $\partial\Omega$ is compact, so its image must also be compact, with the infinum denoted as $-C = \inf_{\mbf{x}\in\partial\Omega} \{ \kappa^{(0)}_{V}(\mbf{x}) \} \leq 0$. This means that $\{ \mbf{x}\in\partial\Omega \,|\, \kappa^{(0)}_{V}(\mbf{x}) \leq 0 \}$ is the preimage of the closed set $[-C,0]$ under a continuous mapping, so it also must be a closed set itself. A similar proof applies for the $m>1$ cases. Therefore, $\partial\Omega/D_1$ is the intersection of countably many closed subsets of $\mbb{R}$, so it must also be closed, which implies that $D_1$ is open. We can similarly show that $D_2$ is also open, so $D$ being the intersection of two open sets is open as well. \\
	
	To show that the field $\mcr{P}_{\Omega}(V,W)$ is Lipschitz continuous in $D$, we first begin by noting that $\psi_{\partial\Omega,V}(\mbf{x}) = \phi_{\partial\Omega,W}(\mbf{x}) = 0$, $\forall \mbf{x} \in D$, which follows directly from the definition of $D$ and definition \ref{bin}. Then from definition \ref{patch}, we see that $\mcr{P}_{\Omega}(V,W)(\mbf{x}) = V(\mbf{x})_{\partial\Omega,\parallel} + W(\mbf{x})_{\partial\Omega,\parallel}$, $\forall \mbf{x}\in D$. From corollary \ref{proj_con}, we see that $V_{\partial\Omega,\parallel}$ and $W_{\partial\Omega,\parallel}$ are Lipschitz continuous vector fields in $D$, then $\mcr{P}(V,W)$ is also Lipschitz continuous.
\end{proof}

\begin{corollary}[Solution in Boundary]
	\label{contain}
	Given a regular domain $\Omega\subset \mbb{R}^n$ and two bounded Lipschitz continuous vector fields $V,W: \mbb{R}^n\mapsto \mbb{R}^n$, let $U(\mbf{x})= V(\mbf{x})_{\partial\Omega,\parallel} + W(\mbf{x})_{\partial\Omega,\parallel}$, there is a unique classical solution $\mbf{x}(t,\mbf{x_0})$ to the ODE $\dot{\mbf{x}} = U(\mbf{x})$ for any $\mbf{x_0} \in \partial\Omega$. 
\end{corollary}

\begin{proof}
	We here provide a brief proof sketch. We begin by treating $\partial\Omega$ as a $n-1$ dimensional differentiable manifold (equipped with the pullback of the Euclidean metric by the natural embedding $\partial\Omega \mapsto \mbb{R}^n$), then $U: \partial\Omega \mapsto T\partial\Omega$ is clearly Lipschitz continuous on the manifold. This implies that there is a unique classical solution to the ODE $\dot{\mbf{x}} = U(\mbf{x})$ on the manifold $\partial\Omega$ (under some suitable connection). 
\end{proof}

\begin{proposition}[Containment in Boundary]
	\label{xtD}
	Given a regular domain $\Omega\subset \mbb{R}^n$ and two bounded Lipschitz continuous vector fields $V,W: \mbb{R}^n\mapsto \mbb{R}^n$, denote $\mbf{x}(t,\mbf{x_0})$ as the classical solution to the ODE $\dot{\mbf{x}} = U(\mbf{x})$ with $\mbf{x_0}\in D$, where $U$ is defined in corollary \ref{contain}. If we restrict the solution to $t \in [0,t_0)$, where $t_0 = \inf\{ t\geq 0 \,|\, \mbf{x}(t,\mbf{x_0}) \in \partial\Omega \}$, with $\partial D$ being the boundary of $D$ with respect to $\partial\Omega$, then $\mbf{x}(t,\mbf{x_0})$ is a classical solution to the ODE $\dot{\mbf{x}} = \mcr{P}(V,W)(\mbf{x})$.
\end{proposition}

\begin{proof}
	From lemma \ref{D}, we see that $\mcr{P}_{\Omega}(V,W)(\mbf{x}) = U(\mbf{x})$, $\forall \mbf{x} \in D$. Since $\mbf{x}(t) \in D$, $\forall t \in [0,t_0)$, we have $\dot{\mbf{x}}(t) = U(\mbf{x}(t)) = \mcr{P}(V,W)(\mbf{x}(t))$, $\forall t \in [0, t_0)$. Furthermore, $\mbf{x_0} \notin \partial D$ as $D$ is open, so $t_0 \neq 0$.
\end{proof}

To conclude, we have shown that the patched field admits a classical solution in $D = D_1 \cap D_2$ at least up to some positive time $t_0$.

\subsection{Solution in the Domain}
\label{dom}

In the previous Section, we have shown how a solution segment can be constructed in the boundary $\partial\Omega$. In this subsection, we focus on the construction of a solution in the interior $\Omega$ and exterior $\ext(\Omega)$ to the ODE $\dot{\mbf{x}}(t) = \mcr{P}(V,W)(\mbf{x})$. WLOG, we can assume that the initial point is in the interior (see the remark of Definition \ref{bin}). \\

There are three possibilities for the evolution of the trajectory. First, the trajectory never leaves the interior $\Omega$. Second, the trajectory escapes to the exterior $\ext(\Omega)$, intersecting the boundary $\partial\Omega$ as required by the Jordan-Brouwer separation theorem~\cite{topology}. Finally, the trajectory hits the boundary $\partial\Omega$ and ``returns" back to the interior $\Omega$. \\

Clearly, in the first case, the trajectory is simply the classical solution to the ODE, $\dot{\mbf{x}} = V(\mbf{x})$, and in the last two non-trivial cases, the trajectory reaches the boundary $\partial\Omega$ at some point. We first begin by noting that if the trajectory were to reach the boundary, it must enter $\partial\Omega$ through its subset $\overline{D_1}$. 

\begin{proposition}
	\label{D1}
	Given a regular domain $\Omega\subset \mbb{R}^n$ and two bounded Lipschitz continuous vector fields $V,W: \mbb{R}^n\mapsto \mbb{R}^n$, let $\mbf{x}(t,\mbf{x_0})$ be the solution to the ODE $\dot{\mbf{x}} = V(\mbf{x})$ with initial value $\mbf{x_0} \in \Omega$. If the solution intersects the boundary $\partial\Omega$ at time $t_0 = \inf\{ t>0 \,|\, \mbf{x}(t)\in\partial\Omega \}$, then $\mbf{x}(t_0) \in \overline{D_1}$.
\end{proposition}

\begin{proof}
	We provide here a sketch of the proof. Note that $\mbf{x} \in \overline{D_1}$ implies the condition $\braket{V(\mbf{x}), \mbf{n}(\mbf{x})} \geq 0$, required at the point of intersection. This condition can be shown by the fact that the trajectory $\mbf{x}(t)$ intersects the boundary $\partial\Omega$ from the interior, and $\mbf{x}(t)$ is continuously differentiable and the boundary $\partial\Omega$ is smooth.
\end{proof}

\begin{remark}
	Similarly, the solution $\mbf{x}(t,\mbf{x_0})$ to the ODE $\dot{\mbf{x}} = W(\mbf{x})$ with $\mbf{x_0}\in \ext(\Omega)$ must intersect the boundary $\partial\Omega$ in $\overline{D_2}$. \\
\end{remark}

At this point, we have shown how a trajectory initialized in the interior $\Omega$ reaches the boundary $\partial\Omega$. In order for the trajectory to be extended, we also have to consider how a solution exits the boundary. In order to guarantee that the trajectory does not violate the patched vector field in a non-zero measure set, we have to carefully specify the direction at which the trajectory exits the boundary to avoid ``collision" with the field. We first formally define the notion of existence for a Caratheodory solution in a manner that suits our purpose. 

\begin{definition}
	Given a regular domain $\Omega\subset \mbb{R}^n$ and a bounded Lipschitz continuous vector field $V: \mbb{R}^n\mapsto \mbb{R}^n$, the solution $\mbf{x}(t,\mbf{x_0})$ to the ODE $\dot{\mbf{x}} = V(\mbf{x})$ is said to {\bf exist} in $\Omega$ up to $t_0$ if $\,\exists t_0 >0$ such that $\mbf{x}(t) \in \Omega$ for $\forall t \in [0,t_0)$. 
\end{definition}

\begin{lemma}
	\label{exit_path}
	Given a regular domain $\Omega\subset \mbb{R}^n$ and a bounded Lipschitz continuous vector field $V: \mbb{R}^n\mapsto \mbb{R}^n$, and a solution to the ODE $\dot{\mbf{x}} = V(\mbf{x})$ initialized at $\mbf{x}_0\in \overline{\Omega}$. Then the following statements are true:
	\begin{itemize}
		\item If $\mbf{x}_0\in \Omega$, then a solution always exists in $\Omega$.
		\item If $\mbf{x}_0 \in \partial\Omega$, then a solution exists in $\Omega$ if $\kappa'_V(\mbf{x}_0) \leq 0$, and a solution does not exist if $\kappa'_V(\mbf{x}_0) > 0$.
	\end{itemize}
\end{lemma}

\begin{proof}
	The proof of the first statement is simple. We first let $\mbf{x}(t)$ be a classical solution to the ODE $\dot{\mbf{x}} = V(\mbf{x})$ initialized at $\mbf{x}_0 \in \Omega$. Note that since $\Omega$ is open, it is possible to find an open ball in $\Omega$, $B_{\delta}(\mbf{x}_0) \subset \Omega$, centered at $\mbf{x}_0$ with radius $\delta$. Since $\mbf{x}(t)$ is continuous with respect to $t$, it is possible to find a $t_{\epsilon}>0$ such that $\mbf{x}(t) \in B_{\delta}(\mbf{x}_0)$ for $\forall t \in (0,t_{\epsilon})$. Therefore, we see that $\mbf{x}(t)$ exists in $\Omega$. \\
	
	The proof of the second statement is more involved, and we here only provide a proof sketch. We first let $\mbf{x}(t)$ be a classical solution to $\dot{\mbf{x}} = V(\mbf{x})$ initialized at $\mbf{x}_0 \in \partial\Omega$. We can then express a small neighborhood of $\mbf{x}_0$ as a graph of some analytic function $f: \mbb{R}^{n-1} \mapsto \mbb{R}$. We can then ``project" the trajectory $\mbf{x}(t)$ onto the boundary, and denote its projection as $\mbf{x}'(t)$. We can time-evolve the trajectory and its projection simultaneously forward infinitesimally by $\delta t$. We can find the displacement between the solution trajectory and its projection along the direction of the normal vector, $\braket{ \mbf{x}(\delta t)-\mbf{x}'(\delta t), \mbf{n}(\mbf{x}_0) }$, and expand it in terms of $\delta t$ into a convergent series. If the series converge into a negative number, then the trajectory is able to ``enter" the domain $\Omega$, so a solution exists in the domain. On the other hand, if the series converge into a positive number, then the trajectory can only ``leave" the domain $\Omega$, so a solution does not exist.
\end{proof}

\begin{remark}
	The visual interpretation of this lemma is rather straightforward. It essentially states that a trajectory initialized at the boundary of a regular domain can enter into the interior only if the field points inward at that point. \\
\end{remark}

If the trajectory is initialized in $D$, then the trajectory clearly must remain in the boundary as discussed in the remark of Lemma \ref{D}. Therefore, the trajectory can exit the boundary only if $\mbf{x}_0 \notin D$, or $\big( \kappa'_V(\mbf{x}_0) \leq 0 \big) \,\lor\, \big( \kappa'_W(\mbf{x}_0) \geq 0 \big)$, in which case a solution exists in the interior and exterior respectively.

\begin{proposition}[Exiting the Boundary]
	\label{exit}
	Given a regular domain $\Omega\subset \mbb{R}^n$ and two bounded Lipschitz continuous vector fields $V,W: \mbb{R}^n\mapsto \mbb{R}^n$, we let the initial condition be $\mbf{x}_0 \in \partial\Omega/D$, then a solution to $\dot{\mbf{x}} = \mcr{P}(V,W)(\mbf{x})$ can be uniquely constructed as:
	\begin{itemize}
		\item The classical solution to $\dot{\mbf{x}} = V(\mbf{x})$ at least up to some positive time if $\kappa'_V(\mbf{x}_0) \leq 0$.
		\item The classical solution to $\dot{\mbf{x}} = W(\mbf{x})$ at least up to some positive time if $\kappa'_W(\mbf{x}_0) \geq 0$ and $\kappa'_V(\mbf{x_0}) > 0$.
	\end{itemize}
\end{proposition}

\begin{proof}
	The proof follows directly from Definition \ref{patch} and Lemma \ref{exit_path}.
\end{proof}

\begin{remark}
	To interpret this proposition visually, we imagine a point in the boundary such that either the interior field or the exterior field points away from the boundary. If the interior field points away from the boundary, then the trajectory should enter $\Omega$ from $\partial\Omega$, and the trajectory will ``follow" the field initially, as both the trajectory and the interior field point inward with respect to the domain $\Omega$. Similarly, if the exterior field points away, then the trajectory should enter $\ext(\Omega)$ instead. If both fields point away from the boundary, then the trajectory has a choice of entering either $\Omega$ or $\ext(\Omega)$, and we let the trajectory enter $\Omega$ as the convention.
\end{remark}

\subsection{Bridging the solutions}

Up to this point, we have shown how a Caratheodory solution can be constructed in an open domain $\Omega$ and its boundary $\partial\Omega$, and we are now ready to construct the {\it maximal} Caratheodory solution that is capable of traversing all three domains: $\Omega$, and $\partial\Omega$, and $\ext(\Omega)$. WLOG, we can assume that the initial value $\mbf{x}_0 \in \Omega$ to be in the interior, then there are three possibilities for the time evolution of the trajectory. Essentially, the maximal Caratheodory solution is constructed as the extension of a classical solution in one domain with another classical solution in another domain. A formal description of the maximal solution is presented as a constructive proof of Theorem \ref{construct} based on the formal definition of {\it extension} as follows.

\begin{definition}
	Given a set $X$ and two functions, $\mbf{x}_1: [0,t_1] \mapsto X$, $\mbf{x}_2: [0,t_2] \mapsto X$, we say that $\mbf{x}_2(t)$ is an {\bf extension} of $\mbf{x}_1(t)$ if $t_2 > t_1$, and $\mbf{x}_2(t) = \mbf{x}_1(t)$ for $\forall t \in [0,t_1]$. Alternatively, we can say that $\mbf{x}_1(t)$ is {\bf extended with} $\mbf{x}_2(t-t_1)$ at point $\mbf{x}_1(t_1)$.
\end{definition}

\begin{lemma}
	Let $\mbf{x}_1: [0,t_1] \mapsto \mbb{R}^n$ and $\mbf{x}_2: [0,t_2] \mapsto \mbb{R}^n$ be two Caratheodory solutions to the ODE $\dot{\mbf{x}} = F(\mbf{x})$ where $F: \mbb{R}^n \mapsto \mbb{R}^n$ is some vector field. If $\mbf{x}_2(0) = \mbf{x}_1(t_1)$, then we can extend $\mbf{x}_1(t)$ with $\mbf{x}_2(t)$, which results in another Caratheodory solution to the ODE.
\end{lemma}

\begin{proof}
	This is obvious if we note that the procedure of attaching the two solution segments will result in potentially violating the ODE only at a single point $\mbf{x_1}(t_1)$.
\end{proof}

\begin{theorem}[Construction of Maximal Caratheodory Solution]
	\label{construct}
	Given an open regular domain $\Omega\in\mbb{R}^n$ and two bounded Lipschitz continuous vector fields $V,W: \mbb{R}^n \mapsto \mbb{R}^n$, it is possible to construct a unique Caratheodory solution to the ODE $\dot{\mbf{x}} = \mcr{P}(V,W)(\mbf{x})$, where $\mcr{P}$ is the patching operation defined in Definition \ref{patch}.
\end{theorem}

\begin{proof}
	WLOG, we assume that the initial value of the ODE is $\mbf{x}_0 \in \Omega$. Let $\mbf{x}_1(t)$ be the classical solution to $\dot{\mbf{x}} = V(\mbf{x})$ existing up to $t_1$ in $\Omega$. If $t_1 = +\infty$, then $\mbf{x}_1(t)$ is trivially a Caratheodory solution as well. We then consider the case where $t_1$ is finite, meaning that the trajectory enters the boundary $\partial\Omega$ at some point $\mbf{p} = \mbf{x}_1(t_1) \in \overline{D_1}$ (see Proposition \ref{D1}). Note that $\kappa'_V(\mbf{p})\geq 0$, so we are left with the following cases:
	
	\begin{itemize}
		\item If $\kappa'_W(\mbf{p}) \geq 0$, then we extend $\mbf{x}_1(t)$ with the maximal classical solution to the ODE $\dot{\mbf{x}} = W(\mbf{x})$ in $\ext(\Omega)$ initialized at $\mbf{p}$. The extended solution violates the ODE only at $\mbf{p}$. 
		
		\item If $\kappa'_W(\mbf{p}) < 0$ and $\kappa'_V(\mbf{p}) = 0$, then we extend $\mbf{x}_1(t)$ with the maximal classical solution to the ODE $\dot{\mbf{x}} = V(\mbf{x})$ in $\Omega$ initialized at $\mbf{p}$. The extended solution violates the ODE only at $\mbf{p}$.
		
		\item If $\kappa'_W(\mbf{p}) < 0$ and $\kappa'_V(\mbf{p}) > 0$, then $\mbf{p} \in D$, and let $\mbf{x}_2(t)$ be the maximal classical solution to the ODE $\dot{\mbf{x}} = U(\mbf{x})$ existing in $D$ up to $t_2$, where $U(\mbf{x})$ is defined in corollary \ref{contain}. If $t_2=+\infty$, then we are done; if $t_2$ is finite, then we let $\mbf{q} = \mbf{x}_2(t_2) \in \partial D$, implying that $\kappa'_V(\mbf{q}) = 0$ or $\kappa'_W(\mbf{q}) = 0$, reducing to the previous two cases. The extended solution violates the ODE only at $\mbf{p}$ and $\mbf{q}$.
	\end{itemize}
	
	We iterate this procedure every time the trajectory enters the boundary, with the treatment of the entrance from the exterior $\ext(\Omega)$ mirroring the entrance from interior $\Omega$. This gives us the maximal Caratheodory solution if we take $t \to \infty$. It is clear that the solution can only be extended countably many times, and each segment is classical in nature (see Proposition \ref{exit}) meaning that the ODE is only violated at countably many points, so the maximal solution is in fact Caratheodory by Definition \ref{def_cara}. 
	
\end{proof}

\begin{remark}
	Visually, for a trajectory initialized in $\Omega$ that enters the boundary $\partial\Omega$, we have three scenarios. In the first scenario, the trajectory is guided by the interior field in a way such that it barely ``scrapes" the boundary and returns back to the interior. In the second scenario, the trajectory ``crosses" the boundary and continues its path into the exterior if the exterior field at the intersection points outward. Finally, if the trajectory enters into the boundary at a point where the interior and exterior fields both point inward, then the trajectory ``tunnels" in the boundary to avoid the two fields and continues to do so until it reaches a point where one of the two fields begins pointing outward, then the trajectory begins to follow that field. If both fields never point outward, then the trajectory remains in the boundary forever. \\
\end{remark}

This concludes the section which establishes the necessary mathematical formalism for discussing the memcomputing dynamics which are guided autonomously by such patchy vector fields (see Eq. (\ref{mem})).

\section{Memory Dynamics}
\label{sec_mem}

To find an assignment $\mbf{v}$ that minimizes the constraint in Eq.~(\ref{eng}), we can time-evolve the voltages autonomously with the following ODE (which is an equivalent way of writing Eqs.~(\ref{eq:C})-(\ref{eq:rigid})):
\begin{equation}
\label{mem}
\begin{split}
&\dot{v}_i = 
\sum_{j=1}^m \Big\{ \frac{1}{2} x_{l,j}x_{s,j} q_{ij} \min_{\{i'\neq i \,|\, q_{i'j} \neq 0 \}}(1-q_{i'j}v_{i'}) 
+ (1+\zeta x_{l,j})(1-x_{s,j})\delta_{i \sigma_j}q_{ij}C_j(\mbf{v}) \Big\},
\\
&\dot{x}_{s,j} = \beta \big( x_{s,j} + \epsilon \big) \big( C_j(\mbf{v}) - \gamma \big), \\
&\dot{x}_{l,j} = \alpha \big(C_j(\mbf{v}) - \delta).
\end{split}
\end{equation}
From now on, we shall refer to this particular ODE as {\it memory dynamics}, where $\mbf{v} \in [-1,+1]^n$ are voltages corresponding to the Boolean variables of the original 3-SAT problem with $n$ variables and $m$ clauses. Furthermore, we refer to $\mbf{x_s} \in [0,1]^m$ as {\it short-term memory} and $\mbf{x_l} \in [1, x_{max}]^m$ as {\it long-term memory}, where $x_{max}>1$ is some upper bound to the slow variable dynamics\footnote{Note that the bounds on the dynamic variables $\{ \mbf{v}, \mbf{x_f}, \mbf{x_s} \}$ are not enforced ``naturally'' by the memory dynamics. They are enforced through the introduction of auxiliary fields in the exterior of the bounded domain. See Section \ref{compact} for a formal discussion of the procedure of doing so.}. The parameters $\{ \alpha, \beta, \gamma, \delta, \zeta, \epsilon \}$ are positive constants empirically tuned to provide the regularity and convergence of the dynamics with a sufficiently fast time scale (see Sec.~\ref{Numerics}).
We will use the non-subscript symbol, $\mbf{x} = \{ \mbf{v}, \mbf{x_s}, \mbf{x_l} \} \in \mbb{R}^{n+2m}$, to denote the collection of all dynamic variables, allowing us to write the ODE as
\begin{equation*}
\dot{\mbf{x}} = F(\mbf{x}),
\end{equation*}
where $F$ is some flow field corresponding to the RHS of Eqs.~(\ref{mem}). \\

For the sake of having a more compact expression for the ODE equations, it is convenient for us to borrow the notation of Eq. (\ref{eq:G}) and denote
\begin{equation}
\label{G}
G_{ij}(\mbf{v}) = \frac{1}{2} q_{ij} \min_{\{i'\neq i \,|\, q_{i'j} \neq 0 \}}(1-q_{i'j}v_{i'})
\end{equation}
as the {\it gradient-like} term, as it approximately follows  the directional gradient of the energy of the $j$-th clause along the direction of $v_i$ (see Eq.~(\ref{c_ham})). Note that the actual directional gradient is similar to Eq. (\ref{G}) with the only exception being the $\min$ operation replaced with the product $\prod$. The magnitude of the gradient-like term for a voltage in the $j$-th constraint is related to the value of the other two voltages in the constraint.
Similarly, we borrow the notation of Eq. (\ref{eq:rigid}) and denote
\begin{equation}
\label{R}
R_{ij}(\mbf{v}) = \delta_{i\sigma_j}q_{ij}C_j(\mbf{v})
\end{equation}
as the {\it rigidity} term. Its magnitude is equivalent to the clause constraint $C_j$ defined in Eq.~(\ref{cj}) if $v_i$ is the voltage that defines $C_j$, and zero otherwise.\\

We can then succinctly write the voltage dynamics as
\begin{equation}
\label{mem_red}
\dot{\mbf{v}} =  \mbf{G}(\mbf{v})(\mbf{x_s} \ast \mbf{x_l})  + \mbf{R}(\mbf{v}) \big( (1 + \zeta \mbf{x_l}) \ast (1- \mbf{x_s}) \big),
\end{equation}
where $\mbf{G}$ and $\mbf{R}$ are treated as $n \times m$ matrices dependent on $\mbf{v}$, the operator $\ast$ denotes element-wise multiplication, and $\mbf{x_s}$ and $\mbf{x_l}$ are treated as column vectors for the sake of matrix operation. In this form, we can clearly see that the gradient-like and rigidity dynamics are weighted clause-wise by the memory variables. The presence of dynamic memory is a central feature of our dynamics. \\



For certain analyses of dynamical properties, it is sufficient and more convenient for us to focus on the analytic properties of the following simplified dynamics
\begin{equation}
\label{mem_sim}
\begin{split}
\dot{\mbf{v}} &= \mbf{G}(\mbf{v})\mbf{x_l}, \\
\mbf{\dot{x}_l} &= \alpha \mbf{C}(\mbf{v}),
\end{split}
\end{equation}
All the dynamical properties derived in this work under the assumption of this simplified dynamics can be easily generalized to the full dynamics if we assume sufficiently general forms for $\mbf{G}$ and $\mbf{C}$ (see Section \ref{reduce} and \ref{no_period}).

\subsection{Discontinuous Hyperplanes}
\label{plane}

We first make the important observation that the gradient-like term $\mbf{G}$ is not differentiable everywhere and the rigidity term $\mbf{R}$ is not Lipschitz continuous. These irregular points form hyperplanes generated by the minimum operation in the voltage space. In this section, we construct the hyperplanes which contain all the points of discontinuity for the rigidity term. These hyperplanes are generated by the binary values of $\delta_{i\sigma_j}$ (which contains implicitly a minimum operation), and they form $n-1$ dimensional hyperplanes in the voltage space $\mbb{R}^n$. A similar construction also applies for the gradient-like term\footnote{Finding these hyperplanes for the gradient-like term is not strictly necessary, as the gradient-like term is already Lipschitz continuous. The hyperplanes will only contain points of non-differentiability, which will not affect the existence and uniqueness of the dynamical trajectory (see Section \ref{cara}).}. 

\begin{proposition}[Hyperplanes]
	There exists a union of countably many ($n-1$)-dimensional hyperplanes in $\mbb{R}^{n}$ such that it contains all the points where the field $F$ is discontinuous.
\end{proposition}

\begin{proof}
	To lessen the burden of notation, we let $N = [[1,n]]$ and $M = [[1,m]]$. We first recall from Eq.~(\ref{sigma}) that
	\begin{equation*}
	\sigma_j = \argmin_{ \{ i \,|\, q_{ij}\neq 0 \} }( 1 - q_{ij}v_i ),
	\end{equation*}
	which implies that the field can only be discontinuous at a point where some $j$ can be chosen such that the $\argmin$ operation is degenerate, which is equivalent to the following condition
	\begin{equation}
	\label{qv}
	\exists j \in M, \, \exists i_1,i_2 \in \{ i\in N \,|\, q_{ij}\neq 0 \}, \, q_{i_1j}v_{i_1} = q_{i_2j}v_{i_2}.
	\end{equation}
	We denote the set of all points $\mbf{x}$ that satisfies the above condition as $\partial\Omega$. \\
	
	For any two distinct indices of the Boolean variables, or $\forall i_1,i_2 \in N$ where $i_1 \neq i_2$, we can define a {\it positive} hyperplane $H_P$ and a {\it negative} hyperplane $H_N$ as follows
	\begin{equation*}
	\begin{split}
	H_{P,i_1i_2} &= \{ \mbf{x}\in\mbb{R}^{n+2m} \,|\, v_{i_1} = v_{i_2} \}, \\
	H_{N,i_1i_2} &= \{ \mbf{x}\in\mbb{R}^{n+2m} \,|\, v_{i_1} = -v_{i_2} \}.
	\end{split}
	\end{equation*}
	Note that both are $(n-1)$-dimensional. If we recall that $q_{ij} = \pm 1$ for all nonzero elements of the polarity matrix $Q$, then it can be shown that any voltage assignment $\mbf{v}$ that satisfies condition (\ref{qv}) must be in one of such hyperplanes. Therefore, the union of all such hyperplanes must contain $\partial\Omega$, or
	\begin{equation*}
	\partial\Omega \subseteq \bigcup_{i_1 \neq i_2} \big( H_{P,i_1i_2} \cup H_{N,i_1i_2} \big).
	\end{equation*}
	Note that there are ${n \choose 2}$ positive and negative hyperplanes each, so there are $2{n \choose 2}$ hyperplanes in total, which is a countable number. This proves the proposition. 
\end{proof}

\begin{remark}
	An immediate consequence of this proposition is that the rigidity term is only discontinuous at a measure zero subset of the phase space, as all the hyperplanes of discontinuities are of measure zero, and there are only countably many of them. Therefore, the rigidity term is smooth almost everywhere. Note that these hyperplanes also contain the points at which the gradient-like term is non-differentiable, meaning that the gradient-like term is also smooth almost everywhere. \\
\end{remark}

As the field is continuous almost everywhere, it clearly admits a Caratheodory solution for any initial value, if the fields are patched appropriately at the hyperplanes according to the procedure in Definition \ref{patch}. Note that the phase space of the dynamics is an $n+2m$-dimensional hypercube (see Section \ref{compact}), which is partitioned into disjoint subsets by the hyperplanes. A caveat here is that the domains are {\it almost} regular as the intersections of the hyperplanes generate regions of non-smoothness. However, note that these intersections have zero measure relative to the hyperplanes, so it is unlikely for a trajectory to encounter them. For the sake of analytic completeness, even if we assume that a trajectory were to encounter an intersection of planes, this does not invalidate our method of constructing a Caratheodory solution, as there is still a unique projection of vector fields on these intersecting regions. As for using the directional curvature as the exit condition, the zeroth order directional curvature can be defined as $\pm \infty$ accordingly at these regions, and the exit protocol as given in proposition \ref{exit} remains unchanged.

\subsection{Compact Positive Invariant Set}
\label{compact}

To respect the Boolean structure of the original 3-SAT problem, the dynamics as given in Eqs.~(\ref{mem}) must be bounded explicitly. First of all, we choose the bound the voltages explicitly in a compact set, which is $[-1,+1]^n$ for our work\footnote{Note the choice of $-1$ and $+1$ is to make an intuitive connection to the false and true state. As $\mbf{\dot{v}}$ and $\mbf{C}(\mbf{v})$ are linear with respect to $\mbf{v}$, the lower and upper bound for the voltages can be chosen arbitrarily (centered at $0$), and the original dynamics can be recovered via an appropriate rescaling of the memory variables and constant parameters.}. Furthermore, the short-term memory $\mbf{x_s}$ has to be bounded in $[0,1]^m$, as a way to completely stop either the gradient-like or rigidity contribution to the dynamics for each clause. Finally, the long-term memory $\mbf{x_l}$ has to be bounded in $[1,x_{\max}]^m$ in practice\footnote{This is mostly for the sake of a practical implementation of our solver. Note that if the upper bound is absent, then a digital implementation would require infinite precision and an analogue implementation would require infinite energy to guarantee accurate simulation, neither of which is possible.}. In fact, for the analysis in the following sections, we will regularly assume that the bound $x_{\max}$ on the long-term memory is absent, meaning that $\mbf{x_l} \in [1,+\infty)^m$, in an effort to increase the generality of certain propositions. The bounds on the short-term memory is crucial, however, and will always be assumed present. \\

Putting everything together, this means that the dynamics must be fully contained within the region $O = [-1, 1]^{n}\times [0, 1]^m \times [1,x_{\max}]^m$, which is a compact set in $\mbb{R}^{n+2m}$. To put this formally, we have to show that $O$ is an invariant set, and any trajectory with initial value in $O$ must remain in $O$ forever. To do so, we consider a general ODE with the flow field defined in a regular domain $\Omega$, such that a Caratheodory solution exists in the domain. In other words, we have $\dot{\mbf{x}} = F(\mbf{x})$, where $F: \mbb{R}^{n+2m} \mapsto \mbb{R}^{n+2m}$ is some sufficiently regular vector field in $\Omega \subset \mbb{R}^{n+2m}$. Suppose we now wish to modify the vector field in such a way that, for any initial value $\mbf{x}(0) \in \Omega$, the trajectory is contained entirely within the closure of that domain $\overline{\Omega}$, or $\mbf{x}(t) \in \overline{\Omega}$ for $\forall t \geq 0$. This has to be done carefully such that the original flow field in $\Omega$ remains the same. We do so by patching the original vector field with a ``bounding" vector field in $\ext(\Omega)$ as follows.

\begin{lemma}[Bounding Field]
	\label{bound}
	Let $\Omega \subset \mbb{R}^n$ be a smooth open domain, and let $F: \mbb{R}^n \mapsto \mbb{R}^n$ be some bounded vector field that admits a Caratheodory solution in $\Omega$. Let $G: \mbb{R}^n \mapsto \mbb{R}^n$ be some Lipschitz continuous vector field satisfying
	\begin{equation*}
	\forall \mbf{x}\in\partial\Omega, \quad G(\mbf{x}) = -M \mbf{n}(\mbf{x}),
	\end{equation*}
	where $M>0$ can be any positive constant, and $\mbf{n}(\mbf{x})$ is the outward pointing unit normal vector of the boundary $\partial\Omega$ at $\mbf{x}$. Then any construction of the Caratheodory solution (see theorem \ref{construct}) to the ODE, $\dot{\mbf{x}} = \mcr{P}(F,G)(\mbf{x})$, with initial value $\mbf{x_0} \in \Omega$, has the property that $\mbf{x}(t) \in \overline{\Omega}$ for $\forall t \geq 0$. 
\end{lemma}

\begin{proof}
	Note that based on the construction given in theorem \ref{construct}, it is sufficient to show that $\kappa'_{G}(\mbf{x}) < 0$ for $\forall \mbf{x} \in \partial\Omega$, as the trajectory will never be able to exit into the region $\ext(\Omega)$. By construction, we have $\braket{G(\mbf{x}), \mbf{n}(\mbf{x})} = -M$ for $\forall \mbf{x}\in\partial\Omega$, so it follows directly from definition \ref{curve} that $\kappa'_G(\mbf{x}) = -M < 0$.
\end{proof}

\begin{remark}
	By adding the ``bounding" vector field $G$, we are essentially ``projecting" any ``stray" fields onto the boundary $\partial\Omega$, such that whenever a trajectory enters the boundary, it will continue to ``flow" inside the boundary (see corollary \ref{contain}) and never escape $\overline{\Omega}$. An important point to note is that the dynamics do {\bf not} stop after reaching $\partial\Omega$. 
\end{remark}

\begin{corollary}[Invariant Hypercube]
	\label{mem_bound}
	Let $O = [-1, 1]^{n}\times [0, 1]^m \times [1,x_{\max}]^m \subset \mbb{R}^{n+2m}$, and let $F: \mbb{R}^n \mapsto \mbb{R}^n$ be some bounded vector field that admits a Caratheodory solution in $O$. For $\forall i\in [[1,n+2m]]$, we let $k_i$ be the lower bound of the $i$-th quotient space of $O$, and let $K_i$ be the upper bound. Then we define the left and right hyperplanes, $L_i$ and $R_i$, as follows
	\begin{equation*}
	L_i = \{ \mbf{x}\in O \,|\, x_i = k_i \}
	\qquad
	R_i = \{ \mbf{x}\in O \,|\, x_i = K_i \}.
	\end{equation*}
	Let $G: \mbb{R}^n \mapsto \mbb{R}^n$ be some Lipschitz continuous vector field such that $\forall i \in [[1,n+2m]]$:
	\begin{equation*}
	\begin{split}
	&\forall \mbf{x}\in L_i, \quad G(\mbf{x}) = M \hat{\mbf{e}}_i \\
	&\forall \mbf{x}\in R_i, \quad G(\mbf{x}) = -M \hat{\mbf{e}}_i,
	\end{split}
	\end{equation*}
	where $M>0$ can be any positive constant, and $\hat{\mbf{e}}_i$ is the $i$-th component of the standard basis. Then $O$ is a positive invariant set under the ODE, $\dot{\mbf{x}} = \mcr{P}(F,G)(\mbf{x})$. Furthermore, the superposed flow field on the hyperplanes is given by
	\begin{equation*}
	\begin{split}
	&\forall \mbf{x} \in L_i, \quad \mcr{P}(F,G)(\mbf{x}) = F(\mbf{x}) - F_i(\mbf{x})\big( 1-H(x_i) \big)\hat{\mbf{e}}_i, \\
	&\forall \mbf{x} \in R_i, \quad \mcr{P}(F,G)(\mbf{x}) = F(\mbf{x}) - F_i(\mbf{x}) H(x_i) \hat{\mbf{e}}_i,
	\end{split}
	\end{equation*}
	where $H$ denotes the Heaviside step function.
\end{corollary}

\begin{remark}
	To visualize the bounding flow field, one can imagine a hypercube $O$ where the internal field remains unchanged, and the exterior field is ``pressing against" the faces of the cube to ensure that any trajectory initialized inside the cube remains inside. The flow field on the ``faces" of the cube is simply the projection of the field onto the plane if the field were to point outward. This bounding procedure effectively mimics the numerical technique that we use to bound the dynamics, where any outward pointing component of the flow field on the boundary is simply ignored. \\
\end{remark}

From now on, when we refer to {\it memory dynamics}, we are referring to the system of ODEs given in Eqs.~(\ref{mem}), with the bounds of the dynamics enforced by the exterior field $G$ as constructed in corollary \ref{mem_bound}. To lessen the burden of notation, we shall refer to the patched flow field of the memory dynamics, $\mcr{P}(F,G)$, simply as $F$. The set $O = [-1, 1]^{n}\times [0, 1]^m \times [1,x_{\max}]^m$ is then a positive invariant set of the memory dynamics. 

\subsection{Gauge Invariance of Dynamics}
\label{gauge_inv}

In this Section, we primarily focus on formalizing the notion of gauge invariance for the dynamics governed by  Eqs.~(\ref{mem}). To do so, it is convenient to first reformulate the flow field as a group action.

\begin{definition}[Time Mapping]
	Given a vector field $V: \mbb{R}^n \mapsto \mbb{R}^n$ such that there is a unique positive solution $\mbf{x}(\mbf{x_0},t)$ to the ODE $\dot{\mbf{x}} = F(\mbf{x})$ for any initial value $\mbf{x_0}\in \mbb{R}^n$, we define a mapping $T_s: \mbb{R}^n \mapsto \mbb{R}^n$ for $\forall s \geq 0$ as follows
	\begin{equation*}
	T_s(\mbf{x}_0) = \mbf{x}(s,\mbf{x}_0).
	\end{equation*}
\end{definition}

\begin{remark}
	It should first be noted that $T_s$ is a well defined operator $\forall s \geq 0$, as the solution to the ODE with any initial value is unique. It can also be easily checked that the operators $T_s$ form a semigroup with the identity element being $T_0$. In fact, we have
	\begin{equation*}
	T_{s_2}T_{s_1}(\mbf{x}_0) = T_{s_2}\big(\mbf{x}(s_1,\mbf{x}_0)\big) = \mbf{x}\Big( s_2, \big(\mbf{x}(s_1,\mbf{x}_0)\big) \Big) = \mbf{x}(s_1+s_2, \mbf{x}_0) = T_{s_1+s_2}(\mbf{x}_0).
	\end{equation*}
	The reason why the operators form only a semigroup is because it does not necessarily have a group inverse, as we do not require the negative solution to the ODE to exist or be unique. \\
\end{remark}

For our memory dynamics, an important property of $T_s$ is that it is invariant under gauge conjugation. This is important as it essentially allows us to simplify the analysis of the memory dynamics by assuming that a solution vector is $\mbf{v}_0 = \mbf{+1}$.

\begin{proposition}[Gauge Invariance of Dynamics]
	Given a polarity matrix $Q$ corresponding to a satisfiable 3-SAT instance with some solution vector $\mbf{v}_0$, and an operator $T_s$ corresponding to the memory flow field $F$, we have the following
	\begin{equation*}
	T_s = G_{\mbf{v}_0} \circ T_s \circ G_{\mbf{v}_0}^{-1},
	\end{equation*}
	where $G_{\mbf{v}_0}$ is the gauge mapping operation in Definition \ref{def_gauge}.
\end{proposition}

\begin{remark}
	We here provide a proof sketch of this proposition. We first begin by noting that the operators $T_s$ form a semigroup, so it is sufficient to show that the infinitesimal group generator $F$ is invariant under gauge conjugation, or
	\begin{equation*}
	F = G_{\mbf{v}_0} \circ F \circ G_{\mbf{v}_0}^{-1}.
	\end{equation*}
	This is equivalent to showing that transforming both the LHS and RHS of the equations in (\ref{mem}) does not violate the equalities, which can be easily shown by recalling that $C_j(\mbf{v})$ is gauge invariant (see lemma \ref{cj_inv}). Then we see that a prefactor of $v_{0,i}$ appears in both the LHS and RHS of the voltage equations. Furthermore, we can also easily show that the discontinuous hyperplanes (see Section \ref{plane}) and the boundaries of the hypercube containing the dynamics (see corollary \ref{mem_bound}) are also invariant under the gauge mapping. Therefore, the operator $T_s$ must be invariant under gauge conjugation for $\forall s \geq 0$.
\end{remark}

\subsection{Correspondence between Fixed Points and Solutions}
\label{sec_cor}

The fixed points of the dynamics must correspond to a solution to the original 3-SAT instance (if the instance is satisfiable). Otherwise, the correctness of the memory dynamics as a SAT solver cannot be guaranteed, as it is possible for the dynamics to terminate at a point corresponding to a non-solution. We dedicate this section to the correspondence between the fixed points of the dynamics and the solutions of a 3-SAT instance. Before we continue this discussion, we first note that it is possible to solve a 3-SAT Boolean formula with a partial assignment of the Boolean variables, which corresponds to hyperfaces on the voltage hypercube (see Section \ref{backbone}). In other words, it is possible for the dynamics to solve a 3-SAT instance by converging to a hyperface instead of any particular solution vector, and the solution can be extracted by choosing an arbitrary vertex of that hyperface.

\begin{definition}[Solution Plane]
	\label{par_sol}
	Consider a 3-SAT problem defined by an $n \times m$ polarity matrix. If we can find a non-empty subset of indices, $I \in [[1,n]]$, such that there are exactly $2^{|I|}$ distinct solutions coinciding to the assignment of the Boolean variables indexed $[[1,n]]/I$, then the problem is said to be {\bf partially solvable}, and we refer to $\mcr{I} = [[1,n]]/I$ as the {\bf isolated index set} of the solutions. $\mcr{I}$ is said to be {\bf proper} if it has no proper subset that is also an isolated index set. \\
	
	Let $\mbf{v'}$ be a solution vector, and $\mcr{I}$ be a proper index set. We define the {\it solution plane} to be
	\begin{equation*}
	H(\mbf{v'}, \mcr{I}) = \{ \mbf{v}\in[-1,+1]^n \lcond \forall i\in\mcr{I}, \, v_i = v'_i \}.
	\end{equation*}
	The vertices (which are solution vectors) are said to be {\bf connected} by this plane. Any solution vector that is not connected by a solution plane is said to be {\bf isolated}.
\end{definition}

\begin{remark}
	Note that for a given solution vector $\mbf{v'}$, its proper index set is not necessarily unique, and depends on the polarity matrix of the 3-SAT Boolean formula. The solution plane is, however, unique given a solution vector and its proper index set. \\
\end{remark}


\begin{lemma}
	\label{c_plane}
	Let $\mbf{v'}$ be a solution vector for which a proper index set $\mcr{I}$ exists. Then
	\begin{equation*}
	\forall \mbf{v} \in H(\mbf{v'}, \mcr{I}), \qquad
	\mbf{C}(\mbf{v}) = \mbf{0}.
	\end{equation*}
	On the other hand, let $\mbf{v}$ be a vector such that $\mbf{C}(\mbf{v}) = 0$, and $\mbf{v'} = \sign(\mbf{v})$ be the corresponding solution vector. If a proper index set $\mcr{I}$ exists for the solution vector, then
	\begin{equation*}
	\exists \mcr{I}, \qquad
	\mbf{v} \in H(\mbf{v'}, \mcr{I}).
	\end{equation*}
\end{lemma}

\begin{proof}
	The proof follows trivially from the definition of the solution plane (see definition \ref{par_sol}) and the definition of the clause constraint (see Eq.~(\ref{cj})).
\end{proof}

\begin{remark}
	One immediate implication of this lemma is that once we have found a voltage assignment such that the global constraint (or energy) is zero, then the voltage vector must be either a solution vector, or it must be in some solution plane. If it is in a solution plane, then we can take any vertex of that plane as a solution to the 3-SAT problem. \\
\end{remark}

Since a solution vector and a vector in a solution plane both solve the 3-SAT problem, we can treat a solution vector equivalently to a solution plane. Then for an isolated solution vector $\mbf{v'}$, its solution plane simply refers to itself.

\begin{proposition}[Solution Fixed Points]
	\label{fix_sol}
	If $\mbf{v'}$ is in a solution plane, then $\mbf{x'} = \{ \mbf{v'}, \mbf{x_s}, \mbf{x_l} \}$ will eventually evolve to a fixed point in the same solution plane $\forall \mbf{x_s} \in [0,1]^m$, $\forall \mbf{x_l} \in [1,+\infty)^m$. Conversely, if $\mbf{x'}$ is a fixed point of the memory dynamics, then $\mbf{v'} = \{ x'_1, ... ,x'_n \}$ is in a solution plane.
\end{proposition}

\begin{proof}
	We first show the first part of the proposition. Given any $\mbf{x_s}$ and $\mbf{x_l}$, we denote $\mbf{x'} = \{ \mbf{v'}, \mbf{x_f}, \mbf{x_s} \}$, where $\mbf{v'}$ is in a solution plane. WLOG, we can assume that the solution plane is $H\big( \mbf{+1}, [[1,n']] \big)$, where $n' < n$ and $\sign(\mbf{v'}) = \mbf{+1}$ (if not, we can simply gauge the polarity matrix and relabel the indices such that it is true). We first begin by showing that $\mbf{C}(\mbf{v}( \mbf{x'}, t )) = \mbf{0}$ for $\forall t>0$. To do so, it is sufficient to show that for all such $\mbf{v'}$ (and arbitrary memory), the voltage flow field is positive, meaning that the trajectory will be pressed against the solution plane. \\
	
	WLOG, we first focus only on the dynamics of $v_1$ influenced by clause $j$ (assuming that $q_{1j} \neq 0$). The gradient influence is
	\begin{equation*}
	G_{1j} = \frac{1}{2} q_{1j} (1 - q_{ij}v_i).
	\end{equation*}
	Note that the gradient-like term is non-positive only if $q_{1j} = -1$, which implies that $q_{ij}v_j = +1$ otherwise $C_j \neq 0$. In this case, we have $G_{1j} = 0$, therefore it is required that $G_{1j} \geq 0$ for all cases. For the rigidity term, we have
	\begin{equation*}
	\begin{split}
	R_{1j} &= \delta_{1\sigma_j} q_{1j} C_j(\mbf{v}),
	\end{split}
	\end{equation*}
	which is necessarily zero as $C_j(\mbf{v}) = 0$. Therefore, all possible contributions to $v_1$ are non-negative, and this applies for $\forall i \in \mcr{I}$. This means that $\mbf{C}(\mbf{v}( t, \mbf{x'} )) = \mbf{0}$ for $\forall t>0$, then $\mbf{\dot{x}_s}(t) < \mbf{0}$ and $\mbf{\dot{x}_l}(t) < \mbf{0}$, so both memory variables will decay and terminate at $0$ and $1$ respectively. \\
	
	The proof of the second part of this proposition is shown as Corollary \ref{corr_conv}, immediately after we establish certain properties of the basin of attraction for our dynamics.
\end{proof}

\begin{remark}
	The proposition essentially states that once the voltage vector reaches a solution plane, then the dynamics will flow to a fixed point. On the other hand, if the voltage vector has not reached a solution plane yet, then the dynamics will continue to evolve (until it finds the solution). If the original 3-SAT problem is unsatisfiable, then the dynamics will continue to evolve forever. 
\end{remark}

\section{Basin of Attraction}
\label{basin}

From proposition \ref{sign}, we see that the 3-SAT problem is essentially solved once we have discovered a voltage assignment such that $\mbf{C}(\mbf{v}) < \frac{1}{2}$, and the dynamics can be terminated. However, in some cases, the implementation of this termination condition is perhaps not feasible, so we have to allow the dynamics to fully converge to a solution vector $\mbf{v}_0$. In this case, it is necessary for us to determine the {\it basin of attraction} in which the dynamics are guaranteed to evolve towards the solution. We first formally define the basin of attraction as follows.

\begin{definition}
	\label{def_basin}
	Given some flow field $F: \mbb{R}^n \mapsto \mbb{R}^n$, let $\mbf{x'}$ be a fixed point of this field. We define the {\bf basin of attraction} of $\mbf{x'}$ as
	\begin{equation*}
	B(\mbf{x'}) = \{ \mbf{x}_0\in\mbb{R}^n \,\,\big\lvert\,\, \lim_{t\to+\infty}\mbf{x}(t,\mbf{x}_0) = \mbf{x'} \}.
	\end{equation*}
\end{definition}

\begin{remark}
	From the first part of proposition \ref{fix_sol}, we see that every solution plane must contain a fixed point. We can then modify the above definition to solution plane as follows
	\begin{equation*}
	B(\mbf{v'}) = \{ \mbf{x}_0\in\mbb{R}^{n+2m} \,\,\big\lvert\,\, \lim_{t\to +\infty} \mbf{v}(t,\mbf{x}_0) \in \bigcup_{\mcr{I}} H(\mbf{v', \mcr{I}}) \},
	\end{equation*}
	where $H(\mbf{v'}, \mcr{I})$ denotes the solution plane of $\mbf{v'}$ corresponding to the proper index set $\mcr{I}$ (see definition \ref{par_sol}). It is important to note that the basin of attraction of a solution vector is a subset of the full space, or $B(\mbf{v'}) \subseteq \mbb{R}^{n+2m}$, even though the fixed points are defined in the voltage space $\mbb{R}^n$. This is because the dynamics of the memory variables still affect the flow field of the voltages. \\
\end{remark}

An object that will be often evoked in the following discussion is the orthant of the voltage space that contains a solution plane. To make the discussion easier, we shall refer to such orthant as a {\it solution orthant}.

\begin{definition}[Solution Orthant]
	\label{orth}
	Given a solution vector $\mbf{v'} \in \mbb{R}^n$ and a proper index set $\mcr{I}$ (see definition \ref{par_sol}), we define its {\bf solution orthant} to be
	\begin{equation*}
	J(\mbf{v'}, \mcr{I}) = \{ \mbf{v}\in [-1,+1]^n \,\,\big\lvert\,\, \forall i \in \mcr{I} ,\, v_i  v'_i \geq 0  \}.
	\end{equation*}
	Furthermore, we define the {\bf restricted solution orthant} to be
	\begin{equation*}
	J'(\mbf{v'}, \mcr{I}) = \{ \mbf{v}\in [-1,+1]^n \,\,\big\lvert\,\, \forall i \in \mcr{I} ,\, v_i  v'_i \geq 1-2\gamma  \},
	\end{equation*}
	where $\gamma < \frac{1}{2}$ is the offset parameter defined in Eqs. (\ref{mem}).
\end{definition}

\begin{remark}
	From the discussion in Section \ref{gauge}, we see that a satisfiable 3-SAT problem can always be gauged such that the solution vector is $\mbf{+1}$. This means that in our analysis, we can always assume that any solution orthant contains the positive orthant of $[-1,+1]^n$, as long as we guarantee that the dynamics are fully contained within the orthant. \\
	
	For better visualization, one can imagine the entire bounded space of the dynamics, $O$, as a hypercube. Then a solution orthant is simply a hyper-rectangle 
	with some of its ``edges" halved in such a way that it still contains a solution plane. A restricted solution orthant is constructed similarly except these edges are quartered instead. This can be described by the following containment relation
	\begin{equation*}
	H(\mbf{v'}, \mcr{I}) \subset J'(\mbf{v'}, \mcr{I}) \subset J(\mbf{v'}, \mcr{I}) \subset O.
	\end{equation*}
\end{remark}

\begin{lemma}
	\label{c14}
	Given a solution vector $\mbf{v'}$ for which a proper index set $\mcr{I}$ exists, we have
	\begin{equation*}
	\forall \mbf{v} \in J'(\mbf{v'}, \mcr{I}), \qquad
	\mbf{C}(\mbf{v}) \leq \gamma.
	\end{equation*}
	On the other hand, given a vector $\mbf{v} \in [-1,+1]^n$ such that $\mbf{C}(\mbf{v}) \leq \gamma$, let $\mbf{v'} = \sign(\mbf{v})$ be the corresponding solution vector. If there is a proper index set for this solution, then
	\begin{equation*}
	\exists \mcr{I}, \qquad
	\mbf{v} \in H(\mbf{v'}, \mcr{I}).
	\end{equation*}
\end{lemma}

\begin{proof}
	The proof follows trivially from definitions \ref{par_sol} and \ref{orth}.
\end{proof}

Equipped with this lemma, we can now show that when a trajectory enters a restricted solution orthant with the fast variable being $\mbf{x_f} = \mbf{0}$, then it is guaranteed to converge to a solution plane, which further implies that it will evolve into a fixed point (see proposition \ref{fix_sol}).

\begin{theorem}[Basin of Attraction]
	\label{basin_th}
	Let $\mbf{v'}$ be a solution vector, then
	\begin{equation*}
	\Big[ \bigcup_{\mcr{I}} J'(\mbf{v'}, \mcr{I}) \Big] \times \{ 0 \}^m \times [1, +\infty)^m \subseteq B(\mbf{v'}).
	\end{equation*}
\end{theorem}

\begin{proof}
	It is sufficient to show
	\begin{equation*}
	\forall \mcr{I}, \qquad
	J'(\mbf{v'}, \mcr{I}) \times \{ 0 \}^m \times [1, +\infty)^m \subseteq B(\mbf{v'}).
	\end{equation*}
	WLOG, we let $\mbf{v'} = \mbf{+1}$ and the proper isolated index set be $\mcr{I} = [[1,n']]$, then $J'(\mbf{v'}) = J'(\mbf{+1}) = [ 1-2\gamma, +1 ] ^{n'}\times [-1,+1]^{n-n'}$, which we simply refer to as $J'$ from here on. We first note that if $\mbf{x_s} = \mbf{0}$, then for $\forall \mbf{v} \in J'$, we have $\dot{\mbf{v}}(0) \geq 0$ (see the first line of Eqs.~(\ref{mem})). Furthermore, $\forall \mbf{v} \in J'$, we have $\mbf{\dot{x}_f} \leq 0$ (which follows from the second line of Eqs.~(\ref{mem}) and Lemma \ref{c14}). We first show, by contradiction, that for any point initialized in the supposed subset of the basin, then the evolution of each isolated component of the voltage vector must be weakly monotonous, or $\dot{v}_i(t) \geq 0$ for $\forall i \in [[1,n']]$ and $\forall t>0$. \\
	
	We let some initial point be $\mbf{x}_0 = \{ \mbf{v}_0, \mbf{x}_{f,0}, \mbf{x}_{s,0} \} \in J' \times \{ 0 \}^m \times [1, x_{\max}]^m$, and the solution trajectory be $\mbf{x}(t)$. WLOG, we assume that $v_1(t)$ is not monotonously increasing, and is the first voltage in time to violate the inequality $\dot{v}_1(t) \geq 0$. We denote this time to be
	\begin{equation*}
	T = \inf\{ t>0 \lcond \exists \dot{v}_1(t) < 0 \}.
	\end{equation*}
	It is clear that $v_1(t) \in J'$ for $\forall t \in [0, T]$. In addition, it is required that $\mbf{x_s}(T) \neq \mbf{0}$. This is, however, only possible if $\exists t' \in [0,T]$ such that
	\begin{equation*}
	\exists j \in [[1,m]], \quad \dot{x}_{s,j}(t') > 0.
	\end{equation*}
	But as $\mbf{v}(t') \in J'$, the above condition is not possible. Therefore, by contradiction, we must have $\dot{\mbf{v}}(t) \geq 0$ for $\forall t \geq 0$. \\
	
	To complete the proof, it is sufficient to show that $\forall i \in \mcr{I}$, we have $\lim_{t\to\infty} v_i(t) = +1$. Again, we prove this by contradiction. We first assume that the statement is not true, then $\exists i \in \mcr{I}$ (WLOG let $i=n'-1$), $\exists \epsilon > 0$ such that $\lim_{t\to\infty}v_i(t) = 1-\epsilon$, and $\lim_{t\to\infty}\dot{v}_i(t) = 0$ (as $\mbf{v}$ is monotonous). This means that there is a time $T$, after which $v_i$ can no longer appear as the most satisfied literal in any clause. If this is not the case, then $\forall T$, $\exists t' > T$ such that $\dot{v}_i(t') = v_i(t')$, which is clearly not possible as the limits of the LHS and RHS converge to different values. As $v_i$ is no longer the most satisfied literal in any clause, we can set its value arbitrarily in $[-1,+1]$, and the condition $\mbf{C}(\mbf{v}) \leq \gamma$ will still remain true, as the clause energy of each clause only depends on the most satisfied literal (see Eq.~(\ref{cj})). From Lemma \ref{c14}, this implies that the restricted solution orthant should be $[+\frac{1}{2}, +1]^{n'-1}\times [-1,+1]^{n-n'+1}$ instead. However, the restricted solution orthant of a solution vector is unique given a proper index set $\mcr{I}$ (see the remark of definition \ref{par_sol}), so we have a contradiction. Therefore, the dynamics must converge to a solution plane, and thus also to a fixed point by proposition \ref{fix_sol}.
\end{proof}

\begin{remark}
	Note that this basin of attraction is a superset of the basin of attraction proven in another work~\cite{zoltan} using continuous dynamics for solving k-SAT problems. This means that the basin of attraction for our dynamics is larger, which is a desirable property for using our dynamics as a SAT solver.
\end{remark}

\begin{corollary}
	\label{corr_conv}
	If $\mbf{x'}$ is a fixed point of the memory dynamics given in Eqs. (\ref{mem}), then $\mbf{v'}$ is in a solution plane.
\end{corollary}

\begin{proof}
	If $\mbf{x'}$ is a fixed point, then clearly we require $C_j(\mbf{v'}) \leq \delta$ for $\forall j$, otherwise $x'_{l,j}$ will increase. And since $\delta < \gamma$, we have $C_j(\mbf{v'}) < \gamma$, meaning that $x'_{s,j} = 0$, otherwise $x'_{s,j}$ will decrease. Since $C_j(\mbf{v}) \leq \delta < \frac{1}{2}$, $\sign(\mbf{v'})$ is a solution vector (see Proposition \ref{sign}), and $\mbf{x'} \in B(\sign(\mbf{v'}))$ (see Theorem \ref{basin_th}). If $\mbf{v'}$ is in a solution plane, then we are done; if not, then $\mbf{x'}$ necessarily evolves to a solution plane by Theorem \ref{basin_th}, meaning that it cannot be a fixed point, creating a contradiction. Therefore, $\mbf{v'}$ must already be in a solution plane to begin with.
\end{proof}

\section{Dynamic Voltage Flow}
\label{reduce}

Often times we are only interested in the convergent properties of the voltage dynamics, as they correspond directly to the state of Boolean variables. On the other hand, the evolution of the memory variables is important in influencing the trajectory of the voltages {\it indirectly} by changing the strength of the gradient-like and rigidity terms (see Eqs.~(\ref{mem})). It then makes sense to restrict our attention to only the components of the flow field that govern the dynamics of $\mbf{v}$ directly, which we can denote as $F_{\mbf{v}}$, and refer to as {\it reduced flow field} in the voltage space, or simply the {\it voltage flow}. It should be noted that this flow is not autonomous and is, in fact, dynamically governed by the memory. In this Section, we establish the tools needed to study such reduced flow field, which we will use to show certain properties such as the absence of periodic orbits (see Section \ref{no_period}) in the voltage space. For the remainder of this material, we shall assume that the full flow field is always Lipschitz continuous to simplify discussion, since we have seen in Section \ref{plane} that the existence of measure-zero discontinuities does not alter significantly the behavior of our dynamics. Often times, we will focus on the simplified dynamics as given in Eq.~(\ref{mem_sim}) and assume continuity for $\mbf{G}$ and $\mbf{C}$.

\subsection{Reduced Flow}
\label{red_flow}

In this Section, we will first factor the full phase space into the {\it reduced} space and the {\it auxiliary} space, which will allow us to formalize the notion of a reduced flow field. Visually, the reduced flow field can be viewed as the full flow field ``projected" onto a subspace. We proceed with the following series of definitions.

\begin{definition}
	Let $X$ be a set and $Y = \prod_{j=1}^{m}Y_j$ be a product space. Given any mapping $F: X \mapsto Y$, we define the {\bf decomposition} of $F$ as $F_j: X \mapsto Y_j$ for $\forall j \in [[1,m]]$ such that
	\begin{equation*}
	\forall x\in X, \qquad F(x) = \big( F_1(x), F_2(x), ... , F_m(x) \big).
	\end{equation*}
\end{definition}

\begin{definition}
	Let $F = (F_1,F_2): \mbb{R}^n\times\mbb{R}^m \mapsto \mbb{R}^n\times\mbb{R}^m$ be a flow field, and $\mbf{x}(t,\mbf{x}_0)$ be a trajectory under this flow field. If we denote $\mbb{R}^n$ as the {\bf reduced space} and $\mbb{R}^m$ as the {\bf auxiliary space}, then we define the {\bf reduced trajectory} and the {\bf auxiliary trajectory}, $\mbf{x}_1(t,\mbf{x}_0) \in \mbb{R}^n$ and $\mbf{x}_2(t,\mbf{x}_0) \in \mbb{R}^m$, such that
	\begin{equation*}
	\mbf{x}(t,\mbf{x}_0) = \big( \mbf{x}_1(t,\mbf{x}_0), \mbf{x}_2(t,\mbf{x}_0) \big).
	\end{equation*}
\end{definition}

\begin{definition}[Reduced Flow]
	Let $F = (F_1,F_2): \mbb{R}^n\times\mbb{R}^m \mapsto \mbb{R}^n\times\mbb{R}^m$ be a flow field, and $\mbf{x}(t,\mbf{x}_0)$ be a trajectory under this flow field. For a given initial point $\mbf{x_0}$, we construct the {\bf reduced flow field} $F_r: \mbb{R}\times\mbb{R}^n$ such that
	\begin{equation*}
	\forall t \geq 0, \qquad
	F_r(t,\mbf{x}_{0,1}) = F_1\big( \mbf{x}_{0,1}, \mbf{x}_2(t, \mbf{x}_0) \big).
	\end{equation*}
\end{definition}

\begin{remark}
	It can be easily verified that the reduced flow field is well defined at any given time. Visually, if we view the reduced space as a hyperplane that ``cuts" the full flow field, then the reduced flow field is simply the ``cross section" of the field in the plane. As the auxiliary variables evolve in time, the plane will move in the auxiliary space, or simply some direction orthogonal to the plane, thereby changing the cross section. We then see that the reduced flow field is effectively a dynamic flow field governed by the auxiliary trajectory.
\end{remark}

\subsection{Flow Kernel and Complement}

In this subsection, we will relate the dynamics of the reduced flow field to the auxiliary trajectory explicitly. We will focus specifically on the case where the reduced flow field is linear in the auxiliary variables, with the simplified memory dynamics in Eq.~(\ref{mem_sim}) as an example. In particular, at every point in the reduced space, the auxiliary space can be factored into two subspaces, one of them in which the auxiliary trajectory can evolve without affecting the reduced flow field. We refer to this subspace as the {\it flow kernel} of the reduced flow field at that point, and the other factor subspace as the {\it flow complement}. We formally define the two subspaces as follows:

\begin{definition}[Flow Kernel]
	\label{kernel}
	Let $F = (F_1,F_2): \mbb{R}^{n}\times \mbb{R}^m \mapsto \mbb{R}^n\times \mbb{R}^m$ be a vector field. If $F_1(\mbf{x}_1, \mbf{x}_2)$ is linear in $\mbf{x}_2$, then we define
	\begin{equation*}
	K_{F_1}(\mbf{x}_1) = \{ \mbf{x}_2 \in \mbb{R}^m \,\,\big\lvert\,\, F_1(\mbf{x}_1, \mbf{x}_2)=0 \},
	\end{equation*}
	as the {\bf flow kernel} of $F_1$ at $\mbf{x}_1$.
\end{definition}

\begin{remark}
	Clearly, the flow kernel is a vector space. In fact, given any fixed $\mbf{x}_1$, the operation $F_1(\mbf{x}_1, \mbf{x}_2)$ can be regarded as a mapping from $\mbb{R}^m$ to $\mbb{R}^n$ via an $n\times m$ matrix, with its kernel being the flow kernel. Let the rank of the matrix be $n' \leq n$. If $n'\geq m$, then clearly the kernel is trivial. On the other hand, if $n' < m$, then the dimension of the kernel is $m-n'$ by the rank-nullity theorem. Hard 3-SAT instances generally are at clause ratios near $4$, meaning that $m\approx 4n$ (if we let the voltage space be the reduced space), and the kernel is generally non-trivial. \\
\end{remark}

\begin{definition}[Flow Complement]
	\label{complement}
	Let $F = (F_1,F_2): \mbb{R}^{n}\times \mbb{R}^m \mapsto \mbb{R}^n\times \mbb{R}^m$ be a vector field. We refer to the orthogonal complement of the flow kernel at point $\mbf{x_1}$,
	\begin{equation*}
	J_{F_1}(\mbf{x}_1) = \{ \mbf{x}_2 \in \mbb{R}^m \,\,\big\lvert\,\, \forall \mbf{x}'_2 \in K_{F_1}(\mbf{x}_1), \quad \mbf{x}_2 \cdot \mbf{x}'_2 = 0 \},
	\end{equation*}
	as the {\bf flow complement} of $F_1$ at $\mbf{x}_1$.
\end{definition}

\begin{remark}
	For any fixed $\mbf{x}_1$, if the domain of $F_1(\mbf{x}_1, \mbf{x}_2)$ is restricted to $J_{F_1}(\mbf{x}_1)$, then the mapping is invertible. In other words, there is a bijection from every configuration of the auxiliary variables in the flow complement to every possible flow vector. In some sense, the component of the auxiliary variable in the flow complement space is the only relevant component generating the reduced flow.
\end{remark}

\begin{definition}[Auxiliary Relevance]
	\label{relevant}
	Let $F = (F_1,F_2): \mbb{R}^{n}\times \mbb{R}^m \mapsto \mbb{R}^n\times \mbb{R}^m$ be a vector field. Given $\mbf{x}_1 \in \mbb{R}^n$ and $\mbf{x}_2 \in \mbb{R}^m$, we refer to the projection of $\mbf{x}_2$ to $K_{F_1}(\mbf{x}_1)$ as the {\bf irrelevant component}, and the projection to $G_{F_1}(\mbf{x}_1)$ as the {\bf relevant component}, which we denote as $\mbf{x}_2^{\ast}$.
\end{definition}

\begin{remark}
	Given a reduced flow field that is linear in the auxiliary variables, it can be shown that the time derivative of the field is zero at time $t$ and location $\mbf{x}_{0,1}$, if and only if the auxiliary variable evolves in the flow kernel of $F_1$, or
	\begin{equation*}
	\dot{\mbf{x}}_2(t,\mbf{x}_0) \in K_{F_1}(\mbf{x}_1).
	\end{equation*}
	Equivalently, this means that the time derivative of the relevant component of $\mbf{x}_2$ must be zero, or
	\begin{equation*}
	\dot{\mbf{x}}^{\ast}_2(t,\mbf{x}_0) = 0.
	\end{equation*}
\end{remark}

\subsection{Unstable Non-solution Fixed Points}
\label{v_flow}

In proposition \ref{fix_sol}, it was shown that every fixed point in the full phase space $\mbb{R}^{n+2m}$ must correspond to a solution of the 3-SAT problem. However, this still leaves open the possibility of the existence of fixed points in the voltage space that correspond to a non-solution. Most of the time, when the dynamics fall into such fixed points, the memory breaks this fixed point by reweighing the clause functions, thereby evolving the reduced flow vector to a non-zero value, effectively freeing the voltage dynamics. However, in very rare instances, the memory variables may evolve in the flow kernel, in which case the voltages may be permanently trapped. Here, we show the unlikeliness of being trapped in such fixed points in {\it general}, and the {\it instability} of the gradient-like influence near fixed points. \\

For simplicity, we focus on the simplified memory dynamics as given in Eq.~(\ref{mem_sim}): \footnote{If we were to extend the analysis of this subsection to the full memory dynamics (by including the rigidity and fast memory dynamics), then the RHS to $\dot{\mbf{v}}$ can be decomposed into two terms, one quadratic in $\mbf{x}$ and the other being only dependent on $\mbf{v}$. The equation $\mbf{\dot{v}} = \mbf{0}$ would still be a polynomial equation for $\mbf{x}$, and the solution space of $\mbf{x}$ can be similarly decomposed into a hyperface defined by the corresponding algebraic variety and its complement, and the analysis in this subsection can be easily extended for the full memory dynamics as well by considering the local tangent space. }
\begin{equation*}
\begin{split}
\dot{\mbf{v}} &= \mbf{G}(\mbf{v})\mbf{x} \\
\mbf{\dot{x}} &= \alpha \mbf{C}(\mbf{v}).
\end{split}
\end{equation*}
We here temporarily relax the specific forms of functions $\mbf{C}$ and $\mbf{G}$ (Eqs. (\ref{cj}) and (\ref{G}) respectively), and simply require they be general and non-singular. A voltage fixed point means that $\dot{\mbf{v}} = \mbf{0}$, implying that the memory must be in the flow kernel, or $\mbf{x} \in K(\mbf{v})$. If the condition $\dot{\mbf{v}} = \mbf{0}$ is to be held in time (or $\ddot{\mbf{v}} = \mbf{0}$), then the memory must also evolve in the flow kernel, or $\dot{\mbf{x}} \in K(\mbf{v})$. Equivalently, $\mbf{G}(\mbf{v}) \cdot \mbf{C}(\mbf{v}) = \mbf{0}$. The LHS is simply a $\mbb{R}^n \mapsto \mbb{R}^n$ mapping, so the preimage of $\mbf{0}$ consist of finitely many points in {\it general}, and they constitute a measure-zero set in $\mbb{R}^n$. This shows the unlikeliness of the dynamics being trapped in a non-solution fixed point.\\

To show that the gradient-like influence is unstable, we first note that a Jacobian element of the gradient-like term can be written as
\begin{equation*}
\begin{split}
\mcr{J}_{ij} = \sum_{k} x_k \partial_{v_i}G_{jk}(\mbf{v}) 
&= \frac{1}{2} \sum_k x_k q_{jk} \partial_{v_i} \min_{\{j' \neq j \,|\, q_{j'k} \neq 0\}} (1 - q_{j'k}v_{j'}),
\end{split}
\end{equation*}
where the last equality is from Eq. (\ref{G}), and derivations across the discontinuous hyperplanes are neglected. In this form, it is clear that the diagonal elements of the Jacobian are zero, or $\mcr{J}_{ii} = 0$ for $\forall i$. To see this, we simply note that $\partial_{v_i} v_{j'} = \delta_{ij'}$, and the condition $i \neq j'$ imposed by the $\min$ function. This means that the trace of the Jacobian is zero, meaning that any fixed point cannot be stable (otherwise the Jacobian would necessarily be negative in the real component of the trace). 

\section{Non-periodicity of Dynamics}
\label{no_period}

In dimensions greater than 2, a dissipative system\footnote{See Section \ref{diss} for a detailed discussion of the dissipativeness of the memory dynamics.} may admit periodic orbits. Therefore, we shall show the absence of periodic orbits explicitly in this Section. This result directly precludes the possibility of chaos (see Section \ref{chaos}). We formulate the proof of non-periodicity on the voltage space by making use of the formalism developed in Section \ref{reduce}. Note that showing the absence of periodic orbits in the full state space (voltages plus memories) is not sufficient for our purpose as it does not preclude the existence of periodic orbits in the reduced voltage space, which is directly relevant to the solution of the 3-SAT problem. For analytic convenience, we shall assume all mentioned fields in this Section is sufficiently well-behaved (i.e., Lipschitz continuous in space and continuous in time) such that it admits a unique classical solution for all initial values.

\subsection{Generalized Periodicity}

As the voltage dynamics by itself is not autonomous (since it is memory dependent), we first have to construct a non-standard definition of periodicity for dynamic fields that suffices in the context of optimization. In general, a dynamic field admits periodic orbits of non-constant periods. We first recall that the classical definition of periodicity for static fields is given as follows, and generalize this definition for dynamic fields.

\begin{definition}[Regular Periodic Orbit]
	Let $\mbf{x}: [0,+\infty) \mapsto \mbb{R}^n$ be a trajectory with initial value $\mbf{x}_0$. The trajectory is said to be {\bf periodic} if $\exists T > 0$ such that
	\begin{equation*}
	\forall t \geq 0, \qquad \mbf{x}(t + T, \mbf{x_0}) = \mbf{x}(t, \mbf{x_0}).
	\end{equation*}
	The {\bf periodic orbit} of $\mbf{x}_0$ is
	\begin{equation*}
	\gamma_x = \{ \mbf{x}(t, \mbf{x_0}) \,\,\big\lvert\,\, t\in [0,T) \},
	\end{equation*}
	and the {\bf period} of this orbit is $T$.
\end{definition}

\begin{remark}
	It is fairly easy to show the following
	\begin{equation*}
	\begin{split}
	\gamma_x 
	&= \{ \mbf{x}(t) \,\,\big\lvert\,\, t\in [0,T) \} \\
	&= \{ \mbf{x}(t) \,\,\big\lvert\,\, t\geq 0\}
	= \gamma_x^{+},
	\end{split}
	\end{equation*}
	meaning that the periodic orbit is also the maximal positive orbit of $\mbf{x}_0$, which makes sense because the trajectory cannot escape the periodic orbit even given infinite time. This property generalizes naturally for dynamic fields.
\end{remark}

\begin{definition}[Speed]
	\label{speed}
	Let $\mbf{x}: [0,+\infty) \mapsto \mbb{R}^n$ be some trajectory with initial value $\mbf{x}_0$. If the trajectory is everywhere differentiable in time, we define the {\bf velocity} along the trajectory to be
	\begin{equation*}
	\dot{\mbf{x}}(t,\mbf{x}_0),
	\end{equation*}
	and the {\bf speed} to be
	\begin{equation*}
	s(t,\mbf{x}_0) = || \mbf{\dot{x}}(t,\mbf{x}_0) ||.
	\end{equation*}
\end{definition}

\begin{remark}
	Clearly, the velocity and speed of the trajectory is also periodic with the same period as the trajectory itself. If the trajectory is governed by the flow field $F$, then the period of the orbit is given by the following contour integral
	\begin{equation*}
	T = \oint\limits_{\gamma} \frac{|| d\mbf{x} ||}{|| F(\mbf{x}) ||}.
	\end{equation*}
	This integral is well-defined for a static flow field, but it is no longer well defined if $F$ is explicitly time dependent, in which case the period may be time-dependent as well.
\end{remark}

\begin{definition}[General Periodic Orbit]
	\label{gen_per}
	Given a time-dependent flow field $F: \mbb{R} \times \mbb{R}^n \mapsto \mbb{R}^n$, a {\bf general periodic orbit} is said to exist for $\mbf{x_0}$ if $\exists T$ such that
	\begin{equation*}
	\mbf{x}(T,\mbf{x_0}) = \mbf{x_0}.
	\end{equation*}
	Let the periodic orbit be $\gamma_x = \{ \mbf{x}(t) \,|\, t\in [0, T) \}$, then it is required that $\mbf{x}(t,0) \in \gamma_x$ for $\forall t > 0$. Furthermore, $\exists (t_1,t_2) \in \{ (s_1,s_2)\in [0,+\infty)^2 \,|\, s_1 \neq s_2 \}$ such that
	\begin{equation*}
	\mbf{x}(t_1) = \mbf{x}(t_2) \quad \land \quad \dot{\mbf{x}}(t_1), \dot{\mbf{x}}(t_2) \neq 0
	\end{equation*}
	
	For a given time $t$, we let
	\begin{equation*}
	T'(t) = \inf\{ t'>t \,\,\big\lvert\,\, \mbf{x}_1(t',\mbf{x}_0) = \mbf{x}_1(t, \mbf{x}_0) \},
	\end{equation*}
	then the {\bf period} at time $t$ is given as $T(t) = T'(t) - t$. If $T(t)$ is a constant in time, then the periodic orbit is said to be {\bf regular}; otherwise, the periodic orbit is said to be {\bf irregular}.
\end{definition}

\begin{remark}
	Essentially, a general periodic orbit is a closed trajectory which contains the maximal positive solution. Furthermore, there must be a point and its neighborhood on the orbit which the trajectory visits two separate times, with the period simply being the time duration until the next visit. Technically, the period can be zero if the flow field is zero at that particular time and point and infinite if the trajectory never revisits the point, but there must be at least one point in time where the period is positive finite.
\end{remark}

\begin{lemma}
	\label{int_aux}
	If a dynamic flow field $F: \mbb{R} \times \mbb{R}^n \mapsto \mbb{R}^n$ admits an irregular periodic orbit, then $\exists \mbf{x}_0$ such that $\exists t_1 \geq 0$, $\exists t_2 \in \{ t'>t_1 \,\,\big\lvert\,\, \mbf{x}(t',\mbf{x_0}) = \mbf{x}(t_1,\mbf{x_0}) \}$,
	\begin{equation*}
	\exists k \neq 1
	\qquad
	F(t,\mbf{x}(t_2)) = kF(t,\mbf{x}(t_1)) \neq 0.
	\end{equation*}
\end{lemma}

\begin{proof}
	The proof is omitted. See remark instead.
\end{proof}

\begin{remark}
	Essentially, there must be at least one point on the periodic orbit where the dynamic flow field align (or anti-align) with itself at two separate times. This is clearly required so that the trajectory can ``revisit" the orbit in the neighborhood of that point.
\end{remark}

\begin{corollary}[Change in Relevant Component]
	\label{change_relevant}
	Given a static flow field $F = (F_1, F_2): \mbb{R}^n\times \mbb{R}^m \mapsto \mbb{R}^n\times \mbb{R}^m$, if $F_1(\mbf{x}_1, \mbf{x}_2)$ is linear in $\mbf{x}_2$ and an irregular periodic orbit $\gamma$ exists in $\mbb{R}^n$, then $\exists \mbf{x}_0 \in \mbb{R}^{n+m}$ such that $\exists t_1 \geq 0$, $\exists t_2 \in \{ t'>t_1 \,\,\big\lvert\,\, \mbf{x_1}(t',\mbf{x_0}) = \mbf{x_1}(t_1,\mbf{x_0}) \}$,
	\begin{equation*}
	\exists k \neq 1
	\qquad
	\mbf{x_2}^{\ast}(t,\mbf{x_0}) = k\mbf{x_2}^{\ast}(t,\mbf{x_0}) \neq 0.
	\end{equation*}
	where $\mbf{x}_2^{\ast}(t, \mbf{x}_0 )$ is the relevant component of $\mbf{x}_2$ at time $t$ as defined in Definition \ref{relevant}.
\end{corollary}

\begin{proof}
	The proof follows directly from Definition \ref{relevant} and Lemma \ref{change_relevant}.
\end{proof}

\begin{remark}
	In terms of the relevant component of the auxiliary variable, the periodic orbit is regular if and only if
	\begin{equation*}
	\forall t \geq 0,\, \exists T>0 \qquad
	\mbf{x}^{\ast}_2\big( t + T, \mbf{x}_0 \big) = \mbf{x}_2^{\ast}(t,\mbf{x}_0).
	\end{equation*}
\end{remark}

\subsection{Absence of Irregular Periodic Orbits}

In the previous subsection, we have seen that a periodic orbit under a dynamic field can be categorized as either being regular or irregular. To show the absence of periodic orbits, we treat the two cases separately, as they require different proof techniques. In this subsection, we focus on the irregular case, which requires the physical notion of speed as defined in definition \ref{speed}; we treat the regular case in the next subsection, by formulating the problem in the geometric context of hypersurface intersections. \\

For the sake of simplicity, we again focus on the simplified dynamics as given in Eq.~(\ref{mem_sim})\footnote{The rigidity influence is negligible in the periodicity analysis as the dynamics is dominated by the gradient-like term when the system is continuously in an unsatisfied state, which is clearly the case when the dynamics is trapped in a periodic orbit. },
\begin{equation*}
\begin{split}
\dot{\mbf{v}} &= \mbf{G}(\mbf{v}) \mbf{x} \\
\mbf{\dot{x}} &= \alpha \mbf{C}(\mbf{v}).
\end{split}
\end{equation*}
We require the functions $\mbf{C}$ and $\mbf{G}$ to be everywhere differentiable\footnote{Note that the actual gradient-like term $\mbf{G}$ defined in Eq. (\ref{G}) is everywhere differentiable except at certain hyperplanes which constitute a measure-zero set in the voltage space (see section \ref{plane}). It is easy to see that the presence of these hyperplanes will not affect the periodicity analysis.} (which automatically guarantees Lipscthiz continuity). This also guarantees that any image of a compact set in $\mbb{R}^n$ is bounded above in norm. For $\mbf{C}$, we can assume that it is bounded below in norm also, otherwise $\mbf{C}=0$ implies that the trajectory is in a solution plane in which case it must converge to a fixed point (see Proposition \ref{fix_sol})\footnote{In fact, it can be assume that $|\mbf{C}(\mbf{v})| \geq \delta$ for the full equations by Proposition \ref{basin_th}.}. To make the definition of speed (see Definition \ref{speed}) useful, we first have to formally define the generalized concept of {\it location} for trajectories governed by a dynamic field.
\begin{definition}[Location]
	\label{location}
	Let $\mbf{x}(t,\mbf{x_0})$ be a classical solution to a flow field $F: \mbb{R}\times\mbb{R}^n \mapsto \mbb{R}^n$, then $\forall t' \in \{ t\in\mbb{R} \,|\, \mbf{x}(t,\mbf{x_0}) \neq 0 \}$, there is $\exists \delta t > 0$ such that a unique isometry $\varphi: \mbf{x}\big( [t'-\delta t,t'+\delta t], \mbf{x_0} \big) \mapsto \mbb{R}$ exists locally. For $\forall t \in [t'-\delta t,t+\delta t]$, we refer to $u(t) = \varphi(t)$ as the {\bf location} on the trajectory around time $t'$.
\end{definition}

\begin{remark}
	The location is quite literally the location on the real line if we unwind the trajectory locally to a straight line. As long as the trajectory keeps moving ``forward" in some time interval, then the mapping $\varphi$ is bijective, meaning that there is a one-to-one correspondence between time and location. This is why the condition $\mbf{x}(t,\mbf{x_0}) \neq 0$ is required locally.
\end{remark}

\begin{lemma}
	\label{s_diff}
	In a time interval in which location can be defined, the speed is differentiable with respect to location in the interval. In other words, the mapping $s \circ u^{-1}$ is locally differentiable.
\end{lemma}

\begin{proof}
	First of all, we have $s(t) = || \dot{\mbf{v}}(t) || \neq 0$ in the time interval, and we note that
	\begin{equation*}
	\ddot{\mbf{v}} = \frac{d}{dt} \Big( \mbf{G}(\mbf{v})\mbf{x} \Big) = \dot{\mbf{v}}\mbf{G}'(\mbf{v})\mbf{x} + \mbf{G}(\mbf{v})\dot{\mbf{x}},
	\end{equation*}
	which is well-defined as $\mbf{G}$ is everywhere differentiable, meaning that $s(t)$ is differentiable with respect to $t$ (as long as $s(t)\neq 0$). Furthermore, we note that $u'(t) = s(t)$, meaning that $u$ is also differentiable with respect to time. Since $s(t) \neq 0$, the inverse $u^{-1}$ is differentiable in the interval as well. Therefore, $s \circ u^{-1}$ is differentiable in the interval, as the composition of two differentiable mappings.
\end{proof}

\begin{theorem}
	\label{no_irregular}
	An irregular orbit does not exist in the voltage space.
\end{theorem}

\begin{proof}
	We prove this by contradiction, by assuming that an irregular orbit does exist. Then by Lemma \ref{int_aux}, there is $\exists \mbf{v_0}$ such that $k\dot{\mbf{v}}(t_1) =  \dot{\mbf{v}}(t_2) \neq 0$, where $t_2 \neq t_1$, $\mbf{v}(t_1) = \mbf{v}(t_2) = \mbf{v_0}$, and WLOG $k \in (0,1)$. We let $u_0$ be the location of $\mbf{v_0}$, $\delta u$ be the infinitesimal change in location from $u_0$. Furthermore, we let $s_1$ be the speed at time $t_1$, and $\delta s$ be the change in speed with respect to $\delta u$ (which is well-defined as shown in Lemma \ref{s_diff}). If we let $\delta \theta$ be the change in direction, then we have the following equality
	\begin{equation*}
	\begin{split}
	(s_1+\delta s)^2 + s_1^2 - 2s_1(s_1+\delta s)\cos(\delta \theta) &= || \dot{\mbf{v}}'(u_0) ||^2 \delta u^2 \\
	\iff \quad s_1^2 \delta\theta^2 + \delta s^2 &= \norm{ \mbf{G}'(u_0)\mbf{x}(t_1) + \frac{\alpha}{s_1}\mbf{G}(u_0)\mbf{C}(u_0) }^2 \delta u^2,
	\end{split}
	\end{equation*}
	where we discarded third order terms on the LHS and applied Eq. (\ref{mem_sim}) and $u'(t)=s(t)$ on the RHS. Rearranging the terms gives us
	\begin{equation*}
	\big( s'(u_0) \big)^2
	= \norm{ \mbf{G}'(u_0)\mbf{x}(t_1) + \frac{\alpha}{s_1}\mbf{G}(u_0)\mbf{C}(u_0) }^2
	- s_1^2 \big( \theta'(u_0) \big)^2.
	\end{equation*}
	
	If we let the speed at time $t_2$ be $s_2$ (with $s_2 = ks_1$), then the above relation will hold similarly. We can assume that $s'(u_0)=0$ at $t_2$, which is justified as $s$ is bounded and differentiable everywhere. Since $s_1 > s_2$, and $\{\mbf{x},\mbf{C},\mbf{G}\}$ are everywhere differentiable, we must have $\big(s'(u_0)\big)^2 < 0$ at time $t_1$ by the above relationship, which is clearly impossible. Therefore, an irregular orbit cannot exist.
	
\end{proof}

\subsection{Absence of Regular Periodic Orbits}

In the previous subsection, we showed the absence of irregular orbits, so if a periodic orbit were to exist in the voltage space, it must be a regular periodic orbit. In this Section, we show that the existence of a periodic orbit is also absent in general. The problem can be described geometrically where a regular orbit can be described as the intersection between two low dimensional hypersurfaces in a high-dimensional space, which cannot occur if the two surfaces are in general positions.

\begin{theorem}
	\label{no_regular}
	In general, a periodic orbit does not exist in the voltage space.
\end{theorem}

\begin{proof}
	In theorem \ref{no_irregular}, we showed the absence of irregular periodic orbits in the voltage space, it is then sufficient to show that a regular periodic orbit is absent as well. \\
	
	If a regular orbit were to exist in the voltage space, then we can denote its period as $T$, and its initial point as $\mbf{x}_0 \in \mbb{R}^{n+m}$. The change in the memory over a period from time $t$ is then given by
	\begin{equation*}
	\begin{split}
	&\mbf{x}(t+T, \mbf{x}_0) - \mbf{x}(t, \mbf{x}_0) \\
	= &\int_{t}^{t+T} \mbf{\dot{x}}(s, \mbf{x}_0)\,ds \\
	= &\alpha \int_t^{t+T} \mbf{C}\big( \mbf{v}(s, \mbf{x}_0) \big)\, ds \\
	= &\alpha \int_0^T \mbf{C}\big( \mbf{v}(s, \mbf{x}_0) \big)\, ds, \qquad \forall t \geq 0,
	\end{split}
	\end{equation*}
	where in the last equality, we used the periodicity of $\mbf{v}$ to remove the explicit dependency on $t$ in the integral bounds. This allows us to simply set the result as some constant vector $\mbf{K}$ that is constant in time. \\
	
	Clearly, $\mbf{K}$ must be in the flow kernel of the reduced flow field $\forall t \geq 0$ (otherwise the velocity would not be the same after a period), which implies
	\begin{equation}
	\label{kG}
	\forall \mbf{v}\in\gamma, \qquad
	\mbf{G}( \mbf{v} )\,\mbf{K} = \mbf{0}.
	\end{equation}
	It can be assumed that any sensible matrix $\mbf{G}$ dictating the evolution of the voltages must coincide with the polarity matrix $\mbf{Q}$ exactly in its nonzero elements, so $\mbf{G}: \mbb{R}^n \mapsto \mbb{R}^{3m}$ as there are $m$ clauses and $3$ literals per clause, which gives us an injective mapping. Furthermore, it can be assumed that the mapping is $C^{\infty}$ diffeomorphic and general so that the image of $\mbb{R}^n$ is a smooth $n$-dimensional hypersurface at general position in $\mbb{R}^{3m}$. On the other hand, condition (\ref{kG}) is a system of $m$ linear equations, so the set of all matrices (with nonzero elements matching the polarity matrix) solving the system for a given $\mbf{K}$ forms a $3m-m = 2m$ dimensional hyperplane in $\mbb{R}^{3m}$, which is also in general position as $\mbf{K}$ is general. The $n$-dimensional hypersurface generated by the voltages and the $2m$ dimensional hyperplane do not intersect if they are in general position, as
	\begin{equation*}
	2m + n < 3m,
	\end{equation*}
	where we have assumed $n<m$ (see the remark of definition \ref{kernel}).
\end{proof}

\begin{remark}
	Note that the dimension of the surface containing the periodic orbit must be at least 2, which means that the intersection of the two surfaces in $\mbb{R}^{3m}$ must be at least 2 dimensional as well, and this makes the existence of periodic orbits even less likely. Even assuming that $n>m$, meaning that the intersection of the two surfaces is non-trivial, the existence of a periodic orbit in the voltage space is still unlikely. We require an initial memory value that generates a reduced flow field that guarantees the containment of the voltage trajectory completely in the intersection, which does not exist in general.
\end{remark}

\subsection{Absence of Chaos}
\label{chaos}

Devaney's definition of chaos~\cite{Devaney} requires a dense set of periodic orbits and topological transitivity. We have shown in the previous Section that periodic orbits are not supported by the dynamics defined by Eqs.~(\ref{mem}), and this directly precludes the existence of chaos. We can then state the following:

\begin{corollary}[Absence of Chaos]
	\label{cj_inv}
	The voltage dynamics are non-chaotic.
\end{corollary} 

\begin{remark}
	Although not required to show the absence of chaos, we also note that the voltage dynamics are not topologically transitive if a solution of the 3-SAT problem exists which implies the existence of fixed points (See Proposition \ref{fix_sol}). This is because the dynamics are confined in a compact positive invariant set $O$ with nonempty interior (see Section \ref{compact}) and there is at least one fixed point in that set with its $\omega$-limit set that is not $O$. 
\end{remark}

\section{Dissipativeness}
\label{diss}

A rather important property of the memory dynamics is dissipativeness. In other words, the measure (or volume) of an initial set contracts under the flow field, eventually evolving to a measure zero set. To show dissipativeness for well-behaved (everywhere differentiable) vector fields, it is sufficient to show that the divergence is negative everywhere. However, our dynamics are governed by a discontinuous flow field, so we have to carefully account for the regions of discontinuities (see Section \ref{plane}). 

\subsection{Preliminaries}

Before we discuss the dissipative property of the memory dynamics, we first have to formally define the notion of dissipativeness for a continuous dynamical system.

\begin{definition}[Dissipativeness]
	\label{def_diss}
	Given a vector field $F: \mbb{R}^n \mapsto \mbb{R}^n$ that admits a positive solution, we let the corresponding time mapping be $T_s: \mbb{R}^n \mapsto \mbb{R}^n$ for $\forall s \geq 0$. Let $\Omega_0 \subseteq \mbb{R}^n$ be a domain of nonzero measure, then $\forall s \geq 0$, we denote $\Omega(s,\Omega_0)$ as the following
	\begin{equation*}
	\Omega(s,\Omega_0) = T_s(\Omega_0) = \{ T_s(\mbf{x}_0) \,|\, \mbf{x_0} \in \Omega_0 \},
	\end{equation*}
	and the measure $\mu(s,\Omega_0)$ as the following
	\begin{equation*}
	\mu(s,\Omega_0) = \mu\big( \Omega(s,\Omega_0) \big).
	\end{equation*}
	If the following is true
	\begin{equation*}
	\forall \Omega_0, \quad \forall s > 0, \quad \mu(s,\Omega_0) < \mu(0,\Omega_0),
	\end{equation*}
	then the system is said to be {\bf dissipative}. If the last inequality is not strict, then the system is said to be {\bf weakly dissipative}.
\end{definition}

\begin{remark}
	If $\mu(s,\Omega_0)$ is everywhere differentiable in $s$, then it is possible for us to quantify the rate of volume contraction as the following forward time derivative
	\begin{equation*}
	\forall s\geq 0, \qquad \dot{\mu}(s,\Omega_0) = \lim_{t\to 0^{+}}\frac{1}{t}\big( \mu(s+t,\Omega_0) - \mu(s,\Omega_0) \big).
	\end{equation*}
	An equivalent definition of a dissipative system would then be the following
	\begin{equation*}
	\forall \Omega_0, \quad \dot{\mu}(0,\Omega_0) < 0,
	\end{equation*}
	meaning that any initial domain must continually shrink in time. For the sake of clarity, we can discard the trivial argument $s=0$ and simply write $\dot{\mu}(\Omega_0) = \dot{\mu}(0,\Omega_0)$ from here on. \\
\end{remark}

It is well known that a bounded domain can be approximated\footnote{Here, we are speaking of approximation in the measure-theoretic sense. In other words, the measure of the domain and the measure of its approximation are the same.} as the union of regular domains~\cite{topology} (see Definition \ref{regular}). Therefore, to show that a system is dissipative, it is sufficient to show that the volume of any regular domain contracts under the flow field. In the case where the flow field is in $C^1$, the mapping $T_s$ is diffeomorphic for $\forall s>0$, meaning that the shape of the boundary will be preserved (being always diffeomorphic to a sphere), allowing us to make use of the following Lemma.

\begin{lemma}
	\label{div_diss}
	Given a vector field $F: \mbb{R}^n \mapsto \mbb{R}^n$ differentiable everywhere, the system is dissipative if the following is true
	\begin{equation*}
	\forall \mbf{x} \in \mbb{R}^n, \qquad \nabla \cdot F(\mbf{x}) < 0.
	\end{equation*}
\end{lemma}

\begin{proof}
	The proof follows directly from divergence theorem
	\begin{equation*}
	\dot{\mu}(\Omega)
	= \int_{\partial\Omega} \big( F(\mbf{x}) \cdot \mbf{n}(\mbf{x}) \big) \,dA
	= \int_{\Omega} \big( \nabla \cdot F(\mbf{x}) \big) \,dV
	< 0,
	\end{equation*}
	where $\Omega \subseteq \mbb{R}^n$ is any smooth domain. This implies that $F$ is dissipative. \\
\end{proof}

\begin{remark}
	The converse of Lemma \ref{div_diss} is almost true, in the sense that if the vector field $F$ contains regions of positive\footnote{Note that it is not sufficient that the divergence be non-negative, as it is possible for the divergence to be zero, forming a closed set. This is why the converse of Lemma \ref{div_diss} is not strictly true.} divergence, then the system cannot be dissipative. To show this, we assume $\exists \mbf{x} \in \mbb{R}^n$ such that $\nabla \cdot F(\mbf{x}) > 0$, then it is clear that the region where the divergence is positive
	\begin{equation*}
	D = \{ \mbf{x}\in\mbb{R}^n \,|\, \nabla\cdot F(\mbf{x}) > 0 \}
	\end{equation*}
	is an open set. This means that for any $\mbf{x}_0 \in D$, $\exists \epsilon >0$ such that the open ball $B_{\epsilon}(\mbf{x}_0) \subset D$, and integrating the divergence over the open ball gives
	\begin{equation*}
	\int_{B_{\epsilon}(\mbf{x}_0)} \big( \nabla \cdot F(\mbf{x}) \big) \,dV > 0,
	\end{equation*}
	implying that $\dot{\mu}\big( B_{\epsilon}(\mbf{x}_0) \big) > 0$, meaning that the system cannot be dissipative. \\
\end{remark}

The analysis in this subsection assumes that the field is differentiable everywhere. In the next subsection, we generalize this analysis to fields that are differentiable everywhere except at certain hyperplanes. This class of fields contains the flow field of the memory dynamics (see Section \ref{plane}).

\subsection{Dissipativeness of Memory Dynamics}
\label{diss_discon}

Recall that the ODE governing the memory dynamics is given in Eqs. (\ref{mem}) as
\begin{equation*}
\begin{split}
&\dot{v}_i = 
\sum_{j=1}^m \Big\{ \frac{1}{2} x_{l,j}x_{s,j} q_{ij} \min_{\{i'\neq i \,|\, q_{i'j} \neq 0 \}}(1-q_{i'j}v_{i'}) 
+ (1+\zeta x_{l,j})(1-x_{s,j})\delta_{i \sigma_j}q_{ij}C_j(\mbf{v}) \Big\},
\\
&\dot{x}_{s,j} = \beta \big( x_{s,j} + \epsilon \big) \big( C_j(\mbf{v}) - \gamma \big), \\
&\dot{x}_{l,j} = \alpha \big(C_j(\mbf{v}) - \delta).
\end{split}
\end{equation*}
In Section \ref{plane}, we argued that the memory flow field is separated into continuous regions by hyperplanes. The divergence cannot be defined at the hyperplanes as the field is discontinuous, so we restrict the divergence analysis to a domain where the field is in $C^1$. \\

The divergence of the flow field in the voltage space\footnote{The reason why we only care about the dissipativeness in the voltage space is because it is directly relevant to the convergence of the 3-SAT solution search. It is possible for the solver to be efficient even if the memory is non-dissipative. } is
\begin{equation*}
\begin{split}
\nabla_{\mbf{v}} \cdot F_{\mbf{v}} 
= & \nabla_{\mbf{v}} \Big( \mbf{G}(\mbf{v}) (\mbf{x_l} \ast \mbf{x_s}) \Big)
+ \nabla_{\mbf{v}} \Big( \mbf{G}(\mbf{v}) \big( (1+\zeta \mbf{x_s}) \ast (1-\mbf{x_l}) \big) \Big) \\
= & \sum_{ij} x_{l,j}x_{s,j} \partial_{v_i}G_{ij}(\mbf{v}) + \sum_{ij} (1+\zeta x_{l,j}) (1-x_{s,j}) \partial_{v_i} R_{ij}(\mbf{v}) \\
= & \sum_{ij} (1+\zeta x_{l,j})(1-x_{s,j})q_{ij} \partial_{v_i}\big( \delta_{i\sigma_j} C_j(\mbf{v}) \big),
\end{split}
\end{equation*}
where the divergence of the gradient-like term is zero because the element $G_{ij}$ never depends on $v_i$ (see Eq. (\ref{G})). The expression for the divergence of the rigidity term can be further simplified if we realize that
\begin{equation*}
\sum_{i} q_{ij} \delta_{i\sigma_j} \partial_{v_i} C_j(\mbf{v})
= - \frac{1}{2} \sum_{i} q_{ij}^2\delta_{i\sigma_j}
= - \frac{1}{2} q_{\sigma_j j}^2 = -\frac{1}{2}.
\end{equation*}
And the divergence expression reduces to
\begin{equation*}
\nabla_{\mbf{v}} \cdot F_{\mbf{v}} = -\frac{1}{2} \sum_j (1+\zeta x_{l,j})(1-x_{s,j}) < 0,
\end{equation*}
meaning that the voltage is {\it dissipative} at any point where the field is continuous. \\

It is easy to study how the addition of discontinuous hyperplanes affects the dissipativeness of the voltages. For the sake of simplicity, we focus on the following simple 2-SAT formula with only one clause
\begin{equation*}
(v_1 \lor v_2),
\end{equation*}
where, WLOG, the polarity can be assumed positive for both literals (see Section \ref{gauge}). The discontinuity of the voltage flow field clearly is in the line $v_1 = v_2$. Note that the gradient-like field
\begin{equation*}
\Big( \frac{1}{2}(1-v_2), \frac{1}{2}(1-v_1) \Big)
\end{equation*}
is continuous everywhere, while the rigidity field is not, which is given in the upper-left and lower-right regions as
\begin{equation*}
\frac{1}{2}(0, 1-v_2), \qquad \frac{1}{2}(1-v_1, 0),
\end{equation*}
respectively. The field points away from the line, meaning that the volume of a domain approaching this boundary will be expanded, with the expansion being greater the more unsatisfied the clause is. This is in fact a desired feature of the rigidity field as it attempts to expand the volume if the initial domain is in a frustrated region, which allows for a more thorough exploration of the voltage space. \\

In conclusion, the flow field in the voltage space is dissipative everywhere except at certain hyperplanes, where the domain may be expanded in a manner 
that facilitates the finding of the fixed points. 

\section{$O(n^{\alpha})$, $\alpha \leq 1$, Scaling with Problem Size}\label{TFT}
In this Section we employ results from the (supersymmetric) topological-field theory (TFT) of dynamical systems~\cite{Entropy} to prove that the {\it continuous-time} dynamics defined by Eqs.~(\ref{mem}) is such that the system reaches a fixed point/solution plane (solution of the 3-SAT) in a time that scales with the 3-SAT instance size, $n$, as $O(n^{\alpha})$, with $\alpha \leq 1$. WLOG by ``size'' we mean the number of variables, $n$, in the instance at a fixed clause-to-variable density, $\alpha_r=m/n$ (see section~\ref{prelim}). When the system has reached an equilibrium point then all clauses, $C_j(\mbf{v}_0)$, in the problem instance are strictly zero at the solution vector, $\mbf{v}_0$.\\ 

\begin{remark}
	Note that for a continuous-time dynamics the time a physical system requires to reach an equilibrium point is strictly infinite, irrespective of the size of the problem/phase space. In practice, as it is done in numerical simulations, we say that the system has found the solution to the problem when all clauses are less then a threshold whose value does not depend on the size of the instance. As shown in proposition~\ref{sign}, this threshold can be chosen to be as large as $1/2$, namely when $\mbf{C}(\mbf{v}) < \mbf{\frac{1}{2}}$, then $\sign(\mbf{v})$ is a solution vector. Therefore, when we discuss about the time to find the solution, we mean the shortest time for the dynamical system to cross a fixed threshold, in either the clauses or voltage variables. \\
\end{remark}

Using supersymmetric TFT it was shown that {\it instantons} are the {\it only} ``low-energy'' (collective) dynamics of digital memcomputing machines, as those described by Eqs.~(\ref{mem})~\cite{topo,DMtopo}. Instantons are families of classical trajectories in the phase space connecting critical points with given index (number of unstable directions) to critical points with a lower index (less number of unstable directions). The difference between indexes of the two critical 
points is typically 1, but it could be larger than 1. The reverse process (anti-instantons) connecting critical points 
of increasing index can only occur in the presence of noise and is ``gapped'', which means that even in the presence of noise it is exponentially suppressed compared to the instantonic process~\cite{Entropy}. \\

As shown in~\cite{topo,DMtopo}, the dynamics described by Eqs.~(\ref{mem}) then proceed via a succession of instantonic ``jumps'' that ``shed'' unstable directions in going from a critical point to the next. In addition, as proved in Section~\ref{no_period}, if the dynamics admit solutions, periodic orbits and chaos cannot co-exit.  

Since the instantonic trajectories are bounded (see Section \ref{compact}) the number of unstable directions is at most equal to the dimensionality of the phase space, $n+2m= n(1+2\alpha_r)$, and the latter grows linearly with problem size (at fixed density). 
Therefore, the total number of instantonic steps to reach equilibrium can only grow at most {\it linearly} with system size~\cite{DMtopo}. The fact that 
the number of steps could scale {\it sub-linearly} with the system size is because an instanton can connect critical points that differ by more 
than one unstable direction. We now want to translate this result into the actual physical time to reach a solution. \\

The time associated with each instanton (the instanton ``width'') is {\it independent} of the size of the instance and depends only on the parameters $\alpha$ and $\beta$ in Eqs.~(\ref{mem}), the rate of change of the $\mbf{x_s}$ and  $\mbf{x_l}$ variables, respectively. This can be seen by considering the path-integral 
form of the topological action associated to Eqs.~(\ref{mem}):
\begin{equation}\label{Q_pathint}
S= i\{  Q, \Psi\} \equiv i\{ Q, i\int_{0}^t d\tau \bar \chi(\tau)(\dot{\mbf{x}}(\tau)-F(\mbf{x}(\tau))) \},
\end{equation}
where the symbol $\{  Q, \Psi\}$ means (summation over repeated indexes is understood)
\begin{equation}
\{  Q, \Psi\} = \int_{0}^t d\tau \left(\chi^i(\tau)\frac{\delta }{\delta x^i(\tau)} + B_i(\tau)\frac{\delta }{\delta \bar\chi_i(\tau)}\right) \Psi,  
\end{equation}
with $\mbf{B}$ the vector of momenta conjugate to the bosonic variables $\mbf{x}$, and the vectors $\mbf{\chi}$ and $\bar{\mbf{\chi}}$ 
representing pairs of Faddeev-Popov ghosts and anti-ghosts, respectively (fermionic/Grassmann variables). \\

The Lagrangian of the system can be read from Eq.~(\ref{Q_pathint}). By taking the second derivative of this Lagrangian 
with respect to the memory variables we obtain the frequency of the instanton, and its inverse is its time width~\cite{Coleman}. Let us call 
this time $T_{inst,j}$ for each instanton $j$. To this time 
we need to add the time, $T_{cr,j}$, the system spends on the initial critical point (local supersymmetric vacuum) before each instantonic jump. This time is also 
independent of the size of the problem and depends only on the degree of memory in the system, which is again dictated by the parameters $\alpha$ and $\beta$~\cite{bearden2018}. Let us call $T_{max}=  \max_j(T_{cr,j}+ T_{inst,j})$ the maximum time required to do an instantonic jump in the phase space, including 
the time the system spends on initial critical points. Again, this time does not depend on the size of the instance (size of the phase space), only on the parameters $\alpha$ and $\beta$. By putting all this together the maximum {\it physical} time, $T_{phys}$, required by the system to reach the solution of a given 3-SAT problem of size $n$ and density $\alpha_r$ is then $T_{phys} \leq n(1+2r)T_{max}$. We have then proved

\begin{proposition}
	Given solvable 3-SAT instances of $n$ variables and fixed density $\alpha_r$. The dynamics described by Eqs.~(\ref{mem}) reach the solution 
	of these instances in a physical time $O(n^{\alpha})$, with $\alpha \leq 1$. 
\end{proposition}

\begin{remark}
	Although the physical time scales (sub-)linearly with problem size, its actual magnitude depends on $T_{max}$ (the slope of the growth with respect to $n$), which in turn depends on the rate of change of the $\mbf{x_s}$ and  $\mbf{x_l}$ variables. In addition, in the presence of physical noise, anti-instantons 
	appear in the dynamics. Since anti-instantons are gapped (exponentially suppressed) they may increase somewhat the degree of the polynomial, but cannot transform a polynomial 
	scalability into an exponential one.
	
	Note also that the (sub-)linear scalability obtained above does not necessarily apply to the numerical integration of Eqs.~(\ref{mem}). The reason is that time discretization transforms continuous dynamics to an effective discrete map. For discrete maps topological supersymmetry is broken explicitly, namely the evolution operator does not commute with the Noether charge of the symmetry and the above analysis does not apply as is~\cite{Entropy}. Of course, an efficient numerical method to integrate Eqs.~(\ref{mem}) may still be found as we have shown in the numerical results of the main text. However, this numerical method cannot be strictly $O(n^{\alpha})$, with $\alpha \leq 1$, because different integration schemes introduce different numerical noise. \\
\end{remark}

%

\end{document}